\documentclass[journal,final]{IEEEtran}
\usepackage[noadjust]{cite}
\usepackage{amsmath,amssymb,amsfonts}
\usepackage{algorithmic}
\usepackage{graphicx}
\usepackage{textcomp}
\usepackage{xcolor}
\usepackage{hyperref}
\usepackage{amsmath}
\usepackage{array}
\usepackage{dblfloatfix}
\usepackage{url}
\usepackage{tikz}
\usetikzlibrary{positioning}
\usepackage{cuted} 

\def\BibTeX{{\rm B\kern-.05em{\sc i\kern-.025em b}\kern-.08em
		T\kern-.1667em\lower.7ex\hbox{E}\kern-.125emX}}

\usepackage{enumerate}
\usepackage{array}
\usepackage[utf8]{inputenc}
\usepackage{amsmath}
\usepackage{amsthm, amssymb}
\usepackage{nccmath}
\usepackage{mathtools}
\usepackage{bbm}
\usepackage[makeroom]{cancel}
\usepackage{amssymb}
\graphicspath{{./figures/}}
\usepackage{fixltx2e}
\usepackage{soul}
\usepackage{xcolor}
\usepackage{ upgreek }
\usepackage{dblfloatfix}

\usepackage[colorinlistoftodos]{todonotes}
\usepackage{listings}
\usepackage{color}
\usepackage{cuted}

\usepackage{url}
\usepackage{hyperref}
\usepackage{amsmath}
\usepackage{esvect}
\usepackage{amsfonts} 
\usepackage{epstopdf}
\usepackage{epsf,epsfig,verbatim,array,multicol,multirow}
\usepackage{amsmath,mathtools,bm}
\DeclareMathOperator\sign{sgn}
\DeclareMathOperator*{\argmax}{arg\,max}
\DeclareMathOperator*{\argmin}{arg\,min}
\usepackage{psfrag,bm,xspace}
\usepackage{hhline}
\usepackage{physics}
\usepackage{bbm}
\usepackage[normalem]{ulem}
\usepackage{enumitem}
\usepackage{tikz}
\usetikzlibrary{positioning,bayesnet}
\usepackage{svg}
\usepackage{esvect}

\newtheorem{prop}{Proposition}
\newtheorem{dfn}{Definition}

\newtheorem{assumption}{Assumption}

\newtheorem{claim}{Claim}
\newcounter{relctr} 
\everydisplay\expandafter{\the\everydisplay\setcounter{relctr}{0}} 

\newcommand\labelrel[2]{%
  \begingroup
    \refstepcounter{relctr}%
    \stackrel{\textnormal{(\alph{relctr})}}{\mathstrut{#1}}%
    \originallabel{#2}%
  \endgroup
}
\AtBeginDocument{\let\originallabel\label} 

\allowdisplaybreaks 
 \global\long\def\11{\mathbbm{1}}
\allowdisplaybreaks 

\newcommand\numberthis{\addtocounter{equation}{1}\tag{\theequation}}
\newcommand{\ra}{\rightarrow}
\newcommand{\la}{\leftarrow}
\newcommand{\mcl}{\mathcal}
\newcommand{\mb}{\mathbb}
\newcommand{\mbf}{\mathbf}

\newcommand{\udl}{\underline}

\def \defeq{\overset{\Delta}{=}}

\def \bsl{\boldsymbol}
\def \wh{\widehat}
\def \wt{\widetilde}
\def \pr{\mb{P}}
\def \xlrarrow{\xleftrightarrow}
\def \l{\left}
\def \r{\right}

\usepackage{float,layouts}
\usepackage[justification=centering]{caption}
\usepackage[T1]{fontenc}

\usepackage{cite}
\usepackage{hyperref}
\usepackage{amsmath}
\usepackage[linesnumbered,ruled,vlined,algo2e]{algorithm2e}

\SetCommentSty{mycommfont}
\let\oldnl\nl
\newcommand{\nonl}{\renewcommand{\nl}{\let\nl\oldnl}}
\SetKwInput{KwInput}{Input}                
\SetKwInput{KwOutput}{Output}              
\SetKwInOut{Parameter}{Parameter}
\SetKwRepeat{Do}{do}{while}%

\usepackage{array}

\ifCLASSOPTIONcompsoc
  \usepackage[caption=false,font=normalsize,labelfont=sf,textfont=sf]{subfig}
\else
  \usepackage[caption=false,font=footnotesize]{subfig}
\fi
\usepackage{dblfloatfix}
\usepackage{url}

\usepackage{color-edits}
\addauthor[VS]{vs}{red}

\begin{document}
%
\title{
Backward and Forward Inference in Interacting Independent-Cascade Processes: A Scalable and Convergent Message-Passing Approach
}
%
%
%
%
\author{Nouman~Khan,~\IEEEmembership{Student Member,~IEEE,}
        Kangle Mu, ~\IEEEmembership{Student Member, ~IEEE,}
       Mehrdad~Moharrami, 
       and~Vijay~Subramanian,~\IEEEmembership{Senior Member,~IEEE}
}

\maketitle
\begin{abstract}	
We study the problems of estimating the past and future evolutions of two diffusion processes that spread concurrently on a network. Specifically, given a known network $G=(V, \vv{E})$ and a (possibly noisy) snapshot $\mcl{O}_n$ of its state taken at (a possibly unknown) time $W$, we wish to determine the posterior distributions of the initial state of the network and the infection times of its nodes. These distributions are useful in finding source nodes of epidemics and rumors---\textit{backward inference}---, and estimating the spread of a fixed set of source nodes---\textit{forward inference}.

To model the interaction between the two processes, we study an extension of the independent-cascade (IC) model where, when a node gets infected with either process, its susceptibility to the other one changes. First, we derive the exact joint probability of the initial state of the network and the observation-snapshot $\mcl{O}_n$. Then, using the machinery of factor-graphs, factor-graph transformations, and the generalized distributive-law, we derive a Belief-Propagation (BP) based algorithm that is scalable to large networks and can converge on graphs of arbitrary topology (at a likely expense in approximation accuracy).

\end{abstract}

\begin{IEEEkeywords}
Independent-Cascade, MAP Inference, Factor Graphs, Distributive Law, Message Passing, Belief Propagation, Diffusion Source Localization, Spread Estimation, Influence Maximization.
\end{IEEEkeywords}

%
\IEEEpeerreviewmaketitle

\section{Introduction}\label{sec:introduction}
An essential aspect of many real-world spreading-processes is that they seldom evolve completely on their own. Instead, they exhibit \textit{cross-process interactions}. In epidemiology, for instance, many a time, infection of an individual with one disease increases their susceptibility to other diseases. In social networks, almost invariably, an individual's preference for adopting an idea 
largely depends on whether they have already adopted some other idea. 
Motivated to incorporate such cross-process interactions in the study of network diffusions, we formulate a 
discrete-time stochastic model in which two diffusion processes spread concurrently on the same graph. Each node in the graph has its own (neighbor dependent) susceptibilities to the two processes and when it gets infected with one process, all of its susceptibilities to the other one change. The two diffusion processes spread over the graph and at some unknown random time, we get to observe the state of the graph. We then wish to determine the past and future of either process. The specific inference problems that we study (listed in Section \ref{sec:problem:inferenceprob}) are specially useful in diffusion source localization (DSL) \cite{ying2018diffusion}, and spread estimation which is a key step in most influence-maximization (IM) \cite{yuxin-im-2022} algorithms. 
\begin{figure}[!b]
\centering
\includegraphics[width=\linewidth]{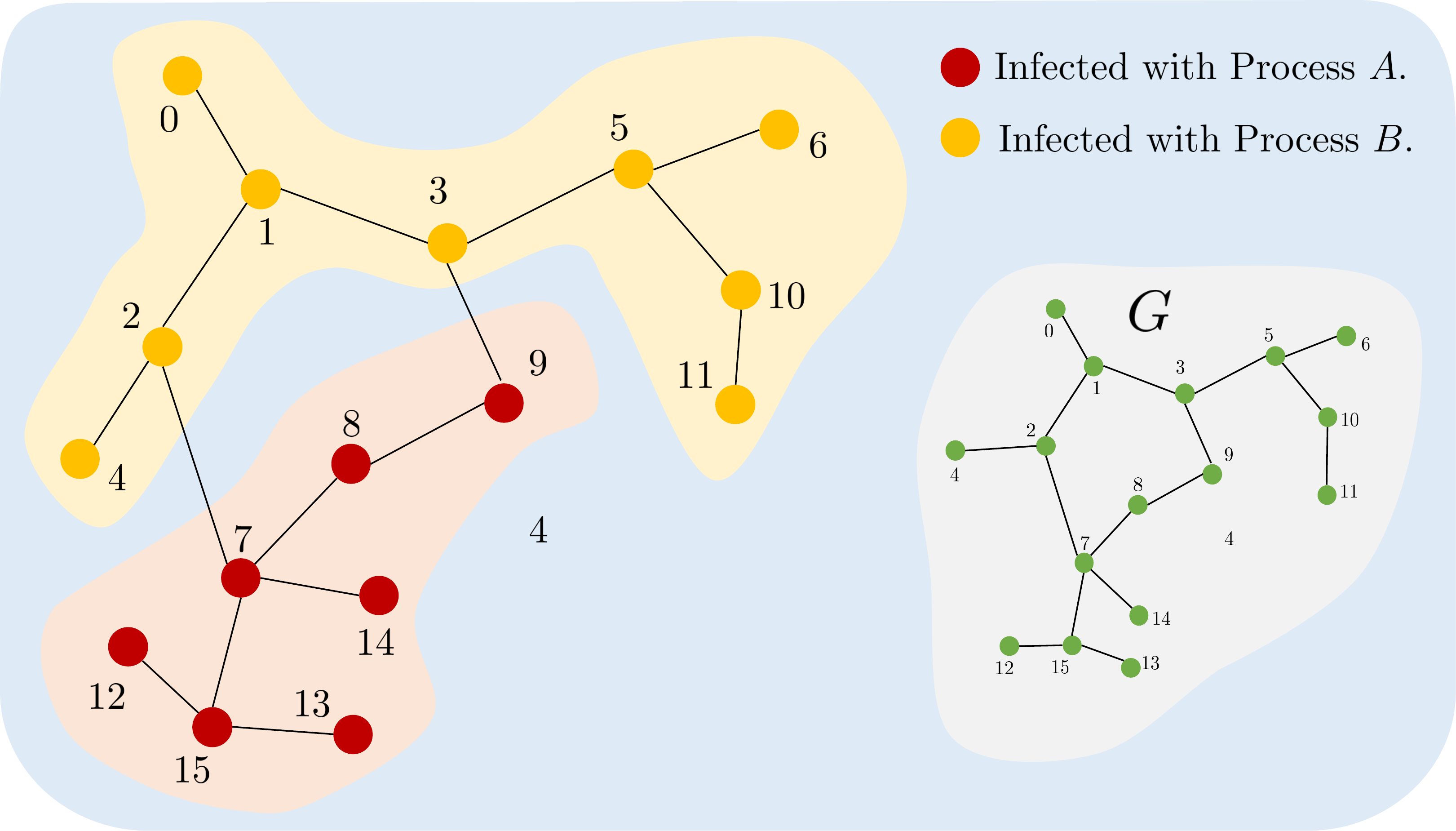}
\caption{Example showing spread of two conflicting rumors in a small gossip network.}
\label{fig:backwardinf_motivationa}
\end{figure}
 
\noindent \textbf{Cross-process Interactions in Backward and Forward Inference}: 
Figure \ref{fig:backwardinf_motivationa} shows a stylized example of gossip in a small group of individuals labelled by numbers 0 through 15. Individuals 1 and 8 
spread two conflicting rumors in the group. Whichever rumor reaches a node first, they believe that rumor to be true (deeming the other to be false). If both rumors reach a node at the same time, they do not believe in either. 
Both Short-Fat-Tree (SFT) \cite{zhu16-1} and Rumor-Centrality \cite{dshah11} algorithms (discussed in Section \ref{sec:relatedwork}), make incorrect predictions by estimating nodes 7 and 3 as origins of $A$ and $B$ respectively. Likewise, a Belief Propagation (BP) algorithm as that of \cite{altarelli14}, when used for each process, incorrectly designates nodes 7 and 3 as the sources also. Compared to these algorithms, the BP algorithm proposed in this work, which considers the cross-process interactions, results in the correct prediction of both sources (in under 11 iterations). 
We would like to capture such settings for large networks where multiple cross-process interactions could occur as a result of multiple sources of one or both processes. It is worth noting that both SFT and RC algorithms do not have a natural extension to study such settings.

Considering cross-process interactions may also be desirable in influence-maximization (IM) problems where one would like to find a fixed number of (source) nodes that can achieve the maximum spread for a process in the presence of another competing one. Such settings can arise in problems related to marketing campaigns, or in subroutines of general diffusion-control or diffusion-detection problems where the task is to the find the minimal set of nodes to effectively resist or detect an ongoing diffusion. Typically, in such algorithms, computationally heavy Markov-Chain-Monte-Carlo (MCMC) simulations are performed to estimate the spread of a node (or set of nodes) \cite{yuxin-im-2022}. This is feasible on small networks but computationally prohibitive on large networks, (specially in the setting of two or more processes). The key advantage of the proposed BP algorithm is that it uses much lesser wall-clock time to compute (approximately) the distributions of infection times which then directly yield the (approximate) spread of the source nodes.\footnote{There is no free lunch however, because BP is guaranteed to be exact on acyclic factor-graphs only. Thus, while MCMC methods trade computation time for approximate guarantees, the BP method trades approximate guarantees for improvements in computation time. It is worth mentioning that despite the absence of approximate guarantees, BP has been widely used on cyclic factor graphs and generally provides good results, (especially on sparsely connected factor-graphs)---for example in error-correcting codes \cite{ackerman-bp-coding-2011, guonai-bp-coding-2006}. The factor-graphs in these applications contain numerous cycles.}

\subsection{Contribution}
The contributions of this work are as follows.
\begin{enumerate}
\item[a)] We formulate a Bayesian framework that studies two diffusion processes spreading concurrently on a network. Despite its simplicity, the studied framework captures key characteristics of real-world spreading processes such as asymmetric interactions, presence of multiple-sources, noisy observation of the network, and unknown observation-time. (Section \ref{sec:problem}).

\item[b)] We propose a BP based algorithm that can be used to solve both the DSL and spread estimation problems. The proposed algorithm is scalable to large networks and can converge on graphs of arbitrary topology (at a likely expense in approximation accuracy). (Sections \ref{sec:map}, \ref{sec:inference_1}, and \ref{sec:inference_2}).

\item[c)] We extend the work of \cite{mooij07} (from BP on standard factor-graphs) to discounted BP on inflated factor-graphs in order to ensure convergence of the proposed algorithm. (Section \ref{sec:convergence}).

\end{enumerate}

\subsection{Related Work}\label{sec:relatedwork}
\subsubsection{Backward Inference}
There is a large body of work on backward inference in spreading processes that has mostly focused on diffusion source localization \cite{ying2018diffusion}---owing to potential applications in epidemic-trace studies, rumor-tracing in social networks, and virus-source identification in computer networks. Theoretically also, the DSL problem presents a wide array of interesting settings to explore, which a number of works have attempted to address. A (non-exhaustive) survey of recent works is available in \cite{jiang17}. Below, we present a few important works most relevant to the setting of this paper. Importantly, all the works mentioned below study DSL in the setting of a single diffusion process.

\begin{enumerate}[leftmargin=*]
\item[a)] \emph{Rumor-Centrality (RC)}: One of the first works that studied the DSL problem is arguably the seminal paper by Shah and Zaman \cite{dshah11}, where they considered a continuous-time \textit{Susceptible-Infected} (SI) model in which the spreading-times on all the edges are independent and identically-distributed exponential random variables. 
The authors proposed an estimator in which each infected node is assigned a graph-based quantity called \textit{rumor-centrality} and the node with the highest score (ties broken uniformly at random) is declared as the source-estimate. The authors showed that the rumor-centrality based estimator is the maximum likelihood estimator (MLE) on regular-trees and also analyzed its performance on line-graphs and geometric-trees. For a general network, the authors proposed a heuristic in which the score assigned to each infected node is the rumor-centrality computed on the BFS tree originating from that node. Later in a follow-up in \cite{dshah12}, the theoretical results were also extended to generic random trees and more general spreading distributions.

\item[b)] \emph{Short-Fat-Tree (SFT) Algorithm}: Following \cite{dshah11}, several works appeared which attempt to cover various settings of the DSL problem. Zhu et al. \cite{zhu16-1} considered the Independent-Cascade (IC) model \cite{kempe15} along with the assumption of a uniformly-chosen unique source and a geometrically distributed observation time. The authors proposed the Short-Fat-Tree (SFT) algorithm which selects the node with the highest weighted-boundary-node-degree (WBND). Loosely speaking, WBND of a node is a measure that is directly proportional to the smallness and fatness of the breadth-first search (BFS) tree starting from it (reason behind the algorithm's name). The authors showed that the SFT algorithm is the ML-estimator when the underlying network is a tree, and they have also established theoretical guarantees of the algorithm for the class of Erdos-Renyi random-graphs.

\item[c)] \emph{Marginal Inference Based On Message-Passing}: A straightforward approach to the DSL problem, one that can use a minimal set of assumptions, would be to perform (an exact or approximate) MAP inference on the underlying probability-measure. However, since the problem is NP-hard, message-passing schemes based on localized computations serve as the only practical resort. In such schemes, messages are passed along the edges of a convenient graphical structure that encapsulates the locality inherent in the probability-measure by virtue of its factorized form.

The belief-propagation (BP) algorithm\footnote{Also known as the sum-product algorithm.}, is one such message-passing scheme. For trees/forests (graphs without cycles), it is guaranteed to converge after a finite number of iterations and yields the exact marginal distributions, whereas for graphs with cycles, it has no convergence guarantees (let alone of a reasonable approximation). Nonetheless, the BP algorithm is widely used in practice and has shown superior performance in a wide range of applications, most popularly in error-correcting decoding algorithms \cite{frey97}. In the context of DSL, Lokhov et al. \cite{lokhov14} were the first to propose a computationally-feasible message-passing algorithm, namely the \emph{dynamic-message-passing} (DMP) algorithm. The algorithm is based on a mean-field approximation of the underlying probability-measure. Later, assuming that the observation-time\footnote{The time at which snapshot of the network is taken.} is known, Altarelli et al. \cite{altarelli14} improved the DMP algorithm by giving a computationally-feasible BP algorithm on the factor-graph representation of the probability-measure. As such, they offer a message-passing algorithm that produces exact marginals whenever the network is a tree.

\item[d)] \emph{Other Works}: Besides rumor-centrality, short-fat-tree, and message-passing schemes, there are several other works worth mentioning. For instance, some treatment of the cases when only sparse or noisy observations are available can be found in  \cite{zhu13,zhu17,zhu16_2,luo13,karamchandani13}. The work \cite{luo13} gives some treatment of the case of multiple sources under the SI model. Work \cite{wang14} studies the case where one gets to observe infection-subgraphs corresponding to multiple rumor spreading instances starting from the same source and then demonstrates significant improvement in the probability of detection (versus the case when only a single snapshot is available). Recently, \cite{kumar17} has proposed a Markov-Chain-Monte-Carlo (MCMC) based estimator that is agnostic to the underlying diffusion model and can also make use of some extra partial-observation structure (in addition to the snapshot of the diffusion-process).
\end{enumerate}

\subsubsection{Forward Inference}
Work on forward inference in spreading processes presents itself largely in the study of influence-maximization (IM) problems; for a non-exhaustive survey, see \cite{yuxin-im-2022}. Here, the goal is to find a small set of seed/source nodes in a network that maximize the spread of influence (under a certain diffusion model). The key step in such problems is to estimate the degree of spread of a fixed node (or set of nodes) which is mostly done through computation-costly Monte-Carlo simulations. Some good references on the IM problem include \cite{kempe-2003,celf-leskovec-2007, celf-goyal-2011, cim-chen-2014}.

Like the DSL problem, the IM problem has been studied in the setting of a single spreading process.

\subsection{Notation}
The key notations in this paper are as follows.
\begin{itemize}
\item Probability and expectation operators are denoted by $\pr$ and $\mb{E}$, respectively. 
\item Random variables are denoted by upper-case letters and their realizations by the corresponding lower-case letters.  At times, we also use the shorthands  $\mb{E}\l[ \cdot | x \r] \defeq \mb{E}\l[ \cdot | X = x \r]$ and $\pr\l( y | x \r) \defeq \pr\l( Y = y | X = x \r)$ for conditional quantities. 
\item Whenever a conditional probability or conditional expectation is written, it is implicitly assumed that the conditioning event has non-zero probability.
\item To aid the reading experience, we strongly encourage the reader to refer to the relevant notation and definitions listed in Appendix \ref{sec:appendix:notation}.
\item For better readability, all appendices excluding Appendix \ref{sec:appendix:notation} are in a single-column format.
\end{itemize}

\subsection{Organization}
The remainder of the paper is organized as follows. Backward and forward inference problems are formulated in a Bayesian setting in Section \ref{sec:problem}. The exact joint probability of the network's initial-state and the observation-snapshot is derived in Section \ref{sec:map}. Work on developing a message-passing algorithm with feasible memory and per-iteration-time complexities is done in Sections \ref{sec:inference_1} and \ref{sec:inference_2} (with details deferred to Appendices \ref{app:getridofd}, \ref{app:messageupdaterules-1}, and \ref{app:messageupdaterules2}). Convergence of message-passing (via discounting of incoming messages) is ensured for graphs of arbitrary topology in Section \ref{sec:convergence}. A discussion of different settings and cases that are covered by the studied model and our message-passing algorithm is provided in \ref{sec:discussion}. Finally, concluding remarks are given in Section \ref{sec:conclusion}.
\section{Problem Formulation}\label{sec:problem}
In this section, we present a stochastic model for concurrent diffusion of two processes on a network. Then, we formalize the backward and forward inference problems that we will attempt to solve through our proposed message-passing algorithms.

\subsection{Network Model}
Consider a static network represented by a directed graph $G=(V, \vv{E})$ where $V$ is the set of nodes, and $\vv{E}$ is the set of directed edges. Typical elements of $V$ and $\vv{E}$ are respectively denoted by $i$ and $(i, j)$. For every pair of connected nodes, it is assumed that two directed edges with opposite directions connect them, i.e., $(i,j)\in \vv{E} \implies (j,i)\in\vv{E}$.\footnote{This symmetry assumption is for ease of exposition and (as we shall see) is without loss of generality.} The set of undirected edges induced by $\vv{E}$ is denoted by $E$ and its typical element is denoted by $\{i,j\}$ or simply $ij$ {with the convention that $i<j$}. The set of all neighbors of a node $i\in V$ is denoted by  $\partial{i}$, i.e., $\partial{i} \defeq \{k\in V: \{k,i\}\in E\}$. 

The set of all connected-components of $G$ is denoted by $\mcl{C}$ and the sub-graph of a connected-component $C\in\mcl{C}$ is denoted by $G_C = (V_C, \vv{E}_C)$ where $V_C$ and $\vv{E}_C$ are respectively the sets of nodes and directed edges in $C$. The set of undirected edges induced by $\vv{E}_C$ is denoted by $E_C$. 

We assume (WLOG) that each connected-component has at least two nodes.

\subsection{Diffusion Model}\label{sec:problem:diffusion_model}
There are many plausible ways in which one can extend the commonly studied 
diffusion models (independent-cascade (IC), SI, SIR, etc.) to study multiple interacting spreading processes. To avoid deviating from the main goal, we consider a reasonably simple time-slotted model that consists of two diffusion processes $A$ and $B$ that unfold and interact with each other over the network $G$. 
\subsubsection{State Description}
For a node $i \in V$, its state at time $t\in \mb{Z}_{\ge 0} \defeq \{0,1,\dots\}$ is denoted by $X_i^{(t)}$ and is defined as the set of processes it has been infected with up to (the end of) time $t$. The state of the network is then given by the vector 
\begin{align*}
	\mbf{X}^{(t)} \defeq \l( X_i^{(t)} : i\in V \r).\numberthis\label{eq:networkstate}
\end{align*}
In our two-spread model, node $i$ can have four possible states, namely susceptible ($X_i^{(t)}=\emptyset$), infected with process $A$ ($X_i^{(t)}=\{A\}$), infected with process $ B$ ($X_i^{(t)}=\{B\}$), and infected with both processes $A$ and $B$ ($X_i^{(t)}=\{A,B\}$). We denote the set of all possible states of a node by $\mcl{X} \defeq \{ \emptyset, \{A\}, \{B\}, \{A, B\} \}$. 

\subsubsection{State Evolution}
The dynamics of the diffusion-process is according to the IC model \cite{kempe15}. A node $i\in V$, after getting infected with process $I\in\{A,B\}$, attempts to transfer it to its neighbors which are yet susceptible to $I$, i.e., $j\in\partial{i}$ with state 
$\mcl{J} \in \{ \emptyset,  \{A,B\} \setminus \{I\} \}$. This infection attempt is made only once and if successful, the corresponding node contracts process $I$ in the next time-step. The duration of this transfer is either 1 or $\infty$, depends on $\mcl{J}$, and is denoted by the activation-variable $D_{i\ra j}^{I\ra \mcl{J}}$. According to the IC model, each $D_{i\ra j}^{I\ra \mcl{J}}$ has an independent distribution on $\{1,\infty\}$ and is given by the kernel,
\begin{align*}
f_{D_{i\ra j}^{I\ra \mcl{J}} } \l(d_{i\ra j}^{I\ra \mcl{J}}\r) 
&\defeq
\begin{cases}
\lambda_{i\ra j}^{I\ra \mcl{J}} &d_{i\ra j}^{I\ra \mcl{J}}=1,\\
1-\lambda_{i\ra j}^{I\ra \mcl{J}} &d_{i\ra j}^{I\ra \mcl{J}}=\infty.
\end{cases}\numberthis\label{eq:fdijij},
\end{align*}
where $\lambda_{i\ra j}^{I\ra \mcl{J}}$ is the probability that node $i$ infects node $j$ with process $I$ when $j$ is in state $\mcl{J}$. We denote the vector of all transmission probabilities by $\bsl{\lambda}$, i.e., $\bsl{\lambda} \defeq \l( \bsl{\lambda}_{i\ra j}: (i,j)\in \vv{E} \r)$, and
\begin{align*}
\bsl{\lambda}_{i\ra j}&\defeq \l\{ 
\lambda_{i\ra j}^{A \ra \emptyset}, \lambda_{i\ra j}^{A \ra \{B\}},
\lambda_{i\ra j}^{B \ra \emptyset}, \lambda_{i\ra j}^{B \ra \{A\}}
\r\}.
\end{align*}
\textbf{Remark}: Each directed edge $(i,j)\in \vv{E}$ has a corresponding set of parameters, which we have denoted by $\bsl{\lambda}_{i\ra j}$. By our symmetry assumption, the undirected edge $\{i,j\}$ will have two sets of parameters, namely $\bsl{\lambda}_{i\ra j}$ and $\bsl{\lambda}_{j\ra i}$.

\subsection{Observation Model}
It is assumed that the network $G$ is known and a noisy snapshot of the network's state at some unknown \emph{observation-time} $W \in \mb{Z}_{\ge 0}$ is available. The availability of this (noisy) snapshot is denoted by the \textit{observation-event} $\mcl{O}_{n} = \{ \wt{\mbf{X}}^{(W)} = \wt{\mbf{x}}\}$, where $\wt{\mbf{X}}^{(W)}$ is the noisy observation of the true state $\mbf{X}^{(W)}$. We assume that the conditional distribution of $\wt{\mbf{X}}^{(W)}$ given $\mbf{X}^{(W)}$ is given by the below factorized observation kernel,
\begin{align*}
f_{\mbf{\wt{X}} | \mbf{X}} \l(\wt{\mbf{x}} | \mbf{x}\r) \defeq  \prod_{i\in V} f_{\wt{X}_i|X_i} (\wt{x}_i|x_i) \qquad \wt{x}_i, x_i \in \mcl{X}. \numberthis\label{eq:fo}
\end{align*}
Here, $f_{\wt{X}_i|X_i} \l( \wt{x}_i|x_i \r)$ is the probability of observing node $i$ in state $\wt{x}_i$ given its (true) state is $x_i$. For simplicity, we assume that $W$ has an independent truncated geometric distribution given by the kernel,
\begin{align*}
f_W (w) &\defeq \frac{ \alpha(1-\alpha)^{w} }{(1-\alpha)^{w_{min}}-(1-\alpha)^{w_{max}+1} },\\
&\hspace{40pt} w\in 
\mcl{W} \defeq 
\{ w_{min},w_{min}+1,\dots,w_{max}\}, \numberthis\label{eq:fw}
\end{align*}
with $\alpha\in(0,1)$ and $0\le w_{min}\le w_{max} \le \infty$.

\subsection{Prior Belief on Network's Initial-State}
We assume that the initial-state of the network is supported on a candidate-set $\bm{\mcl{X}}^{(0)} \subseteq \mcl{X}^{|V|} $ and has an independent distribution given by the kernel, 
\begin{align*}
f_{\mbf{X}^{(0)}}(\mbf{x}^{(0)}) >0 \text{ iff } \mbf{x}^{(0)}\in \bm{\mcl{X}}^{(0)}.\numberthis\label{eq:fx0}
\end{align*}

\subsection{Inference Problems}\label{sec:problem:inferenceprob}
Let $\l(\Omega, \mcl{F}, \pr: \mcl{F}\ra [0,1]\r)$ be the probability-space (induced by $f_W, f_{\mbf{X}^{(0)}}, \bsl{\lambda}, f_{\mbf{\wt{X}} | \mbf{X}} $) for the random process $\l\{\mbf{X}^{(t)}: t\ge 0\r\}$. We consider three distinct inference tasks.
\begin{enumerate}
\item[a)] \textit{Posterior Distribution of Network's Initial-State}:
The primary goal of DSL is to find the most likely initial-state given the observation-event $\mcl{O}_n$.
That is, we wish to find some element
\begin{align*}
^\dagger{\mbf{x}}^{(0)}\in\argmax_{\mbf{x}^{(0)}\in \bm{\mcl{X}}^{(0)}} \pr\l(\mbf{x}^{(0)}|\mcl{O}_n\r).\numberthis\label{eq:argmaxpx0}
\end{align*}
\item[b)] \textit{Posterior Distribution of a Node's Initial-State}: One may also be interested in estimating the posterior distributions of each node's initial-state, i.e., 
\begin{align*}
\l\{ \pr\l(x_i^{(0)}|\mcl{O}_n\r) : x_i^{(0)} \in \mcl{X} \r\}, i\in V
.\numberthis\label{eq:argmaxpxi0}
\end{align*}
These probabilities can be used in ranking infected-nodes according to their likelihood of being the source of either process.
\item[c)] \textit{Posterior Distributions of Infection-times}: Another important inference problem is finding out the individual posterior distributions of the infection-times of all nodes, i.e.,
\begin{align*}
\l\{ \pr\l( t_i^I | \mcl{O}_n \r), I \in \l\{A, B\r\} \r\}, i\in V
,\numberthis\label{eq:argmaxptii}
\end{align*}
where $t_i^I$ is a shorthand for $\inf \{t \ge 0: x_i^{(t)} \ni I \}$. Accurate estimation of above probabilities is useful for the spread-estimation task which is a crucial step of most influence-maximization algorithms. See Section \ref{sec:disc:influencemaximization}.
\end{enumerate}
\section{MAP Probability of Initial State}\label{sec:map}
In this section, we cast problem \eqref{eq:argmaxpx0} as equivalent to the marginalization of a factorized global-function -- one which has local-functions of ``manageable'' complexity.\footnote{For a detailed review of this class of problems, we refer the reader to the excellent expositions by Aji and McEliece \cite{aji00-gdl}, and Kschischang, Frey, and Loeliger \cite{kschischang01}.}. The solutions to problems \eqref{eq:argmaxpxi0} and \eqref{eq:argmaxptii} will be obtained as auxiliary outputs of our message-passing algorithm for \eqref{eq:argmaxpx0}.

\subsection{Introducing Basic Random Variables}
For the IC-based diffusion model described in Section \ref{sec:problem:diffusion_model}, the random variables that fully specify all the realizations of the process $\{\mbf{X}^{(t)}: t\ge 0\}$ and the observation-events are \textit{i)} the initial-state of the network $\mbf{X}^{(0)}$, \textit{ii)} the observation-time $W$, \textit{iii)} the observed state $\wt{\mbf{X}}^{(W)}$, and \textit{iv)} the activations vector $\mbf{D} \defeq \l( \mbf{D}_{i\ra j} :(i,j)\in \vv{E} \r)$, where 
\begin{align*}
\mbf{D}_{i\ra j}&\defeq\l\{ 
D_{i\ra j}^{A \ra \emptyset}, D_{i\ra j}^{A \ra \{B\}},
D_{i\ra j}^{B \ra \emptyset}, D_{i\ra j}^{B \ra \{A\}}
\r\}.
\end{align*}
We call $\mbf{X}^{(0)}$, $W$, $\wt{\mbf{X}}^{(W)}$, and $\mbf{D}$ as \textit{basic random variables}. It is clear that $\mbf{X}^{(0)}$, $W$, and $\mbf{D}$ are mutually independent, i.e.,
\begin{align*}
\pr\l(\mbf{x}^{(0)},w,\mbf{d}\r) &= \pr\l(\mbf{x}^{(0)}\r)\pr\l(w\r)\pr\l(\mbf{d}\r),\label{eq:px0wd}\numberthis
\end{align*}
with their marginals given by
\begin{align*}
\pr\l(\mbf{x}^{(0)}\r) &= f_{\mbf{X}^{(0)}}\l(\mbf{x}^{(0)}\r),\numberthis\label{eq:px0}\\
\pr\l(w\r) &= f_{W}\l(w\r),\numberthis\label{eq:pw}\\
\pr\l(\mbf{d}\r) &= \prod_{(i,j)\in \vv{E}} f_{\mbf{D}_{i\ra j}} \l(\mbf{d}_{i\ra j}\r).\numberthis\label{eq:pd}
\end{align*}
In \eqref{eq:pd}, we have introduced the shorthand,
\begin{align*}
f_{\mbf{D}_{i\ra j}} \l( \mbf{d}_{i\ra j} \r) \defeq \prod_{D_{i\ra j}^{I\ra \mcl{J}} \in \mbf{D}_{i\ra j} } f_{D_{i\ra j}^{I \ra \mcl{J}}} \l( d_{i\ra j}^{I \ra \mcl{J} } \r).
\numberthis\label{eq:fdij}
\end{align*}
\subsection{Posterior Probability of Network's Initial-State}
\subsubsection{Infection-Times Vector and Propagation-Times}
To derive the posterior probability of $\mbf{X}^{(0)}$
, it suffices to derive the joint probability $\pr\l(\mbf{x}^{(0)}, \mcl{O}_n\r)$. To that end, we define a few \emph{auxiliary random variables}. For a node $i\in V$, we define its infection-time variables, $\mbf{T}_i \defeq \{T_i^A, T_i^B\}$, as follows.
\begin{align}
\begin{split}
T_i^{A}&\defeq \inf\l\{t:X_i^{(t)} \ni A \r\}, \text{ and}\label{eq:tii}\\
T_i^{B}&\defeq \inf\l\{t:X_i^{(t)} \ni B \r\}.
\end{split}
\end{align}
Note that $T_i^I$ is the time at which node $i$ gets infected with process $I$, the support of $T_i^I$ is contained in the set $\mcl{T} \defeq \{0,1,\dots, |V|-1,\infty\}$\footnote{$|V|-1$ is the worst-case upper-bound on the longest path in $G$.}, and that \eqref{eq:tii} is consistent with the standard convention, $\inf\emptyset=\infty$. 

Denoting the vector of all infection-times by $\mbf{T} = \l(\mbf{T}_i:i\in V\r)$, we can express the process $\{\mbf{X}^{(t)} : t \ge 0\}$ in terms of $\mbf{T}$, if for each node $i\in V$, we set
\begin{align*}
X_i^{(t)}&=\begin{cases}
\emptyset &\text{ $ t<\min\{T_i^A,T_i^B\} $ },\\
\argmin\l\{T_i^A,T_i^B\r\} &\begin{array}{@{}l@{}} 
\min\l\{T_i^A,T_i^B\r\}\le t\\ <\max\l\{T_i^A,T_i^B\r\}, \end{array} \\
\l\{A,B\r\} &\text{ $t\ge \max\l\{T_i^A,T_i^B\r\}$}.
\end{cases}\numberthis\label{eq:xit}
\end{align*}
From \eqref{eq:xit}, it follows that the conditional probability of the event $\{ \mbf{X}^{(W)} = \mbf{x} \}$ given $W$ and $\mbf{T}$ is a consistency-check kernel, i.e.,
\begin{align*}
\pr\l( \{ \mbf{X}^{(W)} = \mbf{x} \} | w, \mbf{t} \r) &= \prod_{i \in V} \gamma_i^p \l( w, \mbf{t}_i; x_i \r), 
\end{align*}
where,
\begin{align}
\begin{split}\label{eq:gammaip}
&\gamma_i^p \l( w, \mbf{t}_i; {x}_i \r) \hspace{0pt} \defeq \11\l[{x}_{i}=\emptyset, w<\min\l\{t_i^A,t_i^B\r\}\r]\\
&\hspace{10pt}+\11\l[{x}_{i}=\argmin\l\{t_i^A,t_i^B\r\} ,\r. \\
&\hspace{40pt} \l. \min\l\{t_i^A,t_i^B\r\} \le w < 
\max\l\{t_i^A,t_i^B\r\} \r]\\
&\hspace{10pt}+ \11\l[{x}_{i}=\l\{A,B\r\},  w\ge \max
\l\{t_i^A,t_i^B\r\}\r].
\end{split}
\end{align}
\textbf{Remark}: The superscript $p$ in $\gamma_i^p$ indicates that this would be the conditional probability of $\wt{X}_i^{(W)}$ if the observation-snapshot were \emph{perfect}/noiseless.

\subsubsection{Conditional Probability of Infection-Times Vector}
In light of \eqref{eq:xit}, we can focus on the conditional probability of $\mbf{T}$ given $\mbf{X}^{(0)}$, $W$, and $\mbf{D}$. Indeed, $\mbf{T}$ can be determined uniquely once $\mbf{X}^{(0)}$ and $\mbf{D}$ are known. Specifically, for a node $i\in V$, two fundamental quantities for determining its infection times are the \textit{propagation-times} of processes $A$ and $B$ through each of its neighbors. We denote these auxiliary random variables by $\mbf{T}_{k\ra i} \defeq \{ T_{k\ra i}^A, T_{k\ra i}^B \}$ and define them as follows.
\begin{align*}
T_{k\ra i}^I &\defeq T_k^I + 
D_{k\ra i}^{I\ra \emptyset}\11\l[T_k^I<T_i^J\r] +
D_{k\ra i}^{I\ra \{J\}}\11\l[T_k^I\ge T_i^J\r] \\
&\hspace{30pt} \quad \l(I\in\{A,B\}, J\in \{A,B\}\setminus \{I\}\r).\numberthis\label{eq:Tkii}
\end{align*}
We read $T_{k\ra i}^I$ as \textit{propagation-time of process $I$ from $k$ to (its neighbor) $i$}. Intuitively, it represents a ``candidate'' infection-time for node $i$ and the indicator terms in its definition help capture the correct activation-variable which is either $D_{k\ra i}^{I\ra \emptyset}$ or $D_{k\ra i}^{I\ra J}$. We also note that \eqref{eq:Tkii} is consistent with the convention $0\cdot\infty = 0$. 

We can now write,
\begin{align}
\begin{split}\label{eq:tii2}
T_i^A &= \11\l[X_i^{(0)}\not\ni A\r]
\min_{k\in\partial{i}}\l\{ T_{k\ra i}^A \r\}, \text{ and}\\
T_i^B &= \11\l[X_i^{(0)}\not\ni B\r] 
\min_{k\in\partial{i}} \l\{ 
T_{k\ra i}^B
\r\}.
\end{split}
\end{align}
Based on \eqref{eq:Tkii} and \eqref{eq:tii2}, the conditional probability of $\mbf{T}$ given $\mbf{X}^{(0)}$, $W$, and $\mbf{D}$, is a product of consistency-check kernels, i.e.,
\begin{align*}
\pr\l(\mbf{t} | \mbf{x}^{(0)}, \mbf{w}, \mbf{d} \r) &=
\prod_{i\in V} f_{\mbf{T}_i}\l( \mbf{t}_i, x_i^{(0)}, \l\{\mbf{t}_k, \mbf{d}_{k\ra i} \r\}_{k\in\partial{i}}  \r),\numberthis\label{eq:pt}
\end{align*}
where,
\begin{align*}
&f_{ \mbf{T}_i}\l( \mbf{t}_i, x_i^{(0)}, \l\{\mbf{t}_k, \mbf{d}_{k\ra i} \r\}_{k\in\partial{i}} \r)\\
&\defeq \prod_{I\in\{A,B\}} \delta \l( t_i^I, \11\l[x_i^{(0)}\not\ni I\r]
\min_{k\in\partial{i}} \l\{t_{k\ra i}^I \r\} \r).
\numberthis\label{eq:fti}
\end{align*}
Here, in \eqref{eq:fti}, for brevity, we have introduced the shorthands,
\begin{align*}
t_{k\ra i}^I &\defeq 
t_k^I + d_{k\ra i}^{I\ra \emptyset}\11\l[t_k^I<t_i^J\r] 
+ d_{k\ra i}^{I\ra \{J\}}\11\l[t_k^I\ge t_i^J\r]\\
&\hspace{30pt} \quad \l(I\in\{A,B\}, J\in \{A,B\}\setminus \{I\}\r), \numberthis\label{eq:tkii}\\
\mbf{t}_{k\ra i} &\defeq \l\{ t_{k\ra i}^{A}, t_{k\ra i}^{B} \r\},\\
\delta(\cdot, *) &\defeq \11\l[\cdot=*\r].\numberthis\label{eq:delta}
\end{align*}

\subsubsection{Conditional Probability of Observation-Snapshot}
Given $\mbf{X}^{(0)}$, $W$, $\mbf{D}$, and $\mbf{T}$, the conditional probability of $\mcl{O}_{n}$ is given by
\begin{align*}
&\pr\l(\mcl{O}_{n}|\mbf{x}^{(0)},w, \mbf{d}, \mbf{t}\r) = \pr\l( \{ \wt{\mbf{X}}^{(W)}=\wt{\mbf{x}} \} | \mbf{x}^{(0)}, w, \mbf{d}, \mbf{t} \r)\\
&=\sum_{\mbf{x}} \pr\l( \{ \wt{\mbf{X}}^{(W)} = \wt{\mbf{x}} \} | \{ \mbf{X}^{(W)} = \mbf{x} \}, \mbf{x}^{(0)}, w, \mbf{d}, \mbf{t} \r) \\
&\hspace{30pt} \times \pr\l( \{ \mbf{X}^{(W)} = \mbf{x} \} | \mbf{x}^{(0)},w, \mbf{d}, \mbf{t} \r) \\
&=\sum_{\mbf{x}} 
\pr\l( \{ \wt{\mbf{X}}^{(W)} = \wt{\mbf{x}} \} | 
\{ \mbf{X}^{(W)} = \mbf{x} \} \r) 
\pr\l( \{ \mbf{X}^{(W)} = \mbf{x} \} | w, \mbf{t} \r) \\
&=\sum_{\mbf{x}} \prod_{i\in V} f_{\wt{X}_i|X_i} \l(\wt{x}_i|x_i\r) \prod_{i \in V} \gamma_i^p \l( w, \mbf{t}_i; x_i\r)  \\
&\stackrel{\text{DL}}{=} \prod_{i\in V} 
\underbrace{
\sum_{x_i} f_{\wt{X}_i|X_i} \l(\wt{x}_i|x_i\r) \gamma_i^p \l(w, \mbf{t}_i; x_i \r)}_{
\defeq \gamma_i \l( w, \mbf{t}_i; \wt{x}_i\r)
} \\
&
=\prod_{i\in V}\gamma_i \l(w, \mbf{t}_i; \widetilde{x}_i \r),\numberthis\label{eq:po}
\end{align*}
where $DL$ is a shorthand for application of distributive-law.

\subsubsection{Writing the Joint Probability}
With the above setup in place, applying the law of total probability (LTP) gives
\begin{align*}
&\pr\l(\mbf{x}^{(0)}, \mcl{O}_n\r) 
\hspace{0pt} \stackrel{\text{LTP}}{=} \sum_{w, \mbf{d}, \mbf{t}} 
\pr\l(\mbf{x}^{(0)}, \mcl{O}_n, w, \mbf{d}, \mbf{t}\r)\\
&=\sum_{w, \mbf{d}, \mbf{t}} \pr\l(\mbf{x}^{(0)}, w, \mbf{d}\r) \pr\l(\mcl{O}_n, \mbf{t} | \mbf{x}^{(0)}, w, \mbf{d}\r)\\
&=\sum_{w, \mbf{d}, \mbf{t}} \pr\l(w\r) \pr\l(\mbf{x}^{(0)}\r) \pr\l(\mbf{d}\r)  \pr\l(\mbf{t} | \mbf{x}^{(0)}, w, \mbf{d} \r) \\
&\hspace{10pt} \times \pr\l(\mcl{O}_n | \mbf{x}^{(0)}, \mbf{w}, \mbf{d}, \mbf{t} \r)\\
&=\sum_{w, \mbf{d}, \mbf{t}} f_{W}\l(w\r) f_{\mbf{X^{(0)}}}\l(\mbf{x}^{(0)}\r) \prod_{(i,j)\in E} f_{\mbf{D}_{i\ra j} } \l(\mbf{d}_{i\ra j} \r)\\
&\hspace{10pt} \times \prod_{i\in V} 
f_{ \mbf{T}_i} \l( \mbf{t}_i, x_i^{(0)}, \{\mbf{t}_{k\ra i}\}_{k\in\partial{i}} \r) \gamma_i \l(w, \mbf{t}_i; \wt{x}_i\r). \numberthis\label{eq:px0_1}
\end{align*}
In \eqref{eq:px0_1}, we have expressed problem \eqref{eq:argmaxpx0} as the \emph{out-marginalization}\footnote{By out-marginalization of variables, we mean that the summation is performed over all their realizations.} of variables $w, \mbf{d}$, $\mbf{t}$, from the product of the \textit{local-functions} $f_W, f_{\mbf{X}^{(0)}}, \l\{f_{\mbf{T}_i}, \gamma_i \r\}_{i\in V}, \l\{f_{\mbf{D}_{i\ra j}} \r\}_{(i,j)\in\vv{E} }$. 

\subsection{Within-Factor Optimization of High Complexity Local-Functions}\label{subsec:ftioptimization}
For notational simplicity, from hereon, we will skip writing the arguments of a function when they are clear from the context. 

\subsubsection{Reducing Local-Domain of $f_{\mbf{T}_i}$}
It is well-known that the per-iteration-time complexity of a BP algorithm is linear in the size of the largest \textit{local-domain} -- by local-domain, we mean the domain of one of the local-functions. In \eqref{eq:px0_1}, the (local) domain of $f_{\mbf{T}_i}$ is the joint-support of variables, $\mbf{T}_i, X_i^{(0)},\l\{\mbf{T}_k, \mbf{D}_{k\ra i} \r\}_{k\in\partial{i}}$. The support-sizes of $X_i^{(0)}$, $\mbf{D}_{i\ra j}$, and $\mbf{T}_i$, are 4, 16, and $\l(|V|+1\r)^2$ respectively.\footnote{The possible values of $t_i^I$ are $\mcl{T} = \{0,1,\dots, |V|-1,\infty\}$; $|V|-1$ is the worst-case upper-bound on the longest path in the graph $G$.} Therefore, a direct implementation\footnote{That does not explore the internal dependencies of the input variables of $f_{\mbf{T}_i}$.} of BP will have per-iteration-time complexity that is at least on the order of $|V|^{2 \max_{i\in V} \l\{ |\partial{i}|+1\r\} }$. This will make BP infeasible even for moderately-sized network instances. 

One way we could circumvent this issue is if $f_{\mbf{T}_i}$ has a sum-product form of reduced complexity. We refer to the process of using such sum-product forms in a given local-function as \textit{within-factor optimization}. Such optimizations can often play a significant role in complexity reduction, and fortunately so, $f_{\mbf{T}_i}$ has a nice sum-product form -- one where its collective dependence on variables $\{\mbf{t}_k, \mbf{d}_{k\ra i} \}_{k\in\partial{i}}$ is broken down into multiplicatively separable dependencies. This within-factor optimization of $f_{\mbf{T}_i}$ follows from the identity,
\begin{align*}
&\delta\l( a,\min_{i \in [n]}\l\{b_i \r\}\r) = \prod_{i\in[n] } \11\l[ \sigma\l( b_i; a \r)\ge 0\r] \\
&\hspace{100pt} -\prod_{ i\in[n] } \11\l[\sigma\l( b_i; a \r)=1 \r],\numberthis\label{eq:keyidea}
\end{align*}
where,
\begin{align*}
\sigma\l(b_i; a\r) \defeq \sign \l(b_i - a\r).\numberthis\label{eq:sigmadef}
\end{align*}
Identity \eqref{eq:keyidea} is consistent with the convention, $\sign\l(\infty-\infty\r)=0$.\footnote{$\sign\l(c\r) = \11[c>0] -\11[c<0]$.} 
Using \eqref{eq:keyidea}, we can write \eqref{eq:fti} as follows,
\begin{align*}
&f_{\mbf{T}_i} \l( \mbf{t}_i, x_i^{(0)}, \l\{ \mbf{t}_k, \mbf{d}_{k\ra i} \r\}_{ k\in\partial{i} }  \r)\\
&\hspace{10pt} = \zeta_{i} \l( x_i^{(0)}, \mbf{t}_i \r)
\psi_{\ra i} \l( 
\mbf{t}_i, \l\{ \sigma\l(t_{k\ra i}^I, t_i^I\r) \r\}_{k\in\partial{i}, I\in\{A,B\}}
\r).\numberthis\label{eq:fti2}
\end{align*}
Here, $t_{k\ra i}^I$ is the shorthand given by \eqref{eq:tkii}; the kernel $\zeta_i$ ensures that $x_i^{(0)}$ and $\mbf{t}_i$ are consistent, i.e., 
\begin{align*}
\zeta_i \l( x_i^{(0)}, \mbf{t}_i 
\r) &\defeq \11\l[x_i^{(0)}=\emptyset, t_i^A>0, t_i^B>0\r]\\
&\hspace{10pt}+\11\l[x_i^{(0)}=\{A\}, t_i^A=0, t_i^B>0\r]\\
&\hspace{10pt} +\11\l[x_i^{(0)}=\{B\}, t_i^B=0, t_i^A>0\r]\\
&\hspace{10pt}+\11\l[x_i^{(0)}=\{A,B\}, t_i^A=t_i^B=0 \r];\numberthis\label{eq:zetai}
\end{align*}
and the kernel $\psi_{\ra i}$ ensures that node-$i$'s infection-times of the two processes (if non-zero) are consistent with the propagation-times through its neighbors, i.e.,
\begin{align*}
&\psi_{\ra i} \l(\mbf{t}_i, \l\{ \mbf{s}_{k\ra i} \r\}_{k\in\partial{i}} \r) 
\defeq 
\prod_{I\in\{A,B\}} \psi_{\ra i}^I \l( 
t_i^I, \{ s_{k\ra i}^I  \}_{k\in\partial{i}}
\r) \numberthis\label{eq:psii} \\ 
&\psi_{\ra i}^{I} \l( 
t_i^I, \l\{ s_{k\ra i}^I \r\}_{k\in\partial{i}}
\r) \defeq 
\11[t_i^I = 0 ]
\\
&\hspace{10pt} + \11\l[t_i^I>0\r] \l(
\prod_{k\in\partial{i}} \11\l[s_{k\ra i}^I \ge 0 \r]  - \prod_{k\in\partial{i}} \11 \l[s_{k\ra i}^I=1 \r] \r).\numberthis\label{eq:psiii1}
\end{align*}
From \eqref{eq:fti2}, we note that $f_{\mbf{T}_i}$ depends on $ \mbf{t}_k, \mbf{d}_{k\ra i} $ only through $\sigma(t_{k\ra i}^A; t_i^A )$ and $\sigma(t_{k\ra i}^B; t_i^B )$. The function $\sigma(\cdot,*)$ has two useful properties: 
\begin{itemize}
\item
Despite having a large domain, its range is small; it can take at most three values, namely -1, 0, and 1.
\item
The inverse-image of each value in its range is easily represented in terms of its arguments -- one only needs to compare the value of $\cdot$ with $*$. 
\end{itemize}
We shall make use of these two properties to reduce the complexity associated with $f_{\mbf{T}_i}$. To this end, let us introduce the relative-timing variables, $\mbf{S} \defeq \l(\mbf{S}_{i\ra j} : (i,j)\in \vv{E} \r)$ where $\mbf{S}_{i\ra j} \defeq \{S_{i\ra j}^A, S_{i\ra j}^B\}$ and 
\begin{align}
\begin{split}\label{eq:siji}
S_{i\ra j}^A &\defeq \sigma \l( T_{i\ra j}^A; T_j^A \r), \text{ and} \\
S_{i\ra j}^B &\defeq \sigma \l( T_{i\ra j}^B; T_j^B \r).
\end{split}
\end{align}
Here, we note that the support of each $S_{i\ra j}^{I}$ is contained in the set $\mcl{S}_{i\ra j}^{I} \defeq \{0,1\}$. This is because the event $\l\{S_{i\ra j}^I=-1\r\}$ is a zero-measure event (propagation of a process towards a node cannot happen earlier than the node's infection-time for that process; see \eqref{eq:tii2}). This simple observation helps simplify \eqref{eq:psiii1} to
\begin{align*}
&\psi_{\ra i}^{I} \l( t_i^I, \l\{s_{k\ra i}^{I} \r\}_{k\in\partial{i}} \r) \\
&= \11[t_i^I = 0]
+ \11\l[t_i^I>0\r] \l( 1 - \prod_{k\in\partial{i}} \11 \l[ s_{k\ra i}^I=1 \r] \r).\numberthis\label{eq:psiii2}
\end{align*}
Intuitively, $S_{i\ra j}^{I}=0$ means that node $i$ infected node $j$ with process $I$, and $S_{i\ra j}^{I}=1$ means otherwise. \footnote{Given that time is slotted, it is possible that more than one neighbor of node $i$ infects it.}

With the introduction of $\mbf{S}$, we can reduce the complexity of our eventual message-passing scheme by invoking the law of total probability in \eqref{eq:px0_1} for all realizations of $\mbf{S}$ (this helps us change $f_{\mbf{T}_i}$ into $\zeta_i\psi_{\ra i}(\mbf{t}_i, \{ \mbf{s}_{k\ra i}\}_{k\in\partial{i}} )$) and exploiting its deterministic dependence on the propagation-time and infection-time variables (this makes the extra overhead due to additional sum over all realizations of $\mbf{S}$ insignificant compared to the run-time gains obtained from removing all propagation-time variables). Specifically, based on \eqref{eq:fti2}, we can rewrite \eqref{eq:px0_1} as
\begin{align*}
&\pr \l( \mbf{x}^{(0)}|\mcl{O}_n \r) \\
&\hspace{0pt} \stackrel{\text{LTP}}{\propto}
\sum_{w,\mbf{t}, \mbf{d}, \mbf{s} }
f_{W}(w)f_{\mbf{X}^{(0)}}\l(\mbf{x^{(0)}}\r) 
\prod_{(i,j)\in\vv{E}} 
f_{\mbf{D}_{i\ra j}} \l(\mbf{d}_{i\ra j}\r) \\
&\hspace{0pt} \times \prod_{i\in V}\zeta_i\l(x_i^{(0)}, \mbf{t}_i \r) \gamma_i \l(w,\mbf{t}_i;\wt{x}_i\r) \psi_{\ra i} \l( \mbf{t}_i, \l\{ \mbf{s}_{k\ra i} \r\}_{k\in\partial{i}}  \r)  \\ 
&\hspace{0pt} \times \prod_{ (i,j)\in\vv{E} } 
\11\l[ s_{i\ra j}^{A}=\sigma\l(t_{i\ra j}^A; t_j^A \r) \r]
\11\l[ s_{i\ra j}^{B}=\sigma\l(t_{i\ra j}^B; t_j^B \r) \r]
\\
&\hspace{0pt} \stackrel{\text{DL}}{=}
\sum_{w,\mbf{t}, \mbf{s} }
f_{W}(w)f_{\mbf{X}^{(0)}}\l(\mbf{x^{(0)}}\r) \\
&\hspace{0pt} \times \prod_{i\in V}\zeta_i\l(x_i^{(0)}, \mbf{t}_i \r) \gamma_i \l(w,\mbf{t}_i;\wt{x}_i\r) \psi_{\ra i} \l( \mbf{t}_i, \left\{ \mbf{s}_{k\ra i} \right\}_{k\in\partial{i}}  \r)  \\ 
&\hspace{0pt} \times \sum_{\mbf{d} } \prod_{(i,j)\in\vv{E}} 
f_{\mbf{D}_{i\ra j}} \l(\mbf{d}_{i\ra j}\r) 
\11\l[ s_{i\ra j}^{A}=\sigma\left(t_{i\ra j}^A; t_j^A \right) \r] \\
&\hspace{60pt} \times \11\l[ s_{i\ra j}^{B}=\sigma\l(t_{i\ra j}^B; t_j^B \r) \r]\numberthis\label{eq:px0_2}.
\end{align*}
The summand in \eqref{eq:px0_2} now involves a factor $\psi_{\ra i}$ which has a (local) domain of size $O\l(|V|^2 2^{2|\partial{i}|} \r)$. Although with a lesser base, the complexity is still exponential in the node-degrees. This complexity can be reduced further by taking advantage of the fact that $s_{k\ra i}^I$'s are multiplicatively separable in each term of $\psi_{\ra i}^I$ (see \eqref{eq:psiii2}). This fact is exploited in the derivation of message-update rules, (a step which we refer to as \textit{within-message-update optimization}), where the above complexity is further reduced to $O\l(|V|^2 |\partial{i}|\r)$.\footnote{The curious reader may rightfully question why we did not exploit this fact right here in the onset before proceeding to message-passing. The reason is that the local-function $\psi_{\ra i}$ is not a product form but rather a sum-product form. The only way to gain advantage of such sum-product forms of a local-function is by invoking the distributive-law during the message-update rules (where we need not concern with the rest of the local-functions). Traditionally, in the context of belief-propagation, the use of distributive-law is usually concerned with exploiting the factorization of the global-function instead of the local-functions themselves. Therefore, to emphasize this step of invoking distributive-law later during the message-passing, we refer to it as \textit{within-message-update optimization}. We will use this in Section \ref{sec:messageupdaterules_1} (with details given in Appendix \ref{app:messageupdaterules-1}).}.

With the aid of consistency-checks on realizations of $\mbf{S}$, the sum over realizations of $\mbf{D}$ in \eqref{eq:px0_2} can be easily computed.
\begin{claim}\label{claim:getridofd}
Given $\mcl{O}_n = \l\{ \wt{\mbf{X}}^{(W)} = \wt{\mbf{x}}\r\}$,
\begin{align*}
&\pr \l( \mbf{x}^{(0)}|\mcl{O}_n \r) \propto 
\sum_{w,\mbf{t}, \mbf{s} }f_{W}(w)f_{\mbf{X}^{(0)}}\l(\mbf{x^{(0)}}\r)  \\
&\hspace{0pt} \times \prod_{i\in V}\zeta_i\l(x_i^{(0)}, \mbf{t}_i \r)\gamma_i \l(w,\mbf{t}_i;\wt{x}_i\r) \psi_{\ra i} \l( \mbf{t}_i, \left\{ \mbf{s}_{k\ra i} \right\}_{k\in\partial{i}}  \r)  \\ 
&\hspace{0pt} \times \prod_{  \{i,j\}\in E } \psi_{i\xlrarrow{} j} \l( \mbf{t}_i, \mbf{s}_{j\ra i}, \mbf{t}_j, \mbf{s}_{i\ra j} \r),\numberthis\label{eq:px0_3}
\end{align*}
where $\psi_{i\xlrarrow{} j}(\cdot)$ is given by
\begin{align*}
&\psi_{i\xlrarrow{} j}(\cdot) \defeq \psi_{i\ra j}\l(\mbf{t}_i, \mbf{s}_{i\ra j}; \mbf{t}_j \r) \psi_{j\ra i}\l(\mbf{t}_j, \mbf{s}_{j\ra i}; \mbf{t}_i \r), \numberthis\label{eq:psiijji}  \\
&\psi_{i\ra j}(\cdot) \defeq \prod_{I\in \{A,B\}} \psi_{i\ra j}^{I}\l(t_i^I, s_{i\ra j}^I; t_j^I, t_j^J\r) 
,\numberthis\label{eq:psiij}\\
&\psi_{i\ra j}^I(\cdot) \defeq \\
\begin{split}
\begin{cases}
1-\11\l[t_i^I<t_j^I\r] \lambda_{i\ra j}^{I\ra \emptyset} &\text{if } t_j^I<\infty, s_{i\ra j}^{I}=1, t_i^I<t_j^J; \\
1-\11\l[t_i^I<t_j^I\r] \lambda_{i\ra j}^{I\ra \{J\}} &\text{if } t_j^I<\infty, s_{i\ra j}^{I}=1, t_i^I\ge t_j^J; \\
\11\l[t_i^I+1=t_j^I\r]\lambda_{i\ra j}^{I\ra \emptyset} &\text{if } t_j^I<\infty, s_{i\ra j}^{I}=0, t_i^I<t_j^J;\\
\11\l[t_i^I+1=t_j^I\r]\lambda_{i\ra j}^{I\ra \{J\}} &\text{if } t_j^I<\infty, s_{i\ra j}^{I}=0, t_i^I\ge t_j^J;\\
1 - \11\l[t_i^I<\infty\r]\lambda_{i\ra j}^{I\ra \emptyset} &\text{if } t_j^I=\infty, s_{i\ra j}^{I}=0, t_i^I<t_j^J;\\
1 - \11\l[t_i^I<\infty\r]\lambda_{i\ra j}^{I\ra \{J\}} &\text{if } t_j^I=\infty, s_{i\ra j}^{I}=0, t_i^I\ge t_j^J;\\
0 &\text{otherwise},
\end{cases}\\
\l(I\in\{A,B\}, J\in \{A,B\}\setminus \{I\}\r).
\end{split}\numberthis\label{eq:psiiji}
\end{align*}
\end{claim}
The proof of Claim \ref{claim:getridofd} is algebraic; for completeness, it is provided in Appendix \ref{app:getridofd}. Here, each $\psi_{i\ra j}^I$ is a conditional probability of $S_{i\ra j}^I$ given $T_j^I, T_j^J, T_i^I$.

\subsubsection{Reducing Support of $W$}
When $w_{max}=\infty$, the support of $W$ is countably-infinite. This is not much of an issue---one can reduce the local-domain of $f_W$ without using any approximation. Note that the variable $w$ appears only in the local-functions $f_W$ and $\gamma_i$'s. Furthermore, for any fixed $\wt{x}_i \in \mcl{X}$ and $\mbf{t}_i \in \mcl{T}^2$, 
\begin{align*}
    \gamma_i\l(w, \mbf{t}_i; \wt{x}_i \r) &= \gamma_i\l(|V|, \mbf{t}_i; \wt{x}_i \r) \\ 
    &\hspace{-15pt} \text{for all } w \in \l\{ |V|, |V|+1, \dots \r\}.\numberthis\label{eq:gammai:largerws}
\end{align*}
To use \eqref{eq:gammai:largerws}, we define ${\mcl{W}_1} \defeq [w_{min}, (|V|-1) \wedge w_{max} ] \cup \{ \dagger \} $ and let
\begin{align*}
\udl{f}_W (w) &\defeq 
\begin{cases}
f_W(w), & w \in {\mcl{W}_1}\setminus \{\dagger\},\\ 
\sum\limits_{\substack{w = |V| }}^{w_{max}} f_W(w) & w = \dagger. \numberthis\label{eq:fw1}
\end{cases}
\end{align*}
and
\begin{align*}
\udl{\gamma}_i \l( w, \mbf{t}_i; \wt{x}_i \r) &\defeq 
\begin{cases}
\gamma_i \l( w, \mbf{t}_i; \wt{x}_i \r), & w \in {\mcl{W}_1} \setminus \{\dagger\},\\
\gamma_i \l( |V|, \mbf{t}_i; \wt{x}_i \r) & w = \dagger.\numberthis\label{eq:gammai1}
\end{cases}
\end{align*}
Then, we can rewrite \eqref{eq:px0_3} as
\begin{align*}
&\pr \l( \mbf{x}^{(0)}|\mcl{O}_n \r) \propto 
\sum_{w,\mbf{t}, \mbf{s} }\udl{f}_{W}(w)f_{\mbf{X}^{(0)}}\l(\mbf{x^{(0)}}\r)  \\
&\hspace{0pt} \times \prod_{i\in V}\zeta_i\l(x_i^{(0)}, \mbf{t}_i \r) \udl{\gamma}_i \l(w,\mbf{t}_i;\wt{x}_i\r) \psi_{\ra i} \l( \mbf{t}_i, \left\{ \mbf{s}_{k\ra i} \right\}_{k\in\partial{i}}  \r)  \\ 
&\hspace{0pt} \times \prod_{  \{i,j\}\in E } \psi_{i\xlrarrow{} j} \l( \mbf{t}_i, \mbf{s}_{j\ra i}, \mbf{t}_j, \mbf{s}_{i\ra j} \r),\numberthis\label{eq:px0_4}
\end{align*}
\section{Inference Via Marginalization For Small Problem Instances}\label{sec:inference_1}
In this section, we develop a BP algorithm that solves the inference problems \eqref{eq:argmaxpx0}-\eqref{eq:argmaxptii} exactly, whenever the network $G$ is a forest. Understandably, the algorithm is not scalable to large problem instances simply because solving any one of these problems exactly is NP-hard -- for example, one needs exponential memory just to store the posterior probabilities, $\{ \pr\l(\mbf{x}^{(0)}|\mcl{O}_n\r) \}$. However, the algorithm presented in this section would be implementable for small problem instances (namely when the graph $G$ is small). Furthermore, once developed, a scalable and convergent extension for approximate solutions, as done in Sections \ref{sec:inference_2} and \ref{sec:convergence}, will be natural and easy to follow.
\begin{figure}[!t]
\centering
\includegraphics[height=1.2in]{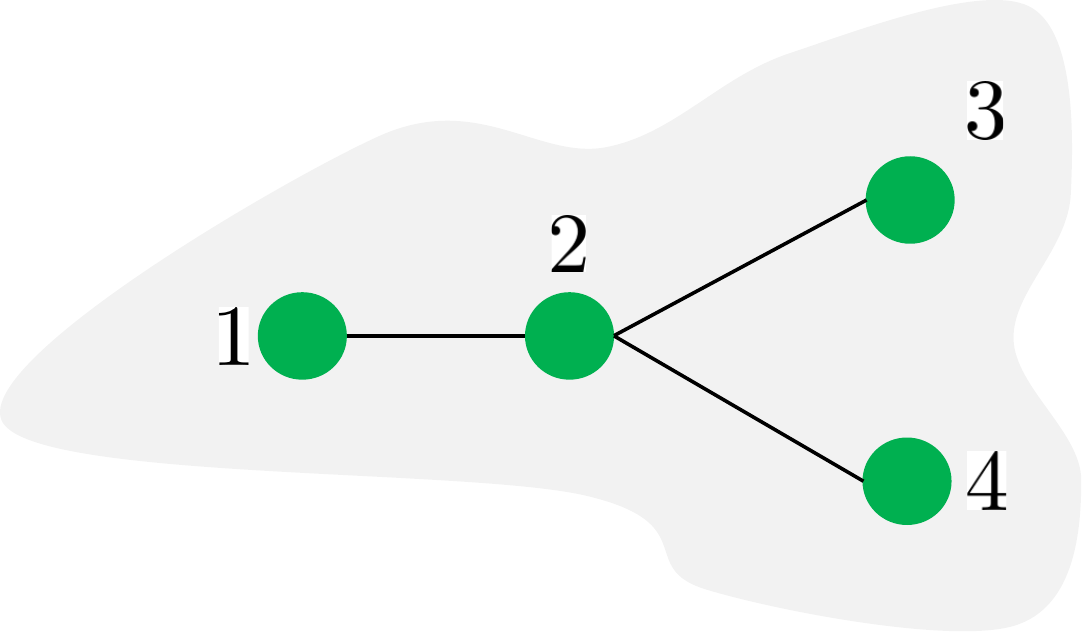}
\caption{Example Network, $G$.}
\label{fig:exampleg}
\end{figure}

\subsection{Factor-Graph Representation}\label{subsec:fg}
We are interested in a suitable factor-graph\footnote{For review of factor-graphs, see \cite{kleinberg10}.} representation of the summand in \eqref{eq:px0_4}, i.e., the \emph{global-function}
\begin{align*}
\Gamma \defeq \udl{f}_W f_{\mbf{X}^{(0)}}\prod_{i\in V}\zeta_i \udl{\gamma}_i \psi_{\ra i} \prod_{\{i,j\}\in E } \psi_{i\xlrarrow{} j}.\numberthis\label{eq:biggamma}
\end{align*}
One may easily check that the standard factor-graph of \eqref{eq:biggamma} in which each variable is assigned a separate (variable) node is far from loop-free, even when the underlying network $G$ is acyclic (in the undirected sense).

The BP algorithm provides exact marginalization when the factor-graph it operates on is a tree (or forest), and in general, provides good approximation on factor-graphs which have few (and preferably long) cycles. In principle, we can convert any cyclic factor-graph into an acyclic one, but this often comes at a great expense in complexity. A sufficient set of \emph{factor-graph transformations} that allow such conversion are given in \cite[Section VI]{kschischang01}. Below, we list a sequence of only those transformations that, when applied to the standard factor-graph of \eqref{eq:biggamma}, convert it into an equivalent one that resembles $G$, in particular its loop-structure. The motivation behind introducing these transformations is that \textit{i}) they will enable exact estimation of $\pr\l(\mbf{x}^{(0)}|\mcl{O}_n\r)$ whenever $G$ is acyclic, and \textit{ii}) hopefully provide a good approximation of  $\pr\l(\mbf{x}^{(0)}|\mcl{O}_n\r)$ when $G$ is locally tree-like. For a connected graph $G$, these transformations are step-wise listed below and are illustrated for the example network shown in \figurename \ref{fig:exampleg}. The extension to graphs with more than one connected-components is discussed thereafter.
\begin{enumerate}
\item[(a)] Cluster the variables $\l\{ x_i^{(0)} \r\}_{i\in V}$ into $\mbf{x}^{(0)}$. Then, cluster $w$ and $\mbf{x}^{(0)}$ into $\{ w, \mbf{x}^{0} \}$ and let $v_G$ denote the variable node that corresponds to it. Then, for every $i\in V$, cluster $\l\{ t_i^A,t_i^B \r\}$ into $\mbf{t}_i$ and for every $(i,j)\in \vv{E}$, cluster $\l\{s_{i\ra j}^{A},s_{i\ra j}^{B} \r\}$ into $\mbf{s}_{i\ra j}$. Finally, cluster $\udl{f}_W$ and $f_{\mbf{X}^{(0)}}$ into factor node 
$\udl{f}_{W, \mbf{X}^{(0)}} \defeq \udl{f}_W f_{\mbf{X}^{(0)}}$, and for every $i\in V$ cluster $\zeta_i$, $\gamma_i$, and $\psi_i$ into $\zeta_i\gamma_i\psi_i$. See \figurename \ref{fig:stepa}.

\item[(b)] For every $i\in V$, stretch $\mbf{t}_i$ to $\mbf{s}_{j \ra i}$ where $j \in\partial{i}$. Afterwards, remove the redundant variable node $\mbf{t}_i$ (along with all of its edges). See \figurename \ref{fig:stepb}.

\item[(c)] Stretch variables of $v_G$ throughout the factor-graph. This results in a (variable) node $v_{j \ra i}$ that corresponds to the variables $\l\{ w,\mbf{x}^{(0)},\mbf{t}_i, \mbf{s}_{j\ra i} \r\}$. 
Now, remove the redundant edges of $v_G$. See \figurename \ref{fig:stepc}. At this point, it is good to note that for each variable of the global-function $\Gamma$, the sub-graph of variable (and factor) nodes that contain it is connected.

\item[(d)] Choose $\wh{e}=\{\wh{i},\wh{j}\}$ from $E$ arbitrarily and connect a new copy of variable node $v_G$ to it. Then, connect the factor node $\udl{f}_{W, \mbf{X}^{(0)}}$ to this new copy of $v_G$ while removing the edge between $\udl{f}_{W, \mbf{X}^{(0)}}$ and original copy of $v_G$. Finally, remove the original (now-redundant) copy of $v_G$. See \figurename \ref{fig:stepd}.
\end{enumerate}
The above set of transformations correspond to the following equivalent (inflated) version of \eqref{eq:biggamma}.
\begin{align*}
&\udl{f}_{W} f_{\mbf{X}^{(0)}} \prod_{i\in V} 
f_i \l( w,\mbf{x}^{(0)},\mbf{t}_i, \{\mbf{s}_{k\ra i} \}_{k\in\partial{i}} \r) \\
&\hspace{20pt} \times 
\prod_{\{i,j\}\in E } f_{ij} \l( w, \mbf{x}^{(0)}, \mbf{t}_i,\mbf{s}_{i \ra j}, \mbf{t}_j,\mbf{s}_{j\ra i} \r),
\end{align*}
where,
\begin{align*}
f_i \l( \cdot \r) \defeq 
\zeta_i 
\udl{\gamma}_i
\psi_{\ra i}, 
\quad \text{ and } \quad
f_{ij} \l( \cdot \r) \defeq 
\psi_{i\xlrarrow{} j}.\numberthis\label{eq:fiandfij}
\end{align*}
When $G$ has more than one connected-components, we can apply steps similar to (b), (c), and (d) for each connected-component and obtain the following factor-graph representation. 
\begin{align}
\begin{split}\label{eq:gbiggamma}
G_{\Gamma} &\defeq \l(V_{\Gamma,1}\cup V_{\Gamma,2}, E_{\Gamma}\r), \\
V_{\Gamma,1} &\defeq \l\{ v_{C} : C \in \mcl{C} \r\} 
\cup \l\{ v_{j \ra i} :(j, i)\in\vv{E} \r\},\\
V_{\Gamma, 2} &\defeq \l\{ \udl{f}_{W, \mbf{X}^{(0)}} \r\} \cup \l\{ f_{i}:i\in V \r\} \cup \l\{ f_{ij}: \{i,j\}\in E \r\}, \\
E_{\Gamma} &\defeq \l\{
\l\{ \l\{ f_i, v_{k \ra i} \r\} \r\}_{i\in V}^{k\in\partial{i}}, 
\l\{ \l\{ f_{ij}, v_{j\ra i} \r\} , \l\{ f_{ij}, v_{i\ra j} \r\} \r\}_{\{i,j\}\in E} \r.\\
&\hspace{20pt}\l. \l\{v_{C}, \udl{f}_{W, \mbf{X}^{(0)}} \r\}_{C \in \mcl{C}}, \l\{v_{C} ,f_{\wh{e}^{(C)}}\r\}
\r\}.
\end{split}
\end{align}
Here, $v_C$ corresponds to a copy of $\{w, \mbf{x}^{(0)}\}$ for connected-component $C$ and $\wh{e}^{(C)}$ is an arbitrarily chosen edge from $E_C$ (the edge-set of connected-component $C$). Figure \ref{fig:exampleggamma} shows an example of deriving $G_{\Gamma}$ for a given network $G$. 

\noindent \textbf{Remark}: $G_{\Gamma}$ is always connected.

\begin{figure*}[!b]
\centering
\subfloat[Factor-Graph of $\Gamma$ after step (a).]{\includegraphics[width=0.49\linewidth]{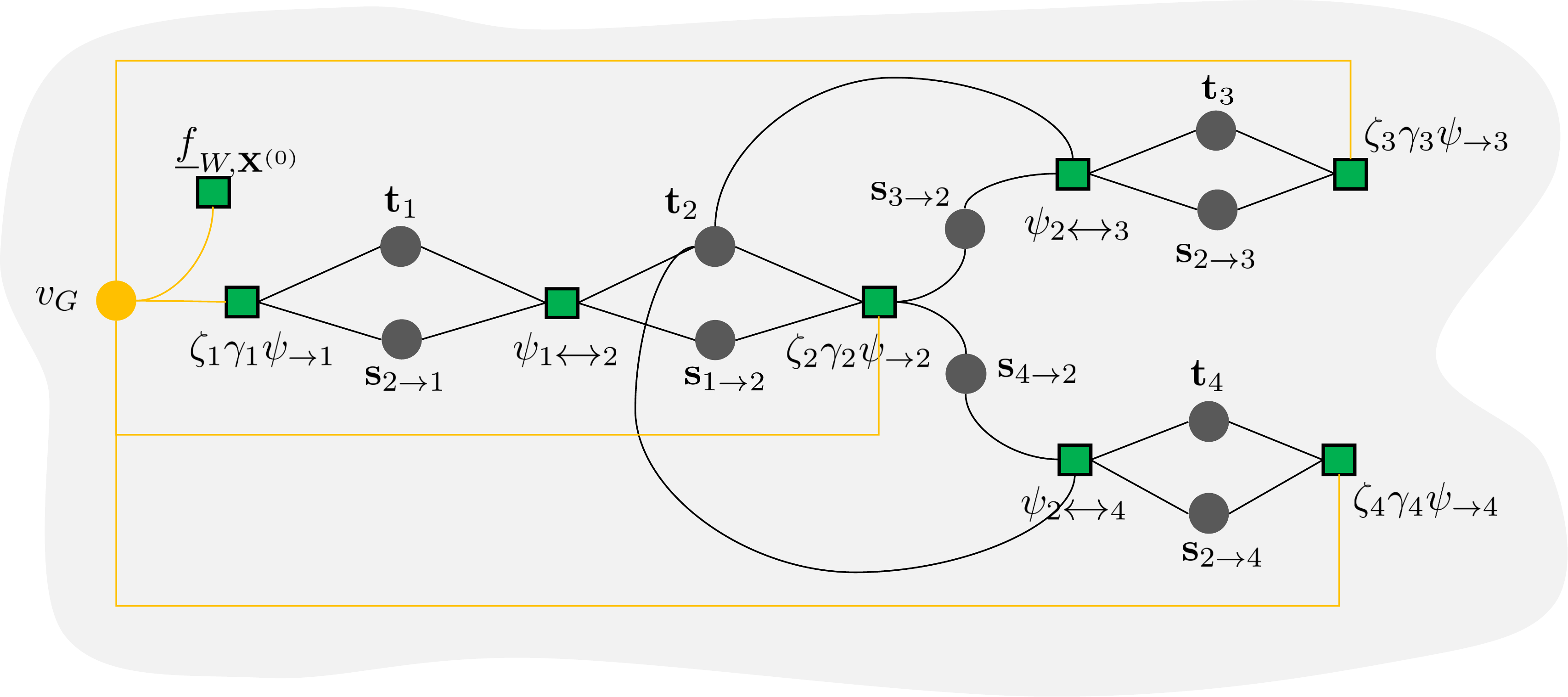}
\label{fig:stepa}}
\hfil
\subfloat[Factor-Graph of $\Gamma$ after step (b).]{\includegraphics[width=0.49\linewidth]{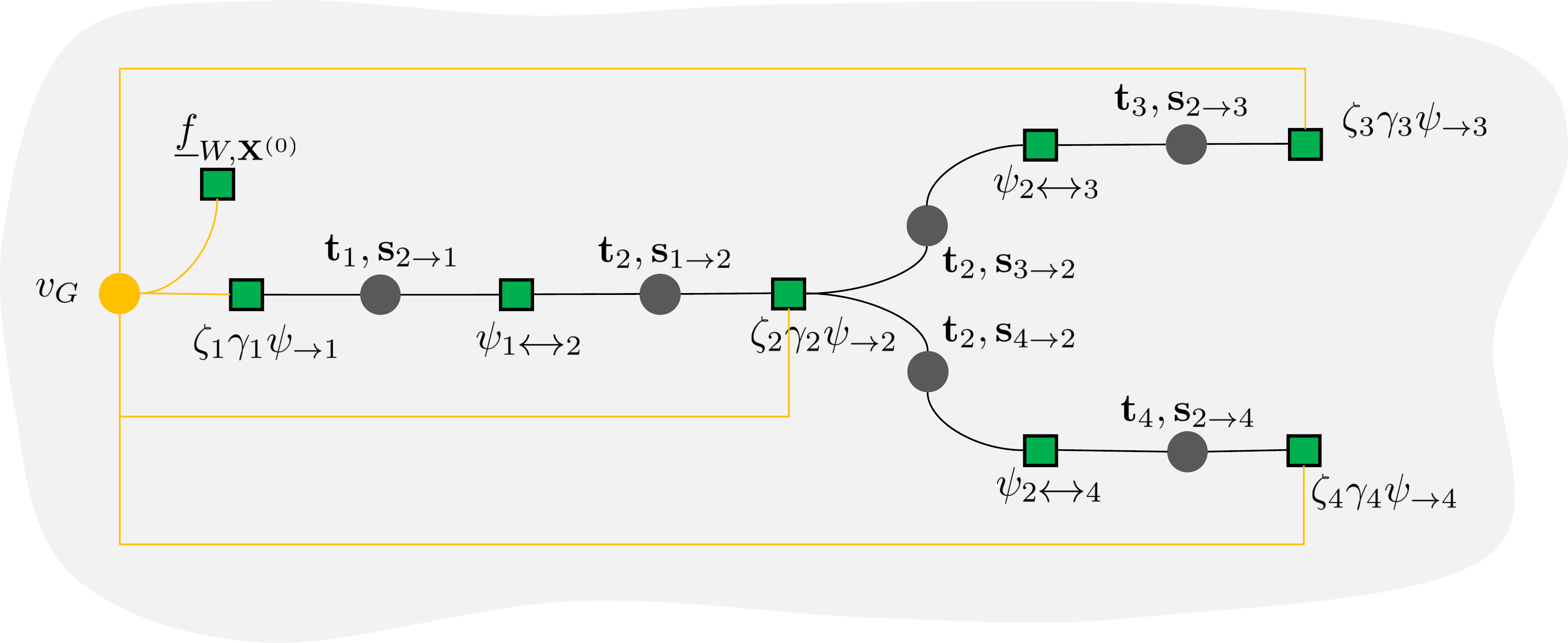}
\label{fig:stepb}}
\hfil
\subfloat[Factor-Graph of $\Gamma$ after step (c).]{\includegraphics[width=0.49\linewidth]{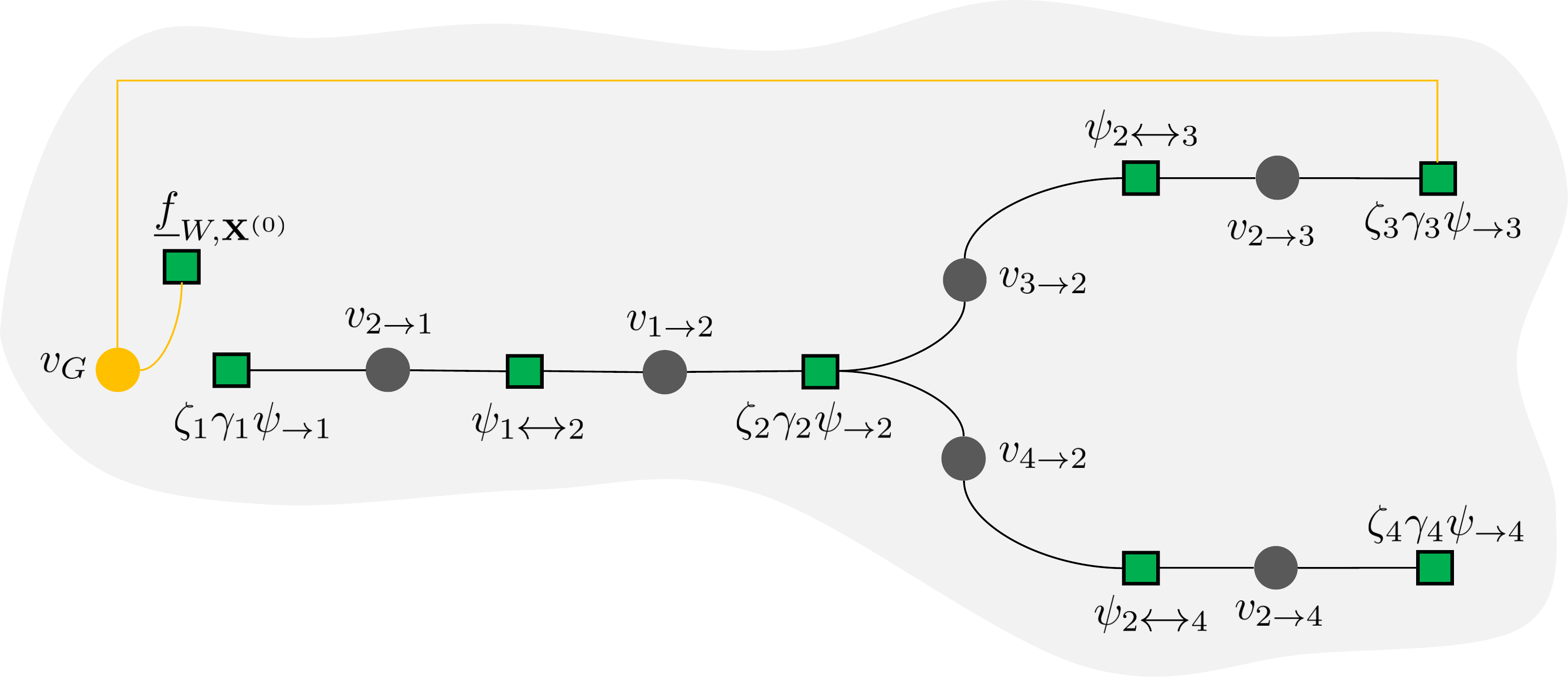}
\label{fig:stepc}}
\hfil
\subfloat[Factor-Graph of $\Gamma$ after step (d).]{\includegraphics[width=0.49\linewidth]{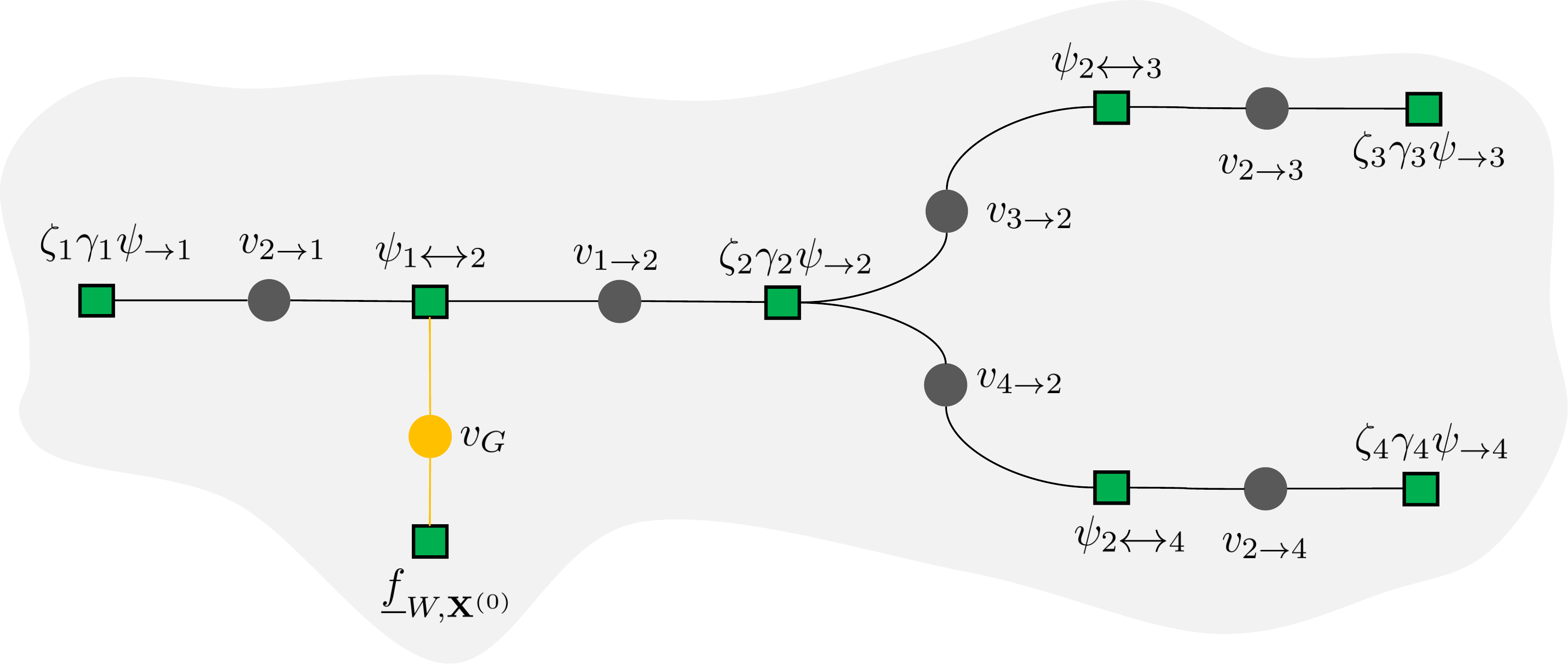}
\label{fig:stepd}}
\caption{Deriving Factor-Graph $G_{\Gamma}$ for the connected network of Figure \ref{fig:exampleg}.}
\label{fig:example}
\end{figure*}
\begin{figure*}[!t]
\centering
\includegraphics[width=.8\linewidth]{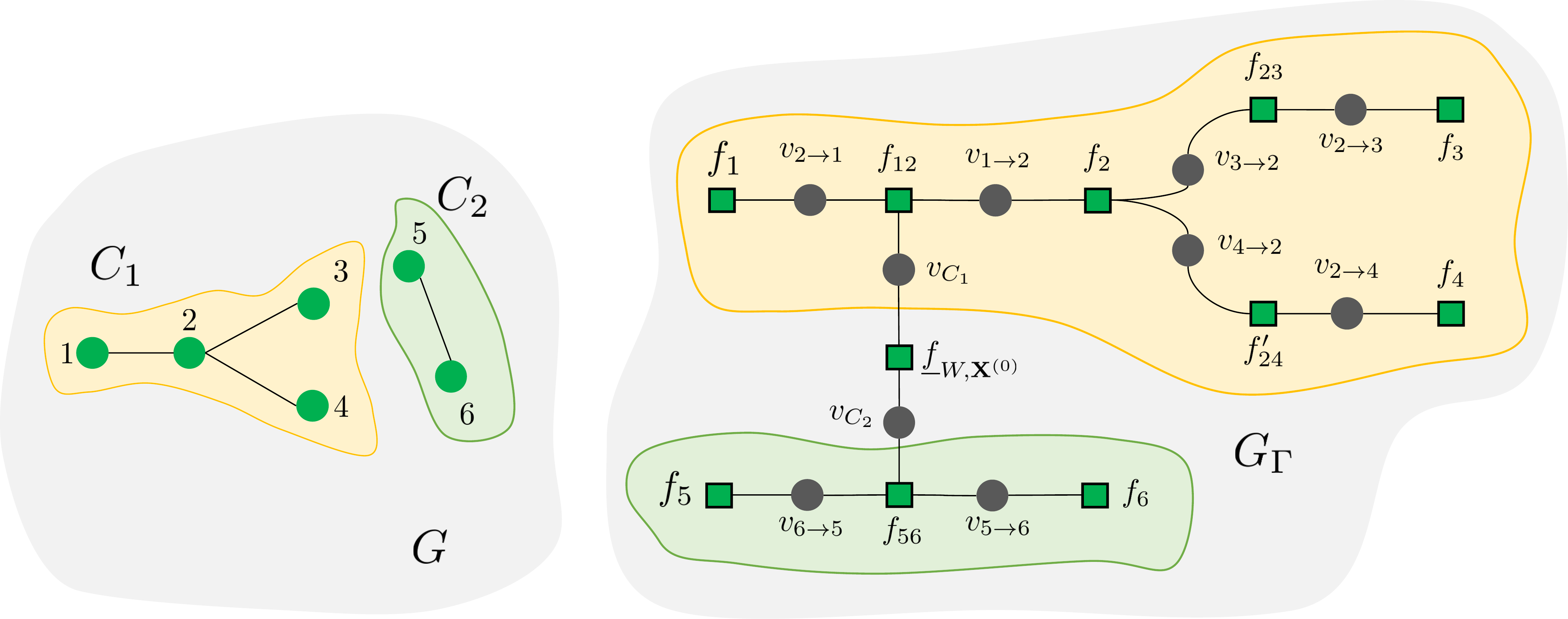}
\caption{Example of Generating Factor-Graph $G_{\Gamma}$ from $G$. There are two connected components in $G$ namely $C_1$ and $C_2$ which are formed from the node-sets $\{1, 2, 3, 4\}$ and $\{5, 6\}$ respectively. Here, we chose $\{1, 2\}$ from
$C_1$ ($\wh{e}^{(C1)} = \{1, 2\}$), and $\{5, 6\}$ from $C_2$ ($\wh{e}^{(C2)} = \{5, 6\}$).}
\label{fig:exampleggamma}
\end{figure*}
\subsection{Belief Propagation Algorithm and Pseudo-Marginals}
Having obtained a suitable factor-graph representation that resembles the underlying network $G$, we are now ready to derive the BP algorithm for approximate marginalization. Since our factor-graph has clusters of variables, let us introduce the function $r(\cdot)$ which maps a variable-node to the variables it represents and similarly, maps a factor-node to the union of variables represented by its neighboring variable nodes. For example in $G_{\Gamma}$, $r(v_{j\ra i}) = \{w, \mbf{x}^{(0)}, \mbf{t}_i, \mbf{s}_{j\ra i} \}$ and $r(f_i) = \{w, \mbf{x}^{(0)}, \mbf{t}_i, \{\mbf{s}_{k\ra i}\}_{k\in\partial{i}} \}$.

The BP algorithm is a message-passing scheme that iteratively computes an approximation to the true marginals by passing two kinds of messages on the factor-graph: at iteration-$k$,
the message from a factor node $f\in V_{\Gamma,2}$ to its neighboring variable node $v\in N(f)$\footnote{$N(\cdot)$ returns the neighborhood set of $\cdot$.} is denoted by $$\mbf{m}_{f\ra v}^{(k)}=\l\{ m_{f\ra v}^{(k)} \l(\mbf{z}_{r(v)}\r) \r\}_{\mbf{z}_{r(v)}},$$ and the opposite direction message from $v$ to $f$ is denoted by $$\mbf{n}_{v\ra f}^{(k)} = \l\{ n_{v\ra f}^{(k)} \l(\mbf{z}_{r(v)}\r)\r\}_{\mbf{z}_{r(v)}}.$$ 
Here, $\mbf{z}_{r(v)}$ denotes a joint realization of the variables in $r(v)$.\footnote{The message $\mbf{m}_{f\ra v}^{(k)}$ is a vector which has one entry for each possible realization of $\mbf{z}_{r(v)}$.} At iteration-0, all messages are initialized to some positive vector, say the all-ones vector, following which they are iteratively computed by the below update rule,
\begin{align*}
m_{f\ra v}^{(k+1)} \l(\mbf{z}_{r(v)}\r) &\propto \sum_{\mbf{z}_{r(f) \setminus r(v)} } f\l(\mbf{z}_{r(f)} \r) \prod_{v'\in N(f) \setminus v } n_{v'\ra f }^{(k)}\l(\mbf{z}_{r(v')}\r)\numberthis\label{eq:f2v_1}, \\
n_{v \ra f}^{(k)}\l(\mbf{z}_{r(v)}\r) &\propto \prod_{g \in N(v)\setminus f} m_{g\ra v}^{(k)}\l(\mbf{z}_{r(v)}\r).\numberthis\label{eq:v2f}
\end{align*}
Here, the $\propto$ sign is meant to indicate that each message $\mbf{m}_{f\ra v}^{(k+1)}$ (or $\mbf{n}_{v\ra f}^{(k)}$) must be normalized after all of its entries have been computed and the symbol $\mbf{z}_{r(f)\setminus r(v)}$ in \eqref{eq:f2v_1} is used to indicate that the summation is performed on all variables in $r(f)\setminus r(v)$. Intuitively speaking, when a factor node $f$ transmits a message to $v$, it gathers the latest messages received from all of its neighboring nodes except $v$, and then uses them to transmit information to $v$; same applies to the variable node $v$.\footnote{The detailed mechanics of \eqref{eq:f2v_1} and \eqref{eq:v2f} can be found in \cite{kschischang01}.}. 

To reduce the storage requirements of the message-passing scheme, one can get rid of all the variable-to-factor messages by substituting \eqref{eq:v2f} in \eqref{eq:f2v_1}. This gives
\begin{align*}
&m_{f\ra v}^{(k+1)} \l(\mbf{z}_{r(v)}\r) \\
&\hspace{5pt} \propto \hspace{-2pt} \sum_{\mbf{z}_{r(f) \setminus r(v)}} \hspace{-2pt} f\l(\mbf{z}_{r(f)}\r)\hspace{-2pt} \prod_{v'\in N(f)\setminus v }
\hspace{-2pt} \l( 
\prod_{g\in N(v')\setminus f } m_{g\ra v'}^{(k)} \l(\mbf{z}_{r(v')}\r) 
 \r)
.\label{eq:f2v_2}\numberthis	
\end{align*}
\textbf{Remark}: The update rule \eqref{eq:f2v_2} is usually known as ``lazy'' update rule. In practice, one uses the update rule which always utilizes the ``most recently updated messages,'' that is, some of $m_{g\ra v'}^{(k)}$ are $m_{g\ra v'}^{(k+1)}$ in the right hand side of \eqref{eq:f2v_2}. This form of message-updating is referred to as ``impatient'' and has faster convergence properties than the lazy version \cite{taga06}.

The factor-to-variable node messages serve to provide approximate marginal probabilities 
$$\l\{\mbf{q}^{(k)}_{r(v)} = \{q^{(k)}_{r(v)}(\mbf{z}_{r(v)})\}\r\}_{v \in V_{\Gamma,1} },$$ 
which are also iteratively updated as follows:
\begin{align*}
q_{r(v)}^{(k)}\l( \mbf{z}_{r(v)} \r) \ \propto \ 
\prod_{f \in N(v)} m^{(k)}_{f\ra v}\l( \mbf{z}_{r(v)}\r).\label{eq:qv}\numberthis
\end{align*}
\textbf{Remark}: In general, a joint distribution what has marginals given by \eqref{eq:qv} and which satisfies all the marginalization constraints of the complete distribution may not exist. For this reason, the marginals retrieved by \eqref{eq:qv} are called \emph{pseudo-marginals}. For more details, see \cite[Section V.A]{yedida05}.

Whenever $G$ is a forest, $G_{\Gamma}$ is a tree, and the above equations result in exact marginal probabilities -- in number of iterations that is at most the diameter\footnote{Diameter of a graph is defined as the maximum shortest distance (in terms of number of nodes traversed) between any pair of nodes.} of $G_{\Gamma}$, \cite{aji00-gdl, kschischang01}. Specifically, we will be able to extract the posterior probabilities, $\pr\l(\mbf{x}^{(0)}|\mcl{O}_n\r)$, $\pr\l(x_i^{(0)}|\mcl{O}_n\r)$, and $\pr\l(t_i^I|\mcl{O}_n\r)$ from the pseudo-marginals listed below. 
\begin{align}
\begin{split}\label{eq:posteriorbeliefs}
q_{w, \mbf{x}^{(0)}} \l( w, \mbf{x}^{(0)} \r) 
&\propto m_{f_{W, \mbf{X}^{(0)}} \ra v_{C_1} } \l( w, \mbf{x}^{(0)} \r) \\
&\hspace{20pt} \times m_{f_{\wh{e}^{(C_1)} } \ra v_{C_1} } \l( w, \mbf{x}^{(0)} \r), \\
q_{\mbf{x}^{(0)}} \l( \mbf{x}^{(0)} \r) 
&\propto 
\sum\limits_{w} q_{w, \mbf{x}^{(0)}}, \\
q_{{x}_i^{(0)}}\l({x}_i^{(0)}\r) &\propto \sum_{\mbf{x}^{(0)}\setminus x_i^{(0)} } q_{\mbf{x}^{(0)}}\l(\mbf{x}^{(0)}\r),\\
q_{{t}_i^I}(t_i^I) &\propto 
\sum_{\substack{w, \mbf{x}^{(0)}, t_i^J, \mbf{s}_{j\ra i}}}
m_{f_i\ra v_{j\ra i}} 
\l(w, \mbf{x}^{(0)}, \mbf{t}_i, \mbf{s}_{j\ra i} \r)
\\
&\hspace{20pt} \times 
m_{f_{ij} \ra v_{j\ra i}} 
\l(w, \mbf{x}^{(0)}, \mbf{t}_i, \mbf{s}_{j\ra i} \r).
\end{split}
\end{align}
In the last equation, $j$ is an arbitrarily chosen neighbor of node $i$ and $J$ is the (single) element of $\{ A, B\} \setminus \{I\} $.

\subsection{Message Update Rules}\label{sec:messageupdaterules_1}
In our factor-graph representation $G_{\Gamma}$, there are four kinds of messages (see \figurename \ref{fig:messages}). These are 
\begin{enumerate}
\item[a)]
Messages from (factor node) $f_i$ to (variable node) $v_{j\ra i}$. 

\item[b)]
Messages from $f_{ij}$ to $v_{j\ra i}$ and $v_{i\ra j}$. 

\item[c)]
Messages from $f_{\wh{e}^{(C)}}$ to $v_{C}$.
 
\item[d)]
Messages from $\udl{f}_{W, \mbf{X}^{(0)}}$ to $v_{C}$.
\end{enumerate}
These messages require an aggregate storage of $O\l(|E| |{\mcl{W}_1}| | \bm{\mcl{X}}^{(0)}||V|^2\r)$. Time-wise, one may note from \eqref{eq:f2v_2} that the complexity of computing $\mbf{m}_{f\ra v}^{(k)}$ is, in general, linear in the number of possible values of $\mbf{z}_{r(f)}$. Therefore, for $G_{\Gamma}$, the straight-forward computation of $\mbf{m}_{f_i\ra v_{j\ra i}}^{(k)}$ will involve $O\l(|{\mcl{W}_1}| |\bm{\mcl{X}}^{(0)}| |V|^2 2^{2|\partial{i}|}\r)$, a number exponential in $|\partial{i}|$. However, as mentioned in Section \ref{subsec:ftioptimization}, we can invoke distributive-law on the sum-product form of $\psi_{\ra i}$. This, as it turns out, will bring down the number of computations to $O\l(|{\mcl{W}_1}||\bm{\mcl{X}}^{(0)}| |V|^2 |\partial{i}|\r)$. For the derivation of the aforementioned message-update rules, please see Appendix \ref{app:messageupdaterules-1}.

Algorithm \ref{alg:sumproductbp} shows the pseudocode for the BP algorithm based on $G_{\Gamma}$, and its (worst-case) per-iteration-complexity, noting that BP is amenable to parallel execution, is $O\l( \frac{ |{\mcl{W}_1}| |\bm{\mcl{X}}^{(0)}| |E| |V|^4}{\text{\# of worker-nodes}}\r)$ (see the end of Appendix \ref{app:messageupdaterules-1}).
\begin{figure*}[!htb]
\centering
\subfloat
{\includegraphics[width=.9\linewidth]{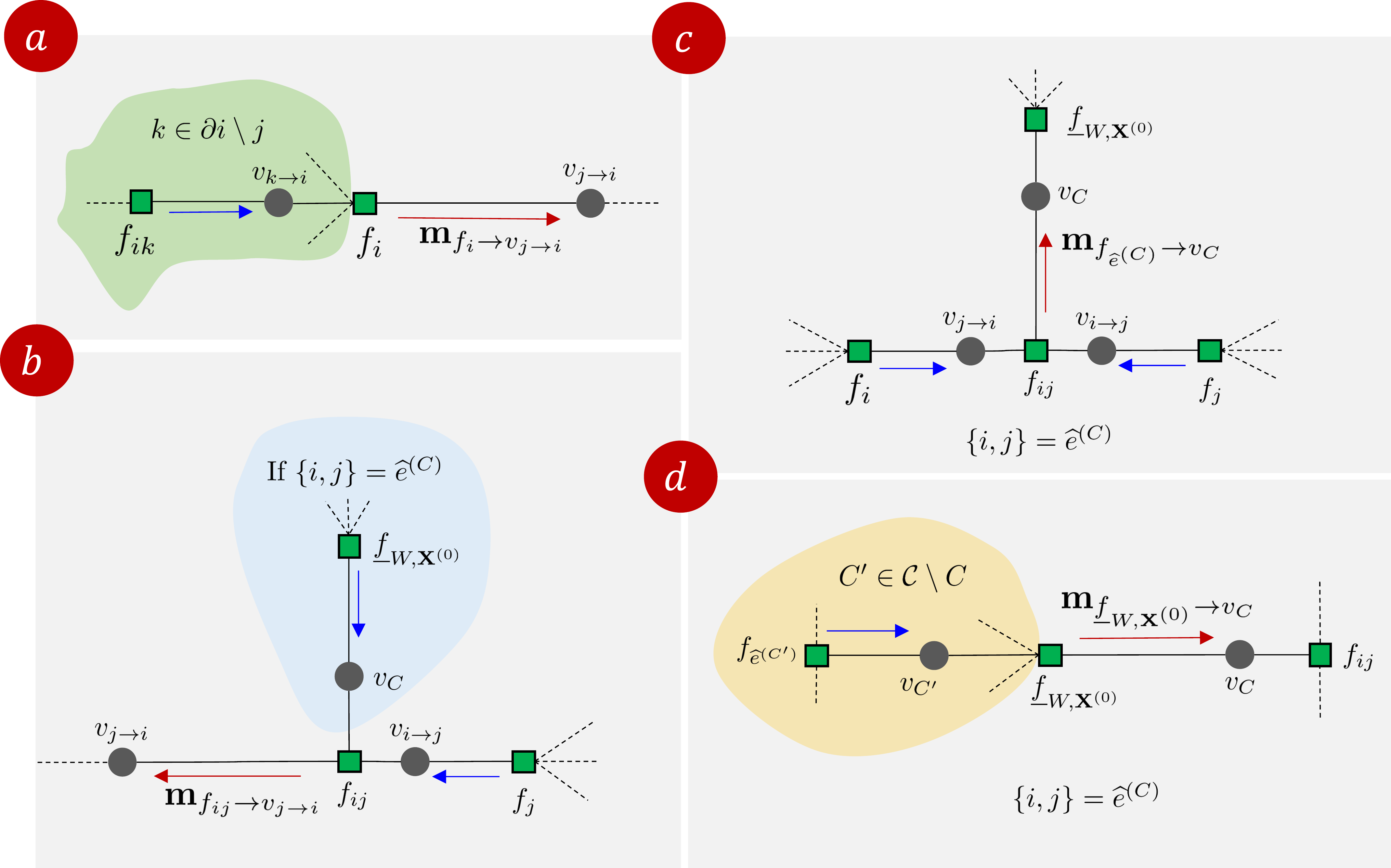}}
\caption{Types of Messages in $G_{\Gamma}$. In each subfigure, the red arrow indicates a message and the blue arrows indicate the messages used in transmitting it.}
\label{fig:messages}
\end{figure*}

{
\setlength{\algomargin}{1.5em}
\begin{algorithm2e}[!hbt]
\DontPrintSemicolon
\KwInput{$G$, $\bsl{\lambda}$, $f_W$, $f_{\mbf{X}^{(0)}}$, $\l\{  f_{\wt{X}_i | X_i }   \r\}_{i\in V}$, $\mcl{O}_n$.}
\Parameter{$max\_iters$.}
\KwOutput{$\wh{\mbf{x}}^{(0)}$, $\{ \wh{x_i}^{(0)} \}_{i\in V}$, $\{\wh{t}_i^A, \wh{t}_i^B \}_{i\in V}$.
}

\textbf{Factor-Graph}: Derive $G_{\Gamma}$ from $G$. See \eqref{eq:gbiggamma}.

Set $t \la 0$.

\textbf{Initialization}: Initialize each factor-to-variable node message to all-ones vector.

\Repeat{\nonl convergence OR $t>max\_iters$.}{
Update all factor-to-variable node messages. See Appendix \ref{app:messageupdaterules-1}.

$t\la t+1$.}

Set
\begin{align*}
\wh{\mbf{x}}^{(0)} &\in \argmax_{\mbf{x}^{(0)}\in\bm{\mcl{X}^{(0)}}} q_{\mbf{x}^{(0)}} \l(\mbf{x}^{(0)}\r),\\
\wh{x_i}^{(0)} &\in \argmax_{x_i^{(0)} \in \mcl{X}} q_{x_i^{(0)}} \l(x_i^{(0)}\r),\\
\wh{t}_i^I &\in \argmax_{t_i^I \in \mcl{T} } q_{t_i^I}(t_i^I),
\end{align*}
where $q_{\mbf{x}^{(0)}}$, $q_{x_i^{(0)}}$, and $q_{t_i^I}$ are given by \eqref{eq:posteriorbeliefs}.
\caption{BP on $G_{\Gamma}$}\label{alg:sumproductbp}
\end{algorithm2e}
}
\section{Inference Via Marginalization For Large Problem Instances}\label{sec:inference_2}
The storage and (worst-case) per-iteration-time complexities of Algorithm \ref{alg:sumproductbp} are $O\l( |{\mcl{W}_1}| |\bm{\mcl{X}}^{(0)} |E| ||V|^2\r)$ and \protect\linebreak $O\l(\frac{ |{\mcl{W}_1}| |\bm{\mcl{X}}^{(0)}| |E| |V|^4}{\text{\# of worker-nodes}}\r)$ respectively. In general, $\bm{\mcl{X}}^{(0)}$ can have size that is on the order of $4^{|V|}$. This renders Algorithm \ref{alg:sumproductbp} infeasible whenever $\bm{\mcl{X}}^{(0)}$ is large or when the network itself is large (large $|V|$). In order to make BP feasible for such large problem instances, one must make a trade-off between complexity and accuracy. In the next subsections, we present modifications to Algorithm \ref{alg:sumproductbp} that when applied together will yield a scalable BP algorithm for large graphs while still maintaining MAP estimation guarantee for one or more of the inference problems (in acyclic graphs).

\subsection{Getting Rid of Clustered Initial States}\label{sec:getrid_clusteredinitialstate}
In Section \ref{sec:inference_1}, the factor-graph $G_{\Gamma}$ had the (clustered) variable-node $\mbf{x}^{(0)}$ which we stretched throughout $G_{\Gamma}$. To remove the dependence of 
complexities on $|\bm{\mcl{X}}^{(0)}|$, we must get rid of $\mbf{x}^{(0)}$. One case in which this is possible is when we forsake exactly solving \eqref{eq:argmaxpx0} (on trees) and assume that the initial-states of all nodes are mutually-independent, i.e., 
\begin{assumption}\label{assump:factorizedfx0}
The distribution of the initial-state of the network has the factorized form given by
\begin{align*}
f_{\mbf{X}^{(0)}} \l(\mbf{x}^{(0)}\r) = 
\prod_{i\in V} f_{X_i^{(0)}} \l(x_i^{(0)}\r).\numberthis\label{eq:fx0factorized}
\end{align*}
\end{assumption}

\subsection{Reducing Large Local-Domains}\label{sec:reducingsupports}
To remove the dependence of complexities on $|{\mcl{W}_1}|$ and $|V|^4$, we invoke the \emph{small-world phenomenon} \cite{kempe15}. Specifically, we assume that for each process, the finite time to reach a given node is at most $T_{max}$ where $T_{max}$ is large enough to accommodate most sample paths of $\{\mbf{X}^{(t)}: t\ge 0\} $ (under $\pr
$) but small enough so that the complexities are linear or sub-linear in $|V|$. Mathematically, this is equivalent to altering the probability measure $\pr
$ to a new one where events inconsistent with the small-world phenomenon are assigned probability 0.
\begin{assumption}\label{assump:tmax}
    Let $\mcl{T}' \defeq \l\{ 0, 1, \dots, T_{max}, \infty \r\}$. The diffusion-model of Section \ref{sec:problem} is such that the basic random-variables $\mbf{X}^{(0)}, W, \mbf{D} $ are resampled if any one of the infection-times is not in $\mcl{T}'$. 
\end{assumption}
\noindent Let $\pr'$ be the probability measure on $\l( \Omega, \mcl{F} \r) $ for the diffusion-model of Assumption \ref{assump:tmax}. It is clear that the event $\cap_{i\in V} \{ T_i^A, T_i^B \in \mcl{T}' \}$ occurs $\pr'$-almost-surely. Therefore, we can restrict the support of each infection-time to $ \mcl{T}'$. Let $\mcl{W}_2 \defeq [w_{min}, T_{max} \wedge w_{max}] \cup \{\dagger\}$ and define
\begin{align*}
f'_W (w) &\defeq 
\begin{cases}
f_W(w), & w \in \mcl{W}_2 \setminus \{\dagger\},\\
\sum\limits_{\substack{w = T_{max}+1 }}^{w_{max}} f_W(w) & w = \dagger. \numberthis\label{eq:fw2}
\end{cases}
\end{align*}
and
\begin{align*}
\gamma'_i \l( w, \mbf{t}_i; \wt{x}_i \r) &\defeq 
\begin{cases}
\gamma_i \l( w, \mbf{t}_i; \wt{x}_i \r), & w \in \mcl{W}_2 \setminus \{\dagger\},\\
\gamma_i \l( T_{max} + 1, \mbf{t}_i; \wt{x}_i \r) & w = \dagger.\numberthis\label{eq:gammai2}
\end{cases}
\end{align*}
Then, combining Assumptions \ref{assump:factorizedfx0} and \ref{assump:tmax}, we can rewrite \eqref{eq:px0_3} as
\begin{align*}
    &\pr' \l( \mbf{x}^{(0)} | \mcl{O}_n \r) 
    \labelrel{\propto}{eqr:pxi0:1} 
    \sum_{  w, \mbf{t}, \mbf{s}} f_W \prod_{i\in V} f_{X_i^{(0)}} \zeta_i \gamma_i \psi_{\ra i} \prod_{\{i, j\}\in E } \psi_{i \xlrarrow{} j}\\
    &= \sum\limits_{\substack{ \mbf{t}, \mbf{s}} } \l( \sum\limits_{\substack{w \in [w_{min}, T_{max}]}} + \sum\limits_{\substack{w \in [T_{max} + 1, w_{max}]}}\r) \\
    &\hspace{20pt} \times f_W(w) 
    \prod_{i\in V} f_{X_i^{(0)}} \zeta_i \gamma_i \l(w, \mbf{t}_i; \wt{x}_i \r) \psi_{\ra i} \prod_{\{i, j\}\in E } \psi_{i \xlrarrow{} j}\\
    &\labelrel{=}{eqr:pxi0:2}\sum_{w, \mbf{t}, \mbf{s}} f'_W \prod_{i\in V} \zeta_i \gamma'_i \psi_{\ra i} \prod_{\{i, j\}\in E } \psi_{i \xlrarrow{} j}.\numberthis\label{eq:p'xi0}
\end{align*}
Here, \eqref{eqr:pxi0:1} follows from \eqref{eq:px0_3} and \eqref{eqr:pxi0:2} uses that fact that for any fixed $\wt{x}_i$ and $\mbf{t}_i \in (\mcl{T}')^2$, 
\begin{align*}
	\gamma_i\l(w, \mbf{t}_i; \wt{x}_i \r) &= \gamma_i\l(T_{max}+1, \mbf{t}_i; \wt{x}_i \r) \\ 
	&\hspace{-15pt} \text{for all } w \in w \in \{T_{max} + 1, T_{max} + 2, \dots, w_{max} \}.
\end{align*}

\subsection{Factor-Graph Representation}
We are interested in a suitable factor-graph representation of the summand in \eqref{eq:p'xi0}, i.e.,
\begin{align*}
\Gamma' \defeq f'_W\prod_{i\in V}f_{X_i^{(0)}} \zeta_i\gamma'_i\psi_{\ra i} \prod_{\{i,j\}\in E } \psi_{i\xlrarrow{} j}.\numberthis\label{eq:biggamma_2}
\end{align*}
Performing the factor-graph transformations listed in Section \ref{subsec:fg} except for the clustering of $x_i^{(0)}$'s, we obtain the below equivalent (inflated) version of $\Gamma'$,
\begin{align*}
&f'_{W}
\prod_{i\in V} 
f_i' \l( w, x_i^{(0)},\mbf{t}_i, \{\mbf{s}_{k\ra i} \}_{k\in\partial{i}} \r) 
\\
&\hspace{30pt} \times 
\prod_{\{i,j\}\in E } f_{ij}' \l( w, \mbf{t}_i,\mbf{s}_{i \ra j}, \mbf{t}_j,\mbf{s}_{j\ra i} \r),
\end{align*}
where,
\begin{align*}
f_i' \l( \cdot \r) \defeq  f_{X_i^{(0)}}
\zeta_i \gamma'_i \psi_{\ra i} 
\quad \text{and} \quad
f_{ij}' \l( \cdot \r) \defeq \psi_{i \xlrarrow{} j}.
\numberthis\label{eq:fiandfij_2}
\end{align*}
\begin{figure*}[!t]
\centering
\includegraphics[width=.8\linewidth]{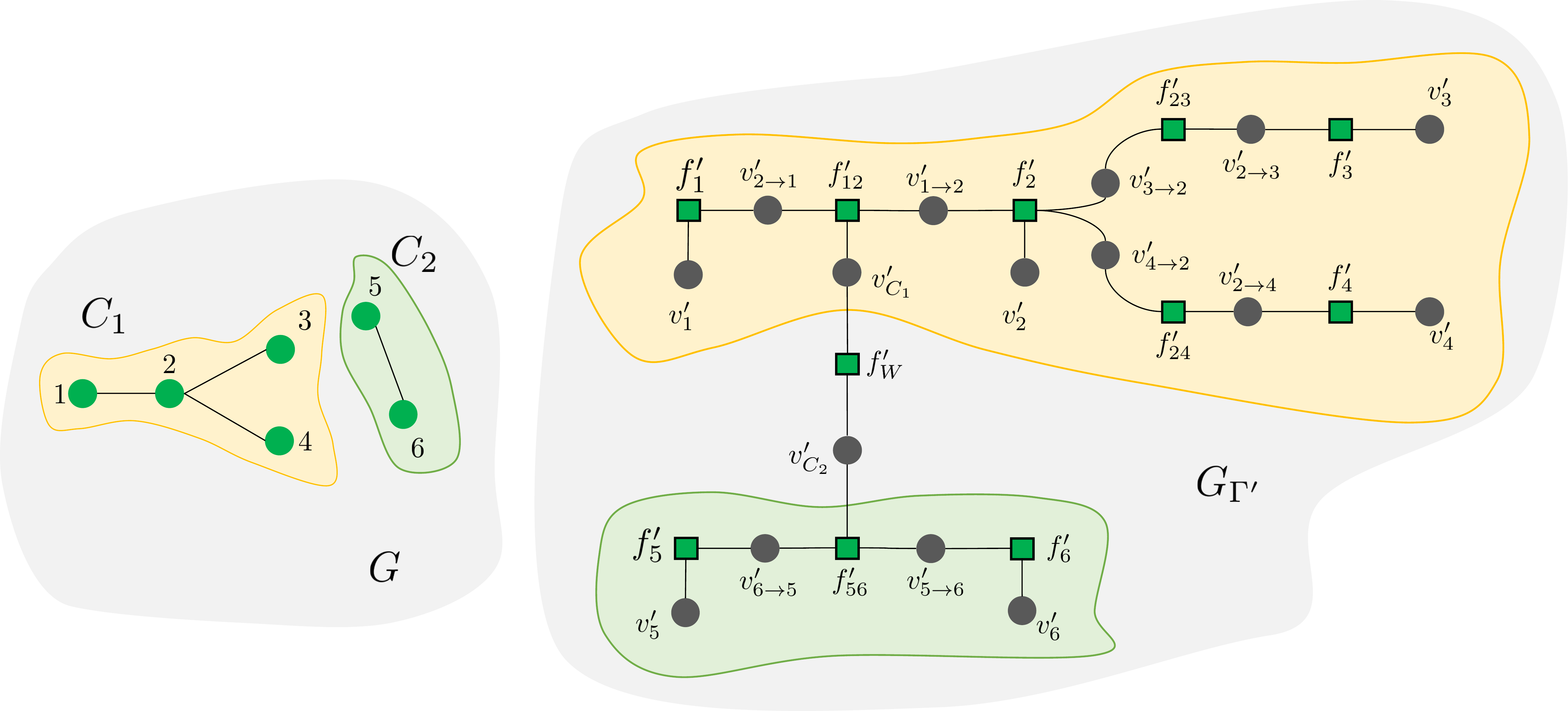}
\caption{Example of Generating Factor-Graph $G_{\Gamma'}$ from $G$. Here, like before, $\wh{e}^{(C_1)} = \{1,2\}$ and $\wh{e}^{(C_2)}=\{5, 6\}$.}
\label{fig:exampleggamma2}
\end{figure*}
Denoting the corresponding factor-graph by $G_{\Gamma'}$, we have,
\begin{align}\label{eq:gbiggamma_2}
\begin{split}
G_{\Gamma'} &\defeq\l(V_{\Gamma',1}\cup V_{\Gamma',2}, E_{\Gamma'}\r), \\
V_{\Gamma',1} &\defeq \l\{ v'_C:C\in\mcl{C} \r\}
\cup  \l\{ v'_i:i\in V \r\}
\cup \l\{ v_{j \ra i}': (j, i)\in \vv{E} \r\},\\
V_{\Gamma',2} &\defeq \l\{ f'_W \r\} \cup \l\{f_i':i\in V\r\} \cup \l\{ f'_{ij}: \{i,j\}\in E \r\}, \\
E_{\Gamma'} &\defeq \l\{
\l\{ \l\{v'_C ,f'_W\r\} \r\}_{C \in \mcl{C}}, 
\l\{ \l\{v'_C, f'_{\wh{e}^{(C)}}\r\}\r\}_{C \in \mcl{C}},  \r.\\
&\hspace{20pt} \l. \l\{ \l\{ f'_{ij}, v_{j\ra i}' \r\} , \l\{ f'_{ij}, v_{i\ra j}'\r\} \r\}_{\{i,j\}\in E}, \r.\\
&\hspace{20pt}\l. \l\{ \l\{v'_i, f_i'\r\}\r\}_{i\in V}, \l\{ \l\{ f_i', v_{j \ra i}' \r\} \r\}_{(j,i)\in\vv{E}} \r\}.
\end{split}
\end{align}
Here, $v'_C$ corresponds to a copy of $\{w\}$ for connected-component $C$, $v'_i$ corresponds to $\{x_i^{(0)}\}$, and each $v'_{j\ra i}$ corresponds to the variables $\{ w, \mbf{t}_i, \mbf{s}_{j\ra i}\}$. Also, like before, $\wh{e}^{(C)}$ is chosen arbitrarily from $E_C$. Figure \ref{fig:exampleggamma2} shows an example of generating factor-graph $G_{\Gamma'}$ for a given network $G$.

\subsection{Pseudo-Marginals}
Once again, when $G$ is a forest, $G_{\Gamma'}$ is also a tree, and one can recover the exact posterior probabilities, $\pr'(x_i^{(0)}|\mcl{O}_n)$ and $\pr'(t_i^I|\mcl{O}_n)$ from the pseudo-marginals listed below.
\begin{align}
\begin{split}\label{eq:posteriorbeliefs_2}
q'_{x_i^{(0)}} (x_i^{(0)}) &\propto m_{f'_i\ra v'_i} \l( x_i^{(0)} \r), \\
q'_{{t}_i^I}(t_i^I) &\propto \hspace{-5pt} \sum_{w, t_i^J, \mbf{s}_{j\ra i}} m_{f'_i\ra v'_{j\ra i}} 
\l(w, \mbf{t}_i, \mbf{s}_{j\ra i} \r)\\
&\hspace{40pt} \times 
m_{f'_{ij} \ra v'_{j\ra i}} 
\l(w, \mbf{t}_i, \mbf{s}_{j\ra i} \r).
\end{split}
\end{align}
In the last equation, $j$ is an arbitrarily chosen neighbor of node $i$ and $J$ denotes the (single) element of $\{ A, B\} \setminus \{I\} $.

As concerns the initial-state of the network, we no longer have a variable-node for $\mbf{x}^{(0)}$. Therefore, we must use some form of approximation to infer the most likely initial-state of the network. One well-known choice is the mean-field approximation,
\begin{align*}
\wh{\mbf{x}}^{(0)} \in \times_{i\in V} \argmax_{ x_i^{(0)}} q'_{x_i^{(0)}} \l( x_i^{(0)} \r).\numberthis\label{eq:x0meanfield2}
\end{align*}

\subsection{Message Update Rules}\label{sec:inference_2:updaterules}
\begin{figure*}[!t]
\centering
\subfloat
{\includegraphics[width=.9\linewidth]{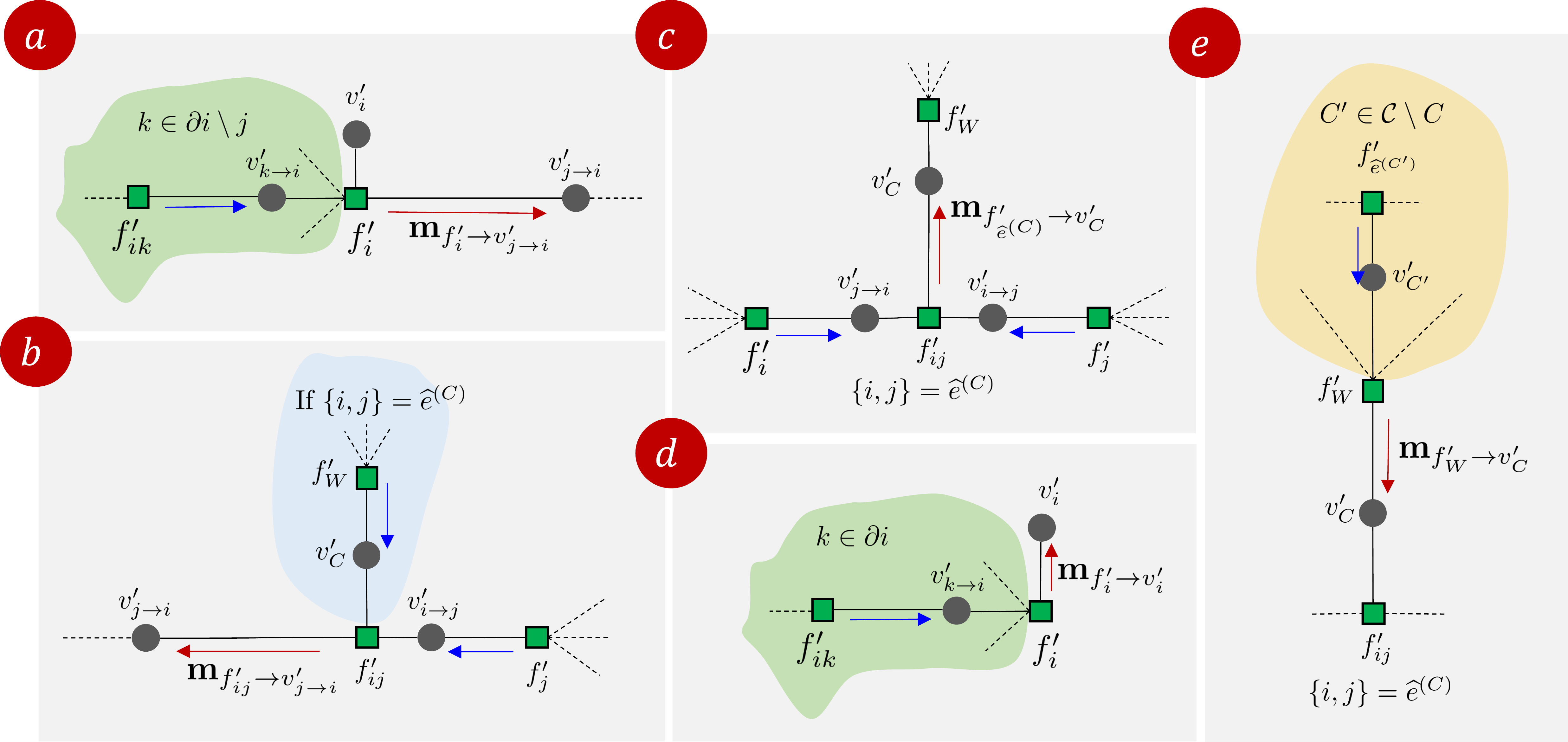}}
\caption{Types of Messages in $G_{\Gamma'}$. In each subfigure, the red arrow indicates a message and the blue arrows indicate the messages used in transmitting it.}
\label{fig:messages2}
\end{figure*}
In $G_{\Gamma^{\prime}}$, there are five kinds of messages (see \figurename \ref{fig:messages2}).
\begin{enumerate}
    \item[a)] Messages from (factor-node) $f'_i$ to (variable-node) $v'_{j\ra i}$.
    \item[b)] Messages from $f'_{ij}$ to $v'_{j\ra i}$ and $v'_{i\ra j}$.
    \item[c)] Messages from $f'_{\wh{e}^{(C)}\ra v'_C } $.
    \item[d)] Messages from $f'_i$ to $v'_i$.
    \item[e)] Messages from $f'_{w}$ to $v'_C$.
\end{enumerate}
These messages require a memory of $O\l(|E|T_{max}^3\r)$ and their efficient update rules (obtained via distributive-law) are derived in Appendix \ref{app:messageupdaterules2}. Algorithm \ref{alg:sumproductbp2} shows the pseudocode for the BP algorithm based on $\Gamma'$. The (worst-case) per-iteration-time complexity of Algorithm \ref{alg:sumproductbp2} is $O\l(\frac{|E| T_{max}^3 \l( T_{max}^2  \vee \max_{i\in V} |\partial{i}|  \r)  }{\text{\# of worker-nodes}} \r)$ (see the end of Appendix \ref{app:messageupdaterules2}).
{
	\setlength{\algomargin}{1.5em}
	\begin{algorithm2e}[!hbt]
		\DontPrintSemicolon
		\KwInput{$G$, $\bsl{\lambda}$, $f_W$, $\l\{ f_{X_i^{(0)}}, f_{\wt{X}_i | X_i }   \r\}_{i\in V}$, $\mcl{O}_n$.}
		\Parameter{$T_{max}$, $max\_iters$.}
		\KwOutput{$\l\{ \wh{x_i}^{(0)}\r\}_{i\in V} $, $\wh{w}$, $\l\{\wh{t}_i^A, \wh{t}_i^B \r\}_{i\in V}$, $\wh{\mbf{x}}^{(0)}$.
		}
		
		\textbf{Factor-Graph}: Derive $G_{\Gamma'}$ from $G$. See \eqref{eq:gbiggamma_2}.
		
		Set $t \la 0$.
		
		\textbf{Initialization}: Initialize each factor-to-variable node message to all-ones vector.
		
		\Repeat{\nonl convergence OR $t>max\_iters$.}{
			Update all factor-to-variable node messages. See Appendix \ref{app:messageupdaterules2}.
			
			$t\la t+1$.}
		
		Set
		\begin{align*}
			\wh{x_i}^{(0)} &\in \argmax_{x_i^{(0)} \in \mcl{X} } q'_{x_i^{(0)}} \l(x_i^{(0)}\r),\\
			\wh{t}_i^I &\in \argmax_{t_i^I\in \mcl{T}'} q'_{t_i^I} \l(t_i^I\r),\\
			\wh{\mbf{x}}^{(0)} &= \times_{i\in V} \l\{ \wh{x}_i^{(0)}\r\},
		\end{align*}
		where $q'_{x_i^{(0)}}$ and $q'_{t_i^I}$ are given by \eqref{eq:posteriorbeliefs_2}.
		\caption{BP on $G_{\Gamma'}$.}\label{alg:sumproductbp2}
	\end{algorithm2e}
}

\noindent \textbf{Remark}: The per-iteration-time of Algorithm-\ref{alg:sumproductbp2} is dominated by either messages $\mbf{m}_{f'_i \ra v'_{j\ra i}} $ or $\mbf{m}_{f'_{ij} \ra v'_{j\ra i}} $ depending on whether the maximum node-degree of the graph $G$ is larger or smaller than $T_{max}$ (see Appendix \ref{app:messageupdaterules2}). In many networks of practical interest, the maximum node degree can be quite large, thus making each iteration of Algorithm-\ref{alg:sumproductbp2} computationally expensive. One way to avoid this issue is by randomly sampling neighbors of high-degree nodes during each iteration. Specifically, if a node $i\in V$ has more than $T_{max}$ neighbors, then, during each iteration, the message $\mbf{m}_{f'_i \ra v'_{j\ra i}} $ is computed for only $T_{max}$ of its neighbors which are chosen randomly. This is similar in spirit to stochastic belief propagation, \cite{noorshams-2011, haddadpour-2016}. Ignoring the sampling computation costs, the resulting complexity of Algorithm-\ref{alg:sumproductbp2} is then $O\l(\frac{|E| T_{max}^5 }{\text{\# of worker-nodes}} \r)$.
\section{Ensuring Convergence of Message-Passing}\label{sec:convergence}
As mentioned earlier, loopy BP converges to exact marginals when the factor-graph it operates on is a tree. When the factor-graph contains cycles, its convergence and quality of approximation are less well-understood.\footnote{Despite that, loopy BP has gained extensive popularity in the AI community -- due at least to its state-of-the-art results in error-correcting codes, computer-vision, and image-processing applications. The factor-graphs in these settings also contain numerous cycles.} However, it is well-known that when loopy BP fails to converge, it is due to the presence of cycles in the factor-graph and strongly correlated variables in the global-function. Together these two give rise to strong message-passing feedback loops which cause loopy BP to fail. To ward off the effect of these feedback loops, one can derive an easy-to-implement\footnote{There are variants of loopy BP with some convergence guarantees such as TRW \cite{wainwright03a} and norm-product BP \cite{hazan10}. These algorithms are inherently complex and not amenable to inference problems on large networks.} variant of loopy BP where, during the update of any message, the incoming messages are appropriately discounted, say by a factor $\eta \in [0,1]$.\footnote{In general, one may discount each message by a separate discount factor.  For simplicity, we have avoided that.}. This changes the update rule \eqref{eq:f2v_2} to
\begin{align*}
&m_{f\ra v}^{(k+1)} \l(\mbf{z}_{r(v)}\r) \\
&\hspace{5pt} \propto \hspace{-4pt} \sum_{\mbf{z}_{r(f)\setminus r(v)}} \hspace{-4pt} f\l(\mbf{z}_{r(f)}\r)
\hspace{-4pt} \prod_{v'\in N(f)\setminus v }
\prod_{g\in N(v')\setminus f } \l( m_{g\ra v'}^{(k)} \l(\mbf{z}_{r(v')}\r) \r)^{\eta}
.\label{eq:f2v_2_discounted}\numberthis	
\end{align*}
We refer to the above scheme as \emph{discounted BP} and for a given choice of $\eta$ denote it by BP($\eta$). 

When all the local-functions are strictly positive, there exists a threshold $\eta^* \in (0, 1]$ below which this variant is guaranteed to converge (see \cite{anna21} for details). However, the local-functions in $G_{\Gamma'}$ are far from being strictly positive, i.e., $f'_i$ and $f'_{ij}$ take zero values on many elements of their local-domains. A sufficient condition for convergence of BP on standard factor-graphs that have non-negative local-functions is established in \cite{mooij07}. This condition is based on a \emph{factor-strength matrix} which, for an inflated factor-graph, can be defined as follows:
\begin{dfn}[Factor-Strength Matrix]\label{dfn:factor_strength_matrix}
	Let $G_{\Gamma} = ( V_{\Gamma_1}, V_{\Gamma_2}, \vv{E}_{\Gamma} ) $ denote a factor-graph where $V_{\Gamma_1}$ is the set of variable-nodes, $V_{\Gamma_2}$ is the set of factor-nodes, and $\vv{E}_{\Gamma}$ is the set of directed edges pointing from nodes in $V_{\Gamma_2}$ to nodes in $V_{\Gamma_1}$. The factor-strength matrix of $G_{\Gamma}$ denoted by $\mbf{M} = \mbf{M}\l( G_{\Gamma} \r)$ is a square matrix 
	with rows and columns indexed by the edges in $\vv{E}_{\Gamma}$ and its typical entry $M_{f\ra v, g\ra u}$, corresponding to the edges $\l(f,v\r)$ and $\l(g,u\r)$, is given by
	\begin{align*}
		&M_{f \ra v, g \ra u} \\
		&\hspace{10pt} \defeq \11\l[ u\in N(f)\setminus v, g \in N(u)\setminus f \r] M\l(f, v, u \r),\\
		&M\l(f, v, u \r) \\
		&\hspace{10pt} \defeq 
		\max\limits_{ \mbf{z}_{r(v)} \ne \mbf{z}'_{r(v)} } \hspace{10pt}
		\max\limits_{ \substack{ 
				\mbf{z}_{ r(u) \setminus r(v)} \\
				\ne \mbf{z}'_{r(u)\setminus r(v) }}} \hspace{10pt} 
		\max\limits_{\substack{ 
				\mbf{z}_{ r(f) \setminus r({u,v})}, \\ \mbf{z}_{r(f) \setminus r({u, v})}}} \\
			&\hspace{20pt} \frac{M_1\l( \mbf{z}_{r(f)}, \mbf{z}'_{r(f)} ; f, v, u \r) - M_2\l( \mbf{z}_{r(f)}, \mbf{z}'_{r(f)} ; f, v, u \r) }{M_1\l( \mbf{z}_{r(f)}, \mbf{z}'_{r(f)} ; f, v, u \r) + M_2\l( \mbf{z}_{r(f)}, \mbf{z}'_{r(f)} ; f, v, u \r)},\\
		&M_1\l( \mbf{z}_{r(f)}, \mbf{z}'_{r(f)} ; f, v, u \r) \\
		&\hspace{10pt} \defeq \sqrt{f\l( \mbf{z}_{r(v)}, \mbf{z}_{r(u)\setminus r(v)}, \mbf{z}_{r(f) \setminus r({u, v})} \r)} \\
		&\hspace{30pt} \times \sqrt{f\l( \mbf{z}'_{r(v)}, \mbf{z}'_{r(u)\setminus r(v)}, \mbf{z}'_{r(f) \setminus r({u, v})} \r)},\\
		&M_2\l( \mbf{z}_{r(f)}, \mbf{z}'_{r(f)} ; f, v, u \r) \\
		&\hspace{10pt} \defeq \sqrt{f\l( \mbf{z}'_{r(v)}, \mbf{z}_{r(u)\setminus r(v)}, \mbf{z}_{r(f) \setminus r({u, v})} \r) } \\
		&\hspace{30pt} \times \sqrt{f\l( \mbf{z}_{r(v)}, \mbf{z}'_{r(u)\setminus r(v)}, \mbf{z}'_{r(f) \setminus r({u, v})} \r)}.\numberthis\label{eq:mfvu}
	\end{align*}
	Here, $r(u,v)\defeq r(u)\cup r(v)$ and the quantity $M\l( f; v, u\r)$ is referred to as the strength of factor-node $f$ with respect to variable-nodes $u$ and $v$. Note that $M\l( f; v, u\r)$  is symmetric, i.e., $M\l( f; v, u\r) = M\l( f; u, v\r)$. \figurename \ref{fig:factor_strength} illustrates the situation when a specific entry in $\mbf{M}$ is (possibly) non-zero.
\end{dfn}
\noindent \textbf{Remark}: The factor-strength matrix introduced in \cite{mooij07} is applicable to standard representation of factor-graphs (in which each variable-node represents a single distinct variable). The above definition is a generalization of it to inflated factor graphs where a variable-node may represent a cluster of variables.
\begin{figure}[!t]
	\centering
	\includegraphics[width=.8\linewidth]{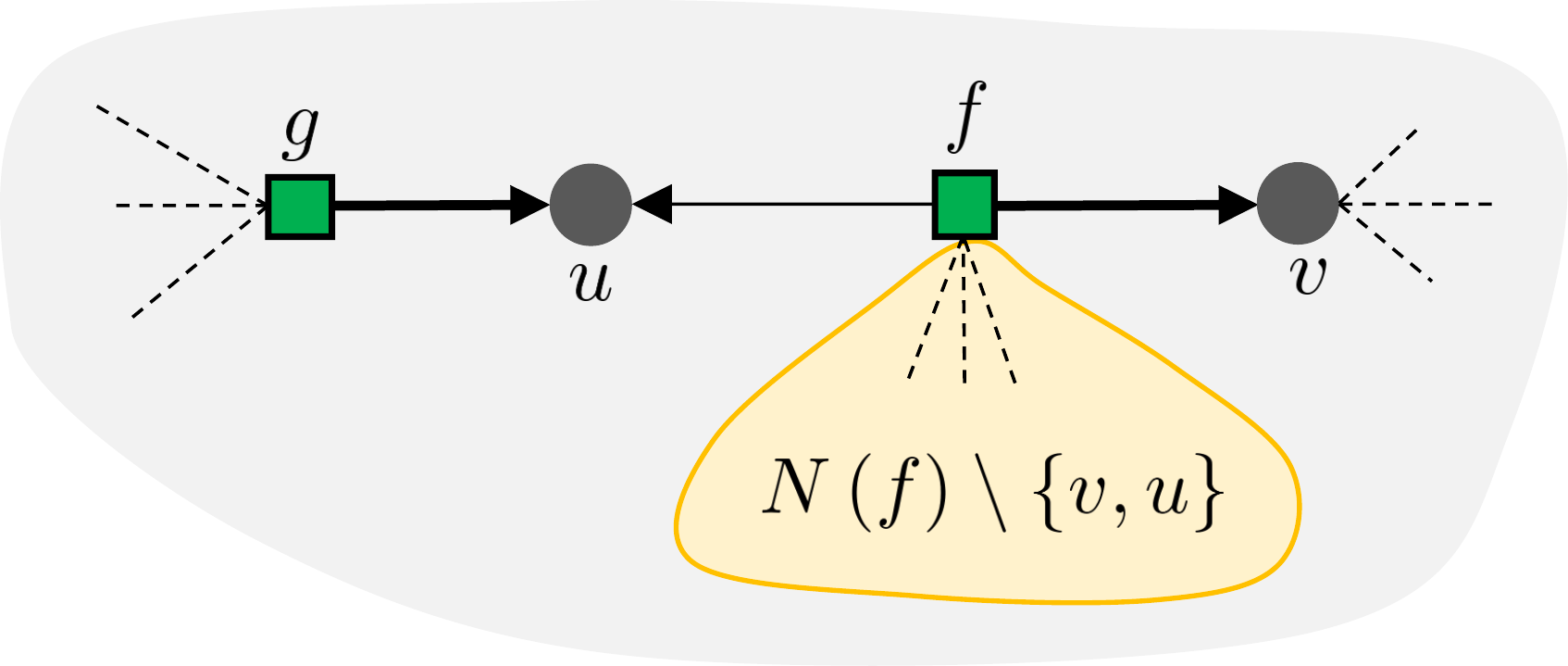}
	\caption{Part of the factor-graph relevant in expression . Here, $u,v\in V_{\Gamma_1}$, $f,g\in V_{\Gamma_2}$, with $u\ne v$ and $g \in N(u) \setminus f$. }
	\label{fig:factor_strength}
\end{figure}

The following proposition is an extension of \cite{mooij07}[Theorem 5] for BP($\eta$) on inflated factor-graphs.
\begin{prop}\label{prop:factor_strength}
	Let $G_{\Gamma} = ( V_{\Gamma_1} \bigsqcup V_{\Gamma_2}, \vv{E}_{\Gamma} ) $ denote a factor-graph where $V_{\Gamma_1}$ is the set of variable-nodes, $V_{\Gamma_2}$ is the set of factor-nodes, and $\vv{E}_{\Gamma}$ is the set of directed edges pointing from nodes in $V_{\Gamma_2}$ to nodes in $V_{\Gamma_1}$. Suppose that the local-function corresponding to each factor-node $f$ is non-negative and let $\mbf{M}$ be the factor-strength matrix (see Definition \ref{dfn:factor_strength_matrix}) of $G_{\Gamma}$. If the spectral-radius of $\mbf{M}$ is less than $\frac{1}{\eta}$, then BP($\eta$) on $G_{\Gamma}$ converges to a unique fixed point irrespective of the initial (positive) messages.
\end{prop}
\begin{proof}
	The proof can be established by extending the work of \cite{mooij07} to BP($\eta$) on inflated factor-graphs (with non-negative local-functions). We inform the reader beforehand that \cite{mooij07}[Theorem 5] requires that 
	$$ \text{for every } (f, v) \in \vv{E}_{\Gamma} \text{ and } \mbf{z}_{r(v)}, \sum_{ \mbf{z}_{r(f)\setminus r(v)} } f\l( \mbf{z}_{r(f)} \r) > 0.$$ 
	This condition is benign and not really needed since whenever it is violated, it necessarily implies that the marginal probability of $\mbf{z}_{r(u)}$ is zero and one can remove it from the set of joint realizations of the variables contained in $r(u)$.
\end{proof}
In light of Proposition \ref{prop:factor_strength}, in our message-passing schemes of Algorithms \ref{alg:sumproductbp} and \ref{alg:sumproductbp2}, we may greedily start with $\eta=1$ (BP) and keep decreasing it using fixed steps $d\eta$ till the message-passing achieves convergence. 

This handles the issue of non-convergence of BP on graphs of arbitrary topology. Table \ref{table:comparison:algs} provides a comparison of Algorithms \ref{alg:sumproductbp} and \ref{alg:sumproductbp2}.


\begin{table*}[!ht]
\caption{Comparison of Algorithms \ref{alg:sumproductbp} and \ref{alg:sumproductbp2}.}
\label{table:comparison:algs}
\centering
\begin{tabular}{|c|c|c|c|c|}
\hline
\multirow{2}{*}{\textbf{Alg.}} &\textbf{Output}
&\multirow{2}{*}{\textbf{BP(1) Exact on Forests?}}
&\multirow{2}{*}{\textbf{Memory}}
&\multirow{2}{*}{\textbf{Per-iteration-time}}\\ [0pt]
&\textbf{Pseudo-Marginals} & & &\\ [0pt] \hline
1 &$q_{\mbf{x}^{(0)}}, \l\{q_{x_i^{(0)}}, q_{t_i^A}, q_{t_i^B}\r\}_{i\in V}, q_w$ &Yes &$O\l(| {\mcl{W}_1} | |\bsl{\mcl{X}}^{(0)}||E||V|^2 \r)$ &$O\l( \frac{|{\mcl{W}_1}| |\bsl{\mcl{X}}^{(0)}||E||V|^4}{\# \text{ of worker nodes}}) \r)$\\ [0pt] \hline
2 &$\l\{q'_{x_i^{(0)}}, q'_{t_i^A}, q'_{t_i^B}\r\}_{i\in V}, q'_w$ &Under Assumptions \ref{assump:factorizedfx0} and \ref{assump:tmax} &$O\l(|E|T_{max}^3 \r)$ &$O\l(\frac{|E|T_{max}^3 \l( T_{max}^2 \vee \max_{i\in V} |\partial{i}| \r) }{\# \text{ of worker nodes}} \r)$\\ [0pt] \hline
\end{tabular}
\end{table*}
\section{Discussion}\label{sec:discussion}
\subsection{Additional Information-Structures}
The marginalization framework based on variables $W, \mbf{T}, \mbf{S}$ enables us to make use of additional information that is helpful for further reduction in complexity of the message-passing algorithms. Consider an enriched observation model where in addition to the snapshot $\mcl{O}_n$, there are three additional information-structures about $\mbf{T}$, $\mbf{S}$, and $W$, namely
\begin{align*}
    \mcl{I}_1 &\defeq \cap_{i\in V} \l\{T_i^I \in \udl{\mcl{T}}_i^I, I\in\{A,B\} \r\},\numberthis\label{eq:i1}\\
    \mcl{I}_2 &\defeq \cap_{(i, j)\in \vv{E} }  \l\{ S_{i\ra j}^I \in \udl{\mcl{S}}_{i\ra j} \r\},\numberthis\label{eq:i2}\\
    \mcl{I}_3 &\defeq \l\{ W \in \udl{\mcl{W}} \r\}.
    \numberthis\label{eq:i3}
\end{align*}
Here, $\udl{\mcl{T}}_i^I$ $\udl{\mcl{S}}_{i\ra j}$, and $\udl{\mcl{W}}$ are subsets of $\mcl{T}$, $\mcl{S}$, and $\mcl{W}$ respectively. 

We can perform MAP estimation of the posterior probability $\pr' \l(\mbf{x}^{(0)} | \mcl{O}_n, \mcl{I}_1, \mcl{I}_2, \mcl{I}_3 \r)$ by restricting $\mbf{t}$, $\mbf{s}$, and $w$ to assume values from their corresponding sets stated in $\mcl{I}_1$, $\mcl{I}_2$, and $\mcl{I}_3$ respectively. Then, richer the information contained in $\mcl{I}_j$'s, lesser the wall-clock time of the message-passing algorithm.

\noindent \textbf{Remark}: Here, we assumed that 
$\mcl{I}_1, \mcl{I}_2, \mcl{I}_3$ and $\mcl{O}_n$ are consistent with each other, i.e., the event $\mcl{O}_n \cap \mcl{I}_1 \cap \mcl{I}_2 \cap \mcl{I}_3$ is an event of positive-measure. This assumption is an important one and worth highlighting. In the case of inconsistency, the message-passing algorithm will converge to the all-zeros vector after a certain number of iterations. The determination of which information structures are consistent with the observation-snapshot is out of the scope of this work.

\subsection{Spread Estimation}\label{sec:disc:influencemaximization}
The spread of process $I \in \{A, B\}$ achieved by an initial-state $ \mbf{x}^{(0)}\in \bm{\mcl{X}}^{(0)} $, denoted here by $\text{Spread}(I, \mbf{x}^{(0)}) $, is defined as the number of nodes that get infected with process $I$ in some finite time-horizon $H \in \mb{Z}_{\ge 0}$. 

By setting $f_W(\cdot) = \delta(0, \cdot) $,  $f_{\mbf{X}^{(0)}} (\cdot) = \prod_{i\in V} \delta({x}_i^{(0)}, \cdot) $, $f_{\wt{X}_i | X_i}( \cdot | *) = \delta ( \cdot, * )$\footnote{This ensures that the observation-event $\mcl{O}_n$ is trivial.}, one can study the spread-estimation problem in the setting of two diffusion processes. Here, the marginal probabilities of nodes' infection-times obtained by 
Algorithm \ref{alg:sumproductbp2} 
can be used to estimate a process's expected spread as follows:
\begin{align*}
	\mb{E}[ \text{Spread}(I, \mbf{x}^{(0)}) ] 
	&= \mb{E} \l[ \sum_{i\in V}  \11[ T_i^{I} \le H] \r]\\
	&= \mb{E} \l[ \sum_{i\in V}  \11[ T_i^{I} \le H] \bigg| \mcl{O}_n \r]\\
	&= \sum_{i\in V} \mb{P} (T_i^I \le H | \mcl{O}_n) \\
	&= \sum_{i \in V} \sum_{t_i^I \le H} \mb{P}(t_i^I | \mcl{O}_n) \\
	&\approx \sum_{i \in V} \sum_{t_i^I \le H} q'(t_i^I).
\end{align*}

\subsection{The Case of Contagion and Anti-contagion Processes}\label{sec:disc:anticontagion}
The setting when the two processes do not initiate at the same time\footnote{For instance in online social networks, a rumor/disinformation may start spreading earlier than anti-rumor/information.} can be studied by slight modifications of the model described in Section \ref{sec:problem}. Specifically, let $A$ and $B$ respectively denote contagion and anti-contagion processes and assume that process $B$ may initiate only from a subset $V_p$ of nodes called the \emph{protector-nodes}. A protector-node $i\in V_p$ is defined as a node that satisfies the following three characteristics.
\begin{itemize}
\item Its initial-state is either $\emptyset$ or $\{B\}$, i.e., $X_i^{(0)} \in \l\{ \emptyset, \l\{B\r\} \r\}$.
\item It is immune to process $A$, i.e., for every $k\in\partial i$, \protect\linebreak $\lambda_{k\ra i}^{A\ra \emptyset} = \lambda_{k\ra i}^{A\ra \{B\}} = 0$.
\item It starts spreading process $B$ either as a result of some random time $T_{delay}^B$ or after observing at least one of its neighbors in an infected state, i.e.,
$$ T_i^B = \min \l\{ T_{delay}^B, \inf_{k\in\partial{i}} \l\{ T_k^A, T_k^B \r\} \r\}, \qquad i\in V_p.$$
\end{itemize}
We refer to $T_{delay}^B$ as the \emph{onset-delay} of process $B$ and assume that it has an independent distribution given by the kernel,
\begin{align*}
&f_{T_{delay}^B}\l( h \r) > 0 \text{ iff } \\
&\hspace{10pt} h \in \mcl{T}_{delay}^B \defeq \l\{ h_{min}, h_{min}+1, \dots, h_{max}, \infty \r\}.\numberthis\label{eq:fh}
\end{align*}
It is worth noting that $T_{delay}^B$ may assume negative values which would correspond to preemptive spreading of process $B$. This changes the random process $\l\{ \mbf{X}^{(t)}: t\ge 0 \r\}$ to $\l\{ \mbf{X}^{(t)}: t\ge T_{delay}^B \wedge 0 \r\}$ but \eqref{eq:xit} still remains valid.

\noindent \textbf{Remark}: By restricting $T_{delay}^B$ to $\infty$, one ignores the preemptive spreads of the anti-contagion and recovers a model similar to that studied in \cite{choi19}.

\subsection{Choice of $T_{max}$}
Assumption \ref{assump:tmax} does not hold for the original diffusion-model in Section \ref{sec:problem}. However, one may still apply Algorithm \ref{alg:sumproductbp2} to snapshots obtained from it. To do so, one must choose $T_{max}$ that is at least consistent with the observed-snapshot $\mcl{O}_n$. Assuming that all infection probabilities are strictly positive, a sufficient condition for satisfying this consistency requirement is:
\begin{align*}
T_{max} &\ge \max \limits_{C\in \mcl{C}} \l\{ d_C \r\}.
\end{align*}
Here, we have used $d_C$ to denote the diameter of connected-component $C$. It is also clear that the below choice of $T_{max}$ is without loss of any sample-path under $\pr$ and therefore will yield exact output posterior marginals whenever the corresponding factor-graph is a tree.
\begin{align*}
T_{max} = \max\limits_{C\in\mcl{C}} \l\{ |V_C|-1 \r\}.
\end{align*}

\subsection{Unique Sources with Uniform Prior Probabilities}
The case of unique source for each process does not satisfy Assumption \ref{assump:factorizedfx0}. However, for large networks, to use Algorithm \ref{alg:sumproductbp2}, a very good approximation is to use:
\begin{align*}
f_{X_i^{(0)}}(x_i^{(0)}) &=
\begin{cases}
\l(1-\frac{1}{|V|}\r)^2 &\text{if } x_i^{(0)} = \emptyset,\\
\frac{1}{|V|}\l(1-\frac{1}{|V|}\r) &\text{if } x_i^{(0)} = \{A\} \text{ or } \{B\}, \\
\frac{1}{|V|^2} &\text{if } x_i^{(0)} = \{A,B\}.
\end{cases}\numberthis\label{eq:fxi0-uniquesources}
\end{align*}

\subsection{The Case of Single Independent-Cascade Process}
If we restrict each $X_i^{(0)}$ to $\{ \emptyset,\{A\} \}$, the infection-time of each node for process-$B$ is trivially infinity. In this case, all the random variables related to process $B$, i.e., $T_i^B, D_{i\ra j}^{A\ra\{B\}},D_{i\ra j}^{B\ra\{A\}},D_{i\ra j}^{B\ra\emptyset}, S_{i\ra j}^B$ are redundant and can be removed. One then obtains the single-spread IC model of \cite{zhu16-1} and \cite{kempe-2003}.

\subsection{The Case of Multiple Independent-Cascade Process} 
Algorithms \ref{alg:sumproductbp2} can be extended to multiple independent-cascade processes. If the number of processes is $l$, then the storage complexity will be $O\l( 2^{l} |E| T_{max}^{l+1}\r)$ since variable-nodes $v'_{j\ra i}$ will include $l$ infection-time variables and $l$ relative-timing variables. Similarly, one can verify that the per-iteration-time complexity will be $O\l( 4^l |E| T_{max}^{l+1} \l( T_{max}^l \vee \max_{i\in V} |\partial{i}| \r)  \r)$. For removing dependence on the maximum degree, please see the remark at the end of Section \ref{sec:inference_2:updaterules}.
\section{Conclusion}\label{sec:conclusion}
In this work, we explored two estimation problems based on the concurrent spread of two diffusion processes on a network. To this end, we formulated an extension of the IC-based model where a node's susceptibility to a process depends on its current state. After deriving the joint probability of the network's initial-state and the observation-snapshot in terms of infection-times of the nodes, we used the distributive law and factor-graph transformations to derive a variant of the BP algorithm that is feasible on large networks and convergent on graphs of arbitrary topology. 

The main motivation behind the proposed message-passing algorithm is that both \emph{Rumor-Centrality} and \emph{Short-Fat-Tree} algorithms are essentially graph-based algorithms and therefore, do not consider the variables that govern the spreading dynamics of the diffusion process. They also do not have natural extensions to settings that involve multiple spreading processes or multiple sources of a single process. 
In contrast, message-passing can take into account the spreading dynamics and can also be generalized to multiple diffusion processes. Furthermore, it does not require symmetry of interactions and can incorporate multiple settings of interest. 

Work on empirical evaluation of the proposed message-passing algorithm for both the DSL and spread-estimation problems is an important direction and is in progress.

\section*{Acknowledgment}
This work was funded by NSF via grants ECCS2038416, EPCN1608361, EARS1516075, CNS1955777, CCF2008130, and CMMI2240981 for V. Subramanian, and grants EARS1516075, CNS1955777, CCF2008130, and CMMI2240981 for N. Khan.

\ifCLASSOPTIONcaptionsoff
  \newpage
\fi



\bibliographystyle{IEEEtran}
\bibliography{backandforinference_ieee_arxiv}
\appendices


\renewcommand{\theequation}{\thesection.\arabic{equation}}

\section{List of Symbols}\label{sec:appendix:notation}
\begin{itemize}
\setlength\itemsep{2pt}
\item IC: Independent-Cascade
\item DSL: Diffusion Source Localization.
\item IM: Influence Maximization.
\item RC: Rumor-Centrality.
\item MLE: Maximum Likelihood Estimator.
\item SFT: Short-Fat-Tree algorithm.
\item WBND: Weight Bounded Node Degree.
\item BFS: Breadth First Search.
\item DMP: Dynamic-Message-Passing algorithm.
\item BP: Belief Propagation algorithm.
\item WLOG: Without loss of generality.
\item DL: Distributive-law.
\item LTP: Law of Total Probability.
\item $G$:  A static directed graph with $V$ and $\vv{E}$ as the sets of nodes and directed edges respectively. WLOG, it is assumed that whenever $(i,j)$ is in $\vv{E}$, $(j, i)$ is also in $\vv{E}$. $E$ denotes the set of undirected edges generated by $\vv{E}$. Each connected-component of $G$ is assumed to contain at least two nodes.
\item $\partial{i}$: Set of neighbors of node $i$ in graph $G$, i.e., $\{ k \in V: \{k, i\} \in E \} $.
\item $\mcl{C}$: Set of connected-components of $G$. Each $C \in \mcl{C}$ is assumed to have at least two nodes.
\item $I$ and $J$: Typically used to denote diffusion processes $A$ and $B$.
\item $\bsl{\lambda}_{i\ra j}$: $ \l\{ 
\lambda_{i\ra j}^{A \ra \emptyset}, \lambda_{i\ra j}^{A \ra \{B\}},
\lambda_{i\ra j}^{B \ra \emptyset}, \lambda_{i\ra j}^{B \ra \{A\}}
\r\}$.
\item $\bsl{\lambda}$: $\l( \bsl{\lambda}_{i\ra j}: (i,j)\in \vv{E} \r)$.
\item $X_i^{(t)}$: State of node $i$ at time $t$. Also, see \eqref{eq:xit}.
\item $\mcl{X} $: $ \{ \emptyset, \{A\}, \{B\}, \{A, B\} \}$.
\item $ \mbf{X}^{(t)}$: State of network at time $t$, i.e., $\l( X_i^{(t)} : t \ge 0 \r) $.
\item $D_{i\ra j}^{I\ra \mcl{J}}$: Activation-time variable - time taken by node $i$ to infect node $j$ with process $I$ when node $j$ has state $\mcl{J} \in \{ \emptyset, \{A,B\}\setminus\{I\} \}$ at the time of infection attempt.
\item $f_{D_{i\ra j}^{I\ra \mcl{J}} }$: See \eqref{eq:fdijij}.
\item $\mbf{D}_{i\ra j}$: $\l\{ 
D_{i\ra j}^{A \ra \emptyset}, D_{i\ra j}^{A \ra \{B\}},
D_{i\ra j}^{B \ra \emptyset}, D_{i\ra j}^{B \ra \{A\}}
\r\}$.
\item $\mbf{D}$: $\l( \mbf{D}_{i\ra j} :(i,j)\in \vv{E} \r)$.
\item $W$: Observation-time.
\item $\mcl{O}_n$: Observation-event $\{ \wt{\mbf{X}}^{(W)} = \wt{\mbf{x}} \} $. The subscript $n$ indicates that the observation is noisy.
\item $f_{\wt{X}_i| X_i} $: Probability of observing node $i$ in state $\wt{x}_i$ given its (true) state is $x_i$.
\item $f_{\wt{\mbf{X}} | \mbf{X}} $: See \eqref{eq:fo}.
\item $f_W$: See \eqref{eq:fw}
\item $\mcl{W}$: $\{ w_{min}, w_{min}+1, \dots, w_{max} \} $.
\item $f_{\mbf{X}^{(0)} } $: Distribution of network's initial-state. See \eqref{eq:fx0}.
\item $f_{\mbf{D}_{i\ra j}} \l( \mbf{d}_{i\ra j} \r)$: See \eqref{eq:fdij}.
\item $T_i^I$: Time at which node $i$ gets infected with process $I$. See \eqref{eq:tii} and \eqref{eq:tii2}.
\item $\mcl{T}$: $\{0,1,\dots, |V|-1, \infty \} $.
\item $\mbf{T}_i$: $ \l\{T_i^A, T_i^B\r\}$.
\item $\mbf{T}$: $\l(\mbf{T}_i:i\in V\r)$.
\item $\gamma_i^p $: See \eqref{eq:gammaip}. 
\item $T_{k\ra i}^I$: Propagation-time of process $I$ from $k$ to (its neighbor) $i$. See \eqref{eq:tkii}.
\item $f_{\mbf{T}_i} $: See \eqref{eq:fti} and \eqref{eq:fti2}.
\item $\gamma_i$: See \eqref{eq:po}.
\item $\sigma(b_i; a)$: See \eqref{eq:sigmadef}.
\item $\zeta_i$: See \eqref{eq:zetai}.
\item $\psi_{\ra i} $: See \eqref{eq:psii}.
\item $\psi_{\ra i}^I $: See \eqref{eq:psiii1} and \eqref{eq:psiii2}.
\item $S_{i\ra j}^I$: See \eqref{eq:siji}.
\item $\mcl{S}$: $ \{ 0, 1\}$.
\item $\mbf{S}_{i\ra j}$: $\{S_{i\ra j}^A, S_{i\ra j}^B\}$.
\item $\mbf{S}$: $\l(\mbf{S}_{i\ra j} : (i,j)\in \vv{E} \r)$.
\item $\psi_{i\xlrarrow{} j} $: See \eqref{eq:psiijji}.
\item $\psi_{i\ra j} $: See \eqref{eq:psiij}.
\item $\psi_{i\ra j}^I$: See \eqref{eq:psiiji}.
\item $\udl{f}_W$: See \eqref{eq:fw1}.
\item $\udl{\gamma}_i $: See \eqref{eq:gammai1}.
\item $\Gamma$: See \eqref{eq:biggamma}.
\item $f_i, f_{ij} $: See \eqref{eq:fiandfij}.
\item $G_{\Gamma}$: Inflated factor-graph representation of $\Gamma$ given by \eqref{eq:gbiggamma}.
\item $m_{f\ra v}^{(k+1)} \l(\mbf{z}_{r(v)}\r)$: See \eqref{eq:f2v_1} and \eqref{eq:f2v_2}.
\item $n_{v \ra f}^{(k)}\l(\mbf{z}_{r(v)}\r)$: See \eqref{eq:v2f}.
\item $N(\cdot)$: $N(\cdot)$ outputs the neighborhood set of its input.
\item $\mbf{m}_{f\ra v}^{(k)}$: $\l\{ m_{f\ra v}^{(k)} \l(\mbf{z}_{r(v)}\r) \r\}_{\mbf{z}_{r(v)}}$.
\item $\mbf{n}_{v\ra f}^{(k)}$: $\l\{ n_{v\ra f}^{(k)} \l(\mbf{z}_{r(v)}\r)\r\}_{\mbf{z}_{r(v)}}$.
\item $q_{r(v)}^{(k)}\l( \mbf{z}_{r(v)} \r)$: See \eqref{eq:qv}.
\item $ T_{max}$: Maximum finite time for either process to reach a given node. Introduced in Assumption \ref{assump:tmax}.
\item $\mcl{T}' $: $ \{0,1,\dots, T_{max}, \infty \}$. See Assumption \ref{assump:tmax}.
\item $\pr'$: Probability measure for the (modified) diffusion-model of Assumption \ref{assump:tmax}.
\item $f'_W$: See \eqref{eq:fw2}.
\item $\gamma'_i $: See \eqref{eq:gammai2}.
\item $ \Gamma'$: See \eqref{eq:biggamma_2}.
\item $f'_i, f'_{ij} $: See \eqref{eq:fiandfij_2}.
\item $G_{\Gamma'}$: Inflated factor-graph representation of $\Gamma'$ given by \eqref{eq:gbiggamma_2}.

\item $ \mcl{I}_1$: See \eqref{eq:i1}.
\item $ \mcl{I}_2$: See \eqref{eq:i2}.
\item $ \mcl{I}_3$: See \eqref{eq:i3}.
\item $\udl{\mcl{T}}_i^I$: Some subset of $\mcl{T}$.
\item $\udl{\mcl{S}}_{i\ra j} $: Some subset of $\mcl{S} $.
\item $\udl{\mcl{W}} $: Some subset of $\mcl{W}$.
\item $V_p$: Set of all protector nodes. See Section \ref{sec:disc:anticontagion}.
\item $T_{delay}^B$: Random Delay in the initiation of process $B$.
\item $f_{T_{delay}^B}$: See \eqref{eq:fh}.
\item $\mcl{T}_{delay}^B$: $\l\{h_{min}, h_{min}+1, \dots, h_{max}, \infty \r\}$.

\item $d_C$: Diameter of connected-component $C$ of $G$.
\end{itemize}

\onecolumn
\section{Proof of Claim \ref{claim:getridofd} (Getting Rid of $\mbf{d}$ in \eqref{eq:px0_2})}\label{app:getridofd}
The proof involves two successive applications of the distributive-law (DL). We have,
\begin{align*}
	&\sum_{\mbf{d} } \prod_{(i,j)\in\vv{E}} 
	f_{\mbf{D}_{i\ra j}} \l(\mbf{d}_{i\ra j}\r) 
	\11\l[ s_{i\ra j}^{A}=\sigma\l(t_{i\ra j}^A; t_j^A \r) \r] \times \11\l[ s_{i\ra j}^{B}=\sigma\l(t_{i\ra j}^B; t_j^B \r) \r]\\
	&=\sum_{\mbf{d} } \prod_{ \substack{ (i,j)\in\vv{E}}} \l( 
	\prod_{I,J} f_{D_{i\ra j}^{I\ra J} } \l(d_{i\ra j}^{I\ra J} \r)
	\r) \11\l[ s_{i\ra j}^{A}=\sigma\l( t_{i\ra j}^A; t_j^A \r) \r]\\
	&\hspace{50pt} \times \11\l[ s_{i\ra j}^{B}=\sigma\l( t_{i\ra j}^B; t_j^B \r) \r]\\
	&\stackrel{\text{DL}}{=}\prod_{ (i,j) \in \vv{E}} \sum_{ \mbf{d}_{i\ra j} }
	f_{D_{i\ra j}^{A\ra \emptyset} } \l(d_{i \ra j}^{A\ra \emptyset} \r)
	f_{D_{i\ra j}^{A\ra \{B\} } } \l(d_{i \ra j}^{A\ra \{B\}} \r)
	\11\l[ s_{i\ra j}^{A}=\sigma\l( t_{i\ra j}^A; t_j^A \r) \r] \\
	&\hspace{50pt} \times 
	f_{D_{i\ra j}^{B\ra \emptyset} } \l(d_{i \ra j}^{B\ra \emptyset} \r)
	f_{D_{i\ra j}^{B\ra \{A\} } } \l(d_{i \ra j}^{B\ra \{A\}} \r)
	\11\l[ s_{i\ra j}^{B}=\sigma\l( t_{i\ra j}^B; t_j^B \r) \r]\\
	&\stackrel{\text{DL}}{=}\prod_{ (i,j) \in \vv{E}}
	\underbrace{
		\sum_{ d_{i\ra j}^{A\ra\emptyset}, d_{i\ra j}^{A\ra\{B\}} }
		f_{D_{i\ra j}^{A\ra \emptyset} } \l(d_{i \ra j}^{A\ra \emptyset} \r)
		f_{D_{i\ra j}^{A\ra \{B\} } } \l(d_{i \ra j}^{A\ra \{B\}} \r)
		\11\l[ s_{i\ra j}^{A}=\sigma\l( t_{i\ra j}^A; t_j^A \r) \r]
	}_{\defeq \psi_{i\ra j}^A \l( t_i^A, s_{i\ra j}^A;t_j^A,t_j^B\r) }\\
	&\hspace{50pt} \underbrace{ 
		\sum_{ d_{i\ra j}^{B\ra\emptyset}, d_{i\ra j}^{B\ra\{A\}} }
		f_{D_{i\ra j}^{B\ra \emptyset} } \l(d_{i \ra j}^{B\ra \emptyset} \r)
		f_{D_{i\ra j}^{B\ra \{A\} } } \l(d_{i \ra j}^{B\ra \{A\}} \r)
		\11\l[ s_{i\ra j}^{B}=\sigma\l( t_{i\ra j}^B; t_j^B \r) \r]
	}_{\defeq \psi_{i\ra j}^B \l( t_i^B, s_{i\ra j}^B;t_j^B,t_j^A\r) }\\
	&=\prod_{(i,j)\in\vv{E}} \underbrace{ \psi_{i\ra j}^A \l( t_i^A, s_{i\ra j}^A;t_j^A,t_j^B\r) \psi_{i\ra j}^B \l( t_i^B, s_{i\ra j}^B;t_j^B,t_j^A\r) }_{\defeq \psi_{i\ra j} \l(\mbf{t}_i, \mbf{s}_{i\ra j}; \mbf{t}_j \r)  }\\
	&=\prod_{\{i,j\}\in E} \underbrace{ \psi_{i\ra j} \l(\mbf{t}_i, \mbf{s}_{i\ra j}; \mbf{t}_j \r) \psi_{j\ra i} \l(\mbf{t}_j, \mbf{s}_{j\ra i}; \mbf{t}_i \r) }_{\defeq \psi_{i \xlrarrow{} j} \l( \mbf{t}_i, \mbf{s}_{i\ra j}; \mbf{t}_j , \mbf{s}_{j\ra i} \r) } \\
	&=\prod_{ \substack{ \{i,j\}\in E }} \psi_{i \xlrarrow{} j}\l( \mbf{t}_i, \mbf{s}_{j\ra i}, \mbf{t}_j, \mbf{s}_{i\ra j}  \r) ,
\end{align*}
where, $\psi_{i\ra j}^I$ is given by,
\begin{align*}
	\psi_{i\ra j}^I(\cdot) &= \sum_{ d_{i\ra j}^{I\ra\emptyset}, d_{i\ra j}^{I\ra\{J\}} }
	f_{D_{i\ra j}^{I\ra \emptyset} } \l(d_{i \ra j}^{I\ra \emptyset} \r)
	f_{D_{i\ra j}^{I\ra \{J\} } } \l(d_{i \ra j}^{I\ra \{J\}} \r) \\
	&\hspace{20pt} \11\l[ s_{i\ra j}^{I}=\sigma\l( t_i^I+d_{i\ra j}^{I\ra \emptyset}\11[t_i^I<t_j^J] + d_{i\ra j}^{I\ra\{J\}} \11[t_i^I\ge t_j^J]; t_j^I \r) \r]\\
	\begin{split}
		&\labelrel{=}{eqr:claim1:psi} \begin{cases}
			1-\11\l[t_i^I<t_j^I\r] \lambda_{i\ra j}^{I\ra \emptyset} &\text{if } t_j^I<\infty, s_{i\ra j}^{I}=1, t_i^I<t_j^J; \\
			1-\11\l[t_i^I<t_j^I\r] \lambda_{i\ra j}^{I\ra \{J\}} &\text{if } t_j^I<\infty, s_{i\ra j}^{I}=1, t_i^I\ge t_j^J; \\
			\11[t_i^I+1=t_j^I]\lambda_{i\ra j}^{I\ra \emptyset} &\text{if } t_j^I<\infty, s_{i\ra j}^{I}=0, t_i^I<t_j^J;\\
			\11[t_i^I+1=t_j^I]\lambda_{i\ra j}^{I\ra \{J\}} &\text{if } t_j^I<\infty, s_{i\ra j}^{I}=0, t_i^I\ge t_j^J;\\
			1 - \11[t_i^I<\infty]\lambda_{i\ra j}^{I\ra \emptyset} &\text{if } t_j^I=\infty, s_{i\ra j}^{I}=0, t_i^I<t_j^J;\\
			1 - \11[t_i^I<\infty]\lambda_{i\ra j}^{I\ra \{J\}} &\text{if } t_j^I=\infty, s_{i\ra j}^{I}=0, t_i^I\ge t_j^J;\\
			0 &\text{otherwise}.
		\end{cases}
	\end{split}
\end{align*}
Here, we recall that $t_{i\ra j}^I$ and $\sigma(\cdot, *)$ are shorthands for $t_i^I + d_{i\ra j}^{I\ra \emptyset}\11[t_i^I<t_j^J] +
d_{i\ra j}^{I\ra \{J\}}\11[t_i^I\ge t_j^J]$
and $\sign\l(\cdot-*\r)$ respectively.

\newpage
\section{Message Update Rules for Algorithm \ref{alg:sumproductbp} (Loopy BP on $G_{\Gamma}$)}\label{app:messageupdaterules-1}

\subsection{Efficient Update Rule of $\mbf{m}_{f_i\ra v_{j\ra i}}$}
Recall that $f_i$ represents the function $\zeta_i \gamma_i \psi_{\ra i}$ (see \eqref{eq:fiandfij}) with $\mbf{z}_{r(f_i)}=\l(w,\mbf{x}^{(0)}, \mbf{t}_i, \l\{ \mbf{s}_{k\ra i} \r\}_{k\in\partial{i}} \r)$ and $v_{j\ra i}$ represents the variable-node with $\mbf{z}_{r(v_{j\ra i})} = (w, \mbf{x}^{(0)}, \mbf{t}_i, \mbf{s}_{j\ra i})$. Referring to \figurename \ref{fig:messages}, for any realization $\mbf{z}_{v_{j\ra i}}=(w,\mbf{x}^{(0)},\mbf{t}_i, \mbf{s}_{j\ra i} )$, we have
\begin{align*}
	&m_{f_i \ra v_{j\ra i}}\l(w,\mbf{x}^{(0)}, \mbf{t}_i,\mbf{s}_{j\ra i} \r) \propto 
	\sum_{ \{\mbf{s}_{k\ra i}\}_{k\in\partial{i}\setminus j } } f_i \l(w,\mbf{x}^{(0)}, \mbf{t}_i, \{\mbf{s}_{k\ra i}\}_{k\in\partial{i}} \r) \prod_{k \in\partial{i}\setminus j } m_{f_{ik} \ra v_{k\ra i} } \l(w,\mbf{x}^{(0)}, \mbf{t}_i, \mbf{s}_{k\ra i}  \r).\numberthis\label{eq:fitovji}
\end{align*}
An important quantity that will appear repeatedly in the computation of \eqref{eq:fitovji} is the marginalization of $m_{f_{ik} \ra v_{k\ra i}}$ over a restricted domain of $\mbf{s}_{k\ra i}$ for each $(w, \mbf{x}^{(0)}, \mbf{t}_i)$. We denote this quantity by the symbol $\Phi_{k\ra i}\l(w,\mbf{x}^{(0)},\mbf{t}_i \big| \cdot \r)$ where $\cdot$ specifies the restriction (if any) on the domain of $\mbf{s}_{k\ra i}$. For example, $\Phi_{k\ra i}\l(w,\mbf{x}^{(0)}, \mbf{t}_i \big| s_{k\ra i}^{A}=1 \r)$ represents the marginalization of $m_{f_{ik} \ra v_{k\ra i}} \l(w,\mbf{x}^{(0)}, \mbf{t}_i, \mbf{s}_{k\ra i} \r)$ over $\mbf{s}_{k\ra i}$ with $s_{k\ra i}^{A}$ restricted to value 1. Similarly, $\Phi_{k\ra i}\l(w,\mbf{x}^{(0)}, \mbf{t}_i \big| \r)$ represents the marginalization of $m_{f_{ik} \ra v_{k\ra i}} \l(w,\mbf{x}^{(0)}, \mbf{t}_i, \mbf{s}_{k\ra i} \r)$ over all realizations of $\mbf{s}_{k\ra i}$.

We have divided the computation into different cases; all cases use the definition of $\psi_{\ra i}$ (see \eqref{eq:psiii2}) followed by an application of the distributive-law (DL).

\vspace{5pt}
\noindent \udl{Case 1A - $0=t_i^A=t_i^B$:} Here, 
\begin{align*}
	&m_{f_i \ra v_{j\ra i}} \l(w,\mbf{x}^{(0)}, \mbf{t}_i,\mbf{s}_{j\ra i} \r) \propto \zeta_i(x_i^{(0)}, \mbf{t}_i ) \udl{\gamma}_i (w,\mbf{t}_i; \widetilde{x}_i) 
	\sum_{ \{\mbf{s}_{k\ra i}\}_{k\in\partial{i}\setminus j } } \prod_{k\in\partial{i}\setminus j } m_{f_{ik} \ra v_{k\ra i} } \l(w,\mbf{x}^{(0)}, \mbf{t}_i, \mbf{s}_{k\ra i}  \r)\\
	&\stackrel{\text{DL}}{=} \zeta_i(x_i^{(0)}, \mbf{t}_i ) \udl{\gamma}_i (w,\mbf{t}_i; \widetilde{x}_i)
	\prod_{k\in\partial{i}\setminus j } 
	\Phi_{k\ra i} \l(w, \mbf{x}^{(0)}, \mbf{t}_i | \r).\numberthis\label{eq:alg1:1:case1a}
\end{align*}

\noindent \udl{Case 2A - $0=t_i^A<t_i^B$:} Here, 
\begin{align*}
	&m_{f_i \ra v_{j\ra i}}\l(w,\mbf{x}^{(0)}, \mbf{t}_i,\mbf{s}_{j\ra i} \r)\\
	&\propto \zeta_i(x_i^{(0)}, \mbf{t}_i ) \udl{\gamma}_i (w,\mbf{t}_i; \widetilde{x}_i) 
	\sum_{ \{\mbf{s}_{k\ra i}\}_{k\in\partial{i}\setminus j } }  
	\l(1 - \prod_{k\in\partial{i}} \11[s_{k\ra i}^{B}=1] \r)
	\prod_{k\in\partial{i}\setminus j } m_{f_{ik}\ra v_{k\ra i} } \l(w,\mbf{x}^{(0)}, \mbf{t}_i, \mbf{s}_{k\ra i}  \r)\\
	&\stackrel{\text{DL}}{=} \zeta_i(x_i^{(0)}, \mbf{t}_i ) \udl{\gamma}_i ( w,\mbf{t}_i; \widetilde{x}_i) 
	\l[ 
	\prod_{k\in \partial{i} \setminus j} 
	\Phi_{k\ra i} \l(w, \mbf{x}^{(0)}, \mbf{t}_i | \r) - \11[s_{j\ra i}^{B}=1]
	\prod_{k\in\partial{i}\setminus j} 
	\Phi_{k\ra i} \l(w, \mbf{x}^{(0)}, \mbf{t}_i | s_{k\ra i}^{B}= 1\r) \r]. \numberthis\label{eq:alg1:1:case2a}
\end{align*}

\noindent \udl{Case 2B - $0=t_i^B<t_i^A$:} By symmetry with case 2A, we have 
\begin{align*}
	&m_{f_i \ra v_{j\ra i}}\l(w,\mbf{x}^{(0)}, \mbf{t}_i,\mbf{s}_{j\ra i} \r)\\
	&\propto \zeta_i(x_i^{(0)}, \mbf{t}_i ) \udl{\gamma}_i ( w,\mbf{t}_i; \widetilde{x}_i) 
	\l[ 
	\prod_{k\in \partial{i} \setminus j} 
	\Phi_{k\ra i} \l(w, \mbf{x}^{(0)}, \mbf{t}_i | \r) - \11[s_{j\ra i}^{A}=1]
	\prod_{k\in\partial{i}\setminus j} 
	\Phi_{k\ra i} \l(w, \mbf{x}^{(0)}, \mbf{t}_i | s_{k\ra i}^{A}= 1\r) \r]. \numberthis\label{eq:alg1:1:case2b}
\end{align*}

\noindent \udl{Case 3 - $0<t_i^A, 0<t_i^B$:} Here,
\begin{align*}
	&m_{f_i \ra v_{j\ra i}} \l(
	w,\mbf{x}^{(0)}, \mbf{t}_i,\mbf{s}_{j\ra i} 
	\r) \\
	&\hspace{10pt} \propto 
	\zeta_i( x_i^{(0)}, \mbf{t}_i ) \udl{\gamma}_i( w,\mbf{t}_i; \widetilde{x}_i) 
	\sum_{ \{\mbf{s}_{k\ra i}\}_{k\in\partial{i}\setminus j } }  
	\l(
	1 - \prod_{k\in\partial{i}} \11[s_{k\ra i}^{A}=1] \r) 
	\l(
	1 - \prod_{k\in\partial{i}} \11[s_{k\ra i}^{B}=1] \r)\\
	&\hspace{150pt} \times \prod_{k\in\partial{i}\setminus j } m_{f_{ik}\ra v_{k\ra i} } \l(w,\mbf{x}^{(0)}, \mbf{t}_i, \mbf{s}_{k\ra i}  \r)
	\\
	&\hspace{10pt} \stackrel{\text{DL}}{=} \zeta_i(x_i^{(0)}, \mbf{t}_i ) \udl{\gamma}_i (w,\mbf{t}_i; \widetilde{x}_i )
	\l[
	\prod_{k\in \partial{i}\setminus j } 
	\Phi_{k\ra i} \l(w, \mbf{x}^{(0)}, \mbf{t}_i | \r) \r.\\
	&\hspace{120pt} \l. + \11[s_{j\ra i}^{A}=1, s_{j\ra i}^{B}=1] \prod_{ k\in\partial{i}\setminus j} \Phi_{k\ra i} \l(w, \mbf{x}^{(0)}, \mbf{t}_i | s_{k\ra i}^{A}= 1, s_{k\ra i}^{B}=1  \r) \r.\\
	&\hspace{120pt} \l. - \11[ s_{j\ra i}^{A}=1] \prod_{ k\in\partial{i}\setminus j } \Phi_{k\ra i} \l(w, \mbf{x}^{(0)}, \mbf{t}_i | s_{k\ra i}^{A}= 1 \r) \r. \\
	&\hspace{120pt} \l. - \11[ s_{j\ra i}^{B}=1 ] 
	\prod_{ k\in\partial{i}\setminus j } 
	\Phi_{k\ra i} \l(w, \mbf{x}^{(0)}, \mbf{t}_i | s_{k\ra i}^{B}=1 \r) \r].\numberthis\label{eq:alg1:1:case3}
\end{align*}
The per-iteration-time complexity of computing all such messages 
is $O \l( |{\mcl{W}_1}| |\bm{\mcl{X}}^{(0)}| |V|^2 \sum_{i\in V}|\partial{i}|^2 \r)= O\l(|{\mcl{W}_1}| |\bm{\mcl{X}}^{(0)}| |E| |V|^2  \max_{i\in V}|\partial{i}|\r)$.

\subsection{Update Rule of $\mbf{m}_{f_{ij}\ra v_{j\ra i}}$ and $\mbf{m}_{f_{ij}\ra v_{i \ra j}}$ }
Let $C$ be the connected-component of $G$ to which an edge $\{i,j\}$ belongs. Recall that $f_{ij}$ represents the function $\psi_{i\xlrarrow{} j}$ (see \eqref{eq:fiandfij}) with $\mbf{z}_{r(f_{ij})} = \l(w,\mbf{x}^{(0)}, \mbf{t}_i, \mbf{s}_{j\ra i}, \mbf{t}_j, \mbf{s}_{i\ra j} \r)$. Referring to \figurename \ref{fig:messages}, for any realization $\mbf{z}_{r(v_{j\ra i})} = \l( w, \mbf{x}^{(0)}, \mbf{t}_i, \mbf{s}_{j\ra i} \r)$, we have
\begin{align*}
	&m_{f_{ij}\ra v_{j\ra i} } \l(w, \mbf{x}^{(0)}, \mbf{t}_i, \mbf{s}_{j\ra i} \r) 
	\propto \xi \l( \{i,j\},w,\mbf{x}^{(0)} ; C \r) 
	\sum_{\mbf{t}_j, \mbf{s}_{i\ra j} } \psi_{i\xlrarrow{} j}\l(\mbf{t}_i, \mbf{s}_{j\ra i}, \mbf{t}_j, \mbf{s}_{i\ra j} \r) m_{f_j \ra v_{i\ra j} } \l( w, \mbf{x}^{(0)}, \mbf{t}_j, \mbf{s}_{i\ra j} \r),\numberthis\label{eq:alg1:2:1}
\end{align*}
where $\psi_{i\xlrarrow{} j} $ is given by \eqref{eq:psiij} and $\xi$ is shorthand for 
$$\11[ \{i,j\} \ne \wh{e}^{(C)} ] + \11[ \{i,j\} = \wh{e}^{(C)} ] m_{\udl{f}_{W, \mbf{X}^{(0)}} \ra v_{C}} \l(w, \mbf{x}^{(0)}\r) .$$
Similarly, 
\begin{align*}
	&m_{f_{ij}\ra v_{i\ra j} }\l(w, \mbf{x}^{(0)}, \mbf{t}_j, \mbf{s}_{i\ra j} \r) \propto \xi \l( \{i,j\},w,\mbf{x}^{(0)} ; C \r)
	\sum_{\mbf{t}_i, \mbf{s}_{j\ra i} } \psi_{i\xlrarrow{} j}\l(\mbf{t}_i, \mbf{s}_{j\ra i}, \mbf{t}_j, \mbf{s}_{i\ra j} \r) 
	m_{f_i \ra v_{j\ra i} } \l( w, \mbf{x}^{(0)}, \mbf{t}_i, \mbf{s}_{j\ra i} \r).\numberthis\label{eq:alg1:2:2}
\end{align*}
The per-iteration-time complexity of computing all such messages is $O\l( |{\mcl{W}_1}| |\bm{\mcl{X}}^{(0)}| |E||V|^4 \r)$.

\subsection{Update Rule of $\mbf{m}_{f_{\wh{e}^{(C)}} \ra v_{C} }$ }
Referring to \figurename \ref{fig:messages}, for any realization $\mbf{z}_{r(v_C)} = \l( w, \mbf{x}^{(0)} \r)$, we have
\begin{align*}
	&m_{f_{\wh{e}^{(C)}} \ra v_C }\l( w, \mbf{x}^{(0)} \r) \propto \sum_{ \mbf{t}_i, \mbf{s}_{j\ra i}, \mbf{t}_j, \mbf{s}_{i\ra j} } \psi_{i\xlrarrow{} j} \l( \mbf{t}_i, \mbf{s}_{j\ra i}, \mbf{t}_j, \mbf{s}_{i\ra j} \r) m_{f_i\ra v_{j\ra i} }\l(w,\mbf{x}^{(0)}, \mbf{t}_i, \mbf{s}_{j\ra i} \r)
	m_{f_j\ra v_{i\ra j} }\l(w,\mbf{x}^{(0)}, \mbf{t}_j, \mbf{s}_{i\ra j} \r).\numberthis\label{eq:alg1:4}
\end{align*}
The per-iteration-time complexity of computing this message is $O\l(|{\mcl{W}_1}| |\bm{\mcl{X}}^{(0)}| |\mcl{C}| |V|^4 \r) $.

\subsection{Update Rule of $\mbf{m}_{f_{W, \mbf{X}^{(0)}} \ra v_{C} }$}
Referring to \figurename \ref{fig:messages2}, for any realization $\mbf{z}_{r(v_{C})} = \l( w, \mbf{x}^{(0)} \r)$, we have
\begin{align*}
	&m_{ \udl{f}_{ W, \mbf{X}^{(0)} } \ra v_{C} } \l( w, \mbf{x}^{(0)} \r) 
	\propto 
	\udl{f}_W(w)f_{\mbf{X}^{(0)}}\l(\mbf{x}^{(0)}\r) \prod_{ C' \in \mcl{C} \setminus {C}} m_{ f_{\wh{e}^{(C')}} \ra v_{C'} } \l(w, \mbf{x}^{(0)}\r).\numberthis\label{eq:alg2:5}
\end{align*}
The per-iteration-time complexity of computing all such messages is $O\l( |{\mcl{W}_1}| |\bsl{\mcl{X}}^{(0)}| |\mcl{C}| \r)$

\textbf{Per-Iteration-Time Complexity}: In \eqref{eq:alg1:1:case1a} through \eqref{eq:alg1:4}, the per-iteration-time is dominated by messages of the form $\mbf{m}_{f_{ij} \ra v_{j\ra i}}$. Therefore, per-iteration-time complexity of Algorithm \ref{alg:sumproductbp} is is $O\l(|{\mcl{W}_1}| |\bsl{\mcl{X}}^{(0)}| |E| |V|^4 \r)$.


\newpage
\section{Message Update Rules for $G_{\Gamma'}$ (Loopy BP on $G_{\Gamma'}$)}\label{app:messageupdaterules2}

\subsection{Efficient Update Rule of $\mbf{m}_{f'_i\ra v'_{j\ra i}}$}
Recall that $f'_i$ represents the function $f_{X_i^{(0)}} \zeta_i \gamma'_i \psi_{\ra i}$ (see \eqref{eq:fiandfij_2}) with $\mbf{z}_{r(f_i)} = \l( w,x_i^{(0)}, \mbf{t}_i, \l\{ \mbf{s}_{k\ra i} \r\}_{k\in\partial{i}} \r)$ and $v'_{j\ra i}$ represents the variable-node containing the variables $(w, \mbf{t}_i, \mbf{s}_{j\ra i})$. Referring to \figurename \ref{fig:messages2}, for any realization $\mbf{z}_{r(v'_{j\ra i})} = (w,\mbf{t}_i, \mbf{s}_{j\ra i} )$, we have
\begin{align*}
	&m_{f'_i \ra v'_{j\ra i}}\l(w, \mbf{t}_i,\mbf{s}_{j\ra i} \r) \propto 
	\sum_{ x_i^{(0)}, \{\mbf{s}_{k\ra i}\}_{k\in\partial{i}\setminus j } } f'_i \l(w,x_i^{(0)}, \mbf{t}_i, \{\mbf{s}_{k\ra i}\}_{k\in\partial{i}} \r) \prod_{k\in\partial{i}\setminus j } m_{f'_{ik} \ra v'_{k\ra i} } \l(w, \mbf{t}_i, \mbf{s}_{k\ra i}  \r).\numberthis\label{eq:fitovji2}
\end{align*}
An important quantity that will appear repeatedly in the computation of \eqref{eq:fitovji2} is the marginalization of $m_{f'_{ik} \ra v'_{k\ra i}}$ over a restricted domain of $\mbf{s}_{k\ra i}$ for each $(w, \mbf{t}_i)$. We denote this quantity by the symbol $\Phi_{k\ra i}\l(w, \mbf{t}_i \big| \cdot \r)$ where $\cdot$ specifies the restriction (if any) on the domain of $\mbf{s}_{k\ra i}$. For example, $\Phi_{k\ra i}\l(w, \mbf{t}_i \big| s_{k\ra i}^{A}=1 \r)$ represents the marginalization of $m_{f'_{ik} \ra v'_{k\ra i}} \l(w, \mbf{t}_i, \mbf{s}_{k\ra i} \r)$ over $\mbf{s}_{k\ra i}$ with $s_{k\ra i}^{A}$ restricted to value 1. Similarly, $\Phi_{k\ra i}\l(w, \mbf{t}_i \big| \r)$ represents the marginalization of $m_{f'_{ik} \ra v'_{k\ra i}} \l(w, \mbf{t}_i, \mbf{s}_{k\ra i} \r)$ over all realizations of $\mbf{s}_{k\ra i}$.

We divide the computation into different cases as follows. All cases use the definition of $\psi_{\ra i}$ (see \eqref{eq:psiii2}) followed by an application of the distributive law.

\vspace{5pt}
\noindent \udl{Case 1 - $0=t_i^A=t_i^B$:} Here, 
\begin{align*}
	&m_{f'_i \ra v'_{j\ra i}} \l(w, \mbf{t}_i,\mbf{s}_{j\ra i} \r) \propto \gamma'_i \l(w,\mbf{t}_i; \wt{x}_i\r) \l( \sum_{x_i^{(0)}} f_{X_i^{(0)}} (x_i^{(0)} ) \zeta_i( x_i^{(0)}, \mbf{t}_i ) \r)  
	\sum_{ \{\mbf{s}_{k\ra i}\}_{k\in\partial{i}\setminus j } } \prod_{k\in\partial{i}\setminus j } m_{f'_{ik} \ra v'_{k\ra i} } \l(w, \mbf{t}_i, \mbf{s}_{k\ra i}  \r)\\
	&\stackrel{\text{DL}}{=} \gamma'_i (w,\mbf{t}_i; \widetilde{x}_i) \l( \sum_{x_i^{(0)}} f_{X_i^{(0)}} (x_i^{(0)} ) \zeta_i( x_i^{(0)}, \mbf{t}_i ) \r)
	\prod_{k\in\partial{i}\setminus j } 
	\Phi_{k\ra i} \l(w, \mbf{t}_i | \r)
	\numberthis\label{eq:alg2:1:case1}
\end{align*}

\noindent \udl{Case 2A - $0=t_i^A<t_i^B$:} Here, 
\begin{align*}
	&m_{f'_i \ra v'_{j\ra i}}\l(w, \mbf{t}_i,\mbf{s}_{j\ra i} \r) \\
	&\propto 
	\gamma'_i (w,\mbf{t}_i; \widetilde{x}_i) \l( \sum_{x_i^{(0)}} f_{X_i^{(0)}} (x_i^{(0)} ) \zeta_i( x_i^{(0)}, \mbf{t}_i ) \r)
	\sum_{ \{\mbf{s}_{k\ra i}\}_{k\in\partial{i}\setminus j } }  
	\l( 1 - \prod_{k\in\partial{i}} \11[s_{k\ra i}^{B}=1] \r) 
	\prod_{k\in\partial{i}\setminus j } m_{f'_{ik}\ra v'_{k\ra i} } \l(w, \mbf{t}_i, \mbf{s}_{k\ra i}  \r)\\
	&\stackrel{\text{DL}}{=} \gamma'_i ( w,\mbf{t}_i; \widetilde{x}_i) \l( \sum_{x_i^{(0)}} f_{X_i^{(0)}} (x_i^{(0)} ) \zeta_i( x_i^{(0)}, \mbf{t}_i ) \r)
	\l[ 
	\prod_{k\in \partial{i} \setminus j} 
	\Phi_{k\ra i} \l(w, \mbf{t}_i | \r) - \11[s_{j\ra i}^B=1 ] 
	\prod_{k\in\partial{i}\setminus j} 
	\Phi_{k\ra i} \l(w, \mbf{t}_i | s_{k\ra i}^{B}= 1 \r) \r].\numberthis\label{eq:alg2:1:case2a}
\end{align*}

\noindent \udl{Case 2B - $0=t_i^B<t_i^A$:} By symmetry with case 2A, we have 
\begin{align*}
	&m_{f'_i \ra v'_{j\ra i}}\l(w, \mbf{t}_i,\mbf{s}_{j\ra i} \r) \\
	&\propto \gamma'_i ( w,\mbf{t}_i; \widetilde{x}_i) \l( \sum_{x_i^{(0)}} f_{X_i^{(0)}} (x_i^{(0)} ) \zeta_i( x_i^{(0)}, \mbf{t}_i ) \r)
	\l[ 
	\prod_{k\in \partial{i} \setminus j} 
	\Phi_{k\ra i} \l(w, \mbf{t}_i | \r) - \11[s_{j\ra i}^A=1 ] 
	\prod_{k\in\partial{i}\setminus j} 
	\Phi_{k\ra i} \l(w, \mbf{t}_i | s_{k\ra i}^{A}= 1 \r) \r].\numberthis\label{eq:alg2:1:case2b}
\end{align*}

\noindent \udl{Case 3 - $0<t_i^A, 0<t_i^B$:} Here, 
\begin{align*}
	&m_{f'_i \ra v'_{j\ra i}} \l(w, \mbf{t}_i,\mbf{s}_{j\ra i} \r) \\
	&\propto \gamma'_i (w,\mbf{t}_i; \widetilde{x}_i) 
	\l( \sum_{x_i^{(0)}} f_{X_i^{(0)}} (x_i^{(0)} ) \zeta_i( x_i^{(0)}, \mbf{t}_i ) \r)
	\sum_{ \{\mbf{s}_{k\ra i}\}_{k\in\partial{i}\setminus j } }  
	\l(1 - \prod_{k\in\partial{i}} \11[s_{k\ra i}^{A}=1] \r) \l( 1 - \prod_{k\in\partial{i}} \11[s_{k\ra i}^{B}=1] \r)
	\\
	&\hspace{240pt} \times 
	\prod_{k\in\partial{i}\setminus j } m_{f'_{ik}\ra v'_{k\ra i} } \l(w, \mbf{t}_i, \mbf{s}_{k\ra i}  \r)
	\\
	&\stackrel{\text{DL}}{=} \gamma'_i (w,\mbf{t}_i; \wt{x}_i) \l( \sum_{ x_i^{(0)} } f_{ X_i^{(0)} } ( x_i^{(0)} ) 
	\zeta_i( x_i^{(0)}, \mbf{t}_i ) \r)
	\l[
	\prod_{ k\in \partial{i} \setminus j } 
	\Phi_{k\ra i} \l( w, \mbf{t}_i | \r) \r. \\
	&\hspace{30pt} \l. + \11[s_{j\ra i}^{A}=1, s_{j\ra i}^{B}=1] \prod_{ k\in\partial{i} \setminus j } \Phi_{k\ra i} \l( w, \mbf{t}_i | s_{k\ra i}^{A}= 1, s_{k\ra i}^{B}=1  \r)
	\r.\\
	&\hspace{30pt} \l. - \11[ s_{j\ra i}^{A}=1] 
	\prod_{ k\in\partial{i}\setminus j } 
	\Phi_{k\ra i} \l(w, \mbf{t}_i | s_{k\ra i}^{A}= 1 \r) \r. \\
	&\hspace{30pt} \l. - \11[ s_{j\ra i}^{B}=1] 
	\prod_{ k\in\partial{i}\setminus j } 
	\Phi_{k\ra i} \l(w, \mbf{t}_i | s_{k\ra i}^{B}=1 \r) \r].\numberthis\label{eq:alg2:1:case3}
\end{align*}
The per-iteration-time complexity of computing all such messages is $O\l(T_{max}^3 \sum_{i\in V} |\partial{i}|^2  \r) = O \l( |E|T_{max}^3 \max_{i\in V}| {\partial{i}}|\r)$.

\subsection{Update Rule of $\mbf{m}_{f'_{ij} \ra v'_{j\ra i}}$ and $\mbf{m}_{f'_{ij} \ra v'_{i \ra j} }$}
Let $C$ denote the connected-component to which an edge $\{i, j\}$ belongs. Recall that $f'_{ij}$ represents the function $\psi_{i\xlrarrow{} j}$ (see \eqref{eq:fiandfij_2}) with $\mbf{z}_{r(f'_{ij})} = \l(w, \mbf{t}_i, \mbf{s}_{j\ra i}, \mbf{t}_j, \mbf{s}_{i\ra j} \r)$. Referring to \figurename \ref{fig:messages2}, for any realization $\mbf{z}_{r(v'_{j\ra i})} = ( w,\mbf{t}_i, \mbf{s}_{j\ra i} )$, we have
\begin{align*}
	&m_{f'_{ij}\ra v'_{j\ra i} } \l(w, \mbf{t}_i, \mbf{s}_{j\ra i} \r) 
	\propto \xi' \l( \{i,j\}, w; C \r) 
	\sum_{\mbf{t}_j, \mbf{s}_{i\ra j} } \psi_{i\xlrarrow{} j}\l(\mbf{t}_i, \mbf{s}_{j\ra i}, \mbf{t}_j, \mbf{s}_{i\ra j} \r) m_{f'_j \ra v'_{i\ra j} } \l( w, \mbf{t}_j, \mbf{s}_{i\ra j} \r),\numberthis\label{eq:alg2:2:1}
\end{align*}
where $\psi_{i\xlrarrow{} j} $ is given by \eqref{eq:psiijji} and $\xi'$ is shorthand for 
$$\11[\{i,j\} \ne \wh{e}^{(C)} ]+\11[\{i,j\}= \wh{e}^{(C)}] m_{f'_W \ra v'_{C}}(w).$$ 
Similarly,
\begin{align*}
	&m_{f'_{ij}\ra v'_{i\ra j} }\l(w, \mbf{t}_j, \mbf{s}_{i\ra j} \r) \propto \xi' \l( \{i,j\},w ; C \r)
	\sum_{\mbf{t}_i, \mbf{s}_{j\ra i} } \psi_{i\xlrarrow{} j}\l(\mbf{t}_i, \mbf{s}_{j\ra i}, \mbf{t}_j, \mbf{s}_{i\ra j} \r) 
	m_{f'_i \ra v'_{j\ra i} } \l( w, \mbf{t}_i, \mbf{s}_{j\ra i} \r).\numberthis\label{eq:alg2:2:2}
\end{align*}
The per-iteration-time complexity of computing all such messages is $O \l( |E|T_{max}^5 \r) $.

\subsection{Update Rule of $\mbf{m}_{f'_{i} \ra v'_i }$ }
Referring to \figurename \ref{fig:messages2}, for any realization $\mbf{z}_{r(v'_i)} = x_i^{(0)}$, we have
\begin{align*}
	&m_{f'_i \ra v'_i }\l( x_i^{(0)} \r) 
	\propto 
	\sum_{ w,\mbf{t}_i, \{\mbf{s}_{k\ra i}\}_{k\in\partial{i} } } f'_i \l(w,x_i^{(0)}, \mbf{t}_i, \{\mbf{s}_{k\ra i}\}_{k\in\partial{i}} \r) \prod_{k\in\partial{i} } m_{f'_{ik} \ra v'_{k\ra i} } \l(w, \mbf{t}_i, \mbf{s}_{k\ra i}  \r)\\
	&\hspace{10pt} = f_{X_i^{(0)}}\l(x_i^{(0)}\r) \sum_{\mbf{t}_i} \zeta_i( x_i^{(0)}, \mbf{t}_i ) \sum_{w} \gamma'_i (w,\mbf{t}_i; \widetilde{x}_i) 
	\sum_{ \{ \mbf{s}_{k\ra i} \}_{k\in\partial{i}} } \psi_{\ra i} \l( \mbf{t}_i, \{\mbf{s}_{k\ra i} \}_{k\in\partial{i}} \r) \prod_{k\in\partial{i}} m_{f'_{ik}\ra v'_{k\ra i}} \l(w, \mbf{t}_i, \mbf{s}_{k\ra i} \r)\\
	&\hspace{10pt} = f_{X_i^{(0)}}\l(x_i^{(0)}\r) \l( \Phi_{1}+\Phi_{2A}+\Phi_{2B}+\Phi_{3} \r),
	\numberthis\label{eq:alg2:4}
\end{align*}
where
\begin{align*}
	\Phi_{1} &=\sum_{0=t_i^A=t_i^B}  \zeta_i( x_i^{(0)}, \mbf{t}_i ) \sum_{w} \gamma'_i (w,\mbf{t}_i; \widetilde{x}_i) \prod_{k\in\partial{i}} \Phi_{k\ra i}\l(w, \mbf{t}_i|\r),\\
	\Phi_{2A} &=\sum_{0=t_i^A<t_i^B} 
	\zeta_i( x_i^{(0)}, \mbf{t}_i ) 
	\sum_{w} \gamma'_i (w,\mbf{t}_i; \wt{x}_i)
	\l[ 
	\prod_{k\in \partial{i}} 
	\Phi_{k\ra i} \l(w, \mbf{t}_i | \r) 
	- \prod_{k\in\partial{i}}
	\Phi_{k\ra i} \l(w, \mbf{t}_i | s_{k\ra i}^{B}=1 \r)
	\r],\\
	\Phi_{2B} &=\sum_{0=t_i^B<t_i^A} 
	\zeta_i( x_i^{(0)}, \mbf{t}_i ) 
	\sum_{w} \gamma'_i (w,\mbf{t}_i; \wt{x}_i)
	\l[ 
	\prod_{k\in \partial{i}} 
	\Phi_{k\ra i} \l(w, \mbf{t}_i | \r) 
	- \prod_{k\in\partial{i}}
	\Phi_{k\ra i} \l(w, \mbf{t}_i | s_{k\ra i}^{A}=1 \r)
	\r],\\
	\Phi_{3} &=\sum_{0<t_i^A, 0<t_i^B} \zeta_i( x_i^{(0)}, \mbf{t}_i ) \sum_{w} \gamma'_i (w,\mbf{t}_i; \wt{x}_i)
	\l[
	\prod_{k\in \partial{i}} 
	\Phi_{k\ra i} \l(w, \mbf{t}_i | \r)
	+ \prod_{ k\in\partial{i} } \Phi_{k\ra i} \l( w, \mbf{t}_i | s_{k\ra i}^{A}= 1, s_{k\ra i}^{B}=1  \r)
	\r.\\
	&\hspace{30pt} \l. -
	\prod_{ k\in\partial{i} } 
	\Phi_{k\ra i} \l(w, \mbf{t}_i | s_{k\ra i}^{A}= 1 \r)
	- \prod_{ k\in\partial{i} } 
	\Phi_{k\ra i} \l(w, \mbf{t}_i | s_{k\ra i}^{B}=1 \r)
	\r].
\end{align*}
The per-iteration-time complexity of computing all such messages is $O\l( T_{max}^3 \sum_{i\in V} |\partial{i}|^2 \r) = O\l( |E|T_{max}^3 \max_{i\in V}|\partial{i}| \r) $.

\subsection{Update Rule of $\mbf{m}_{f'_{\wh{e}^{(C)}} \ra v'_C }$}
Let $\wh{e}^{(C)} = \{i, j\}$. Referring to \figurename \ref{fig:messages2}, for any realization $\mbf{z}_{r(v'_C)} = w$, we have
\begin{align*}
	&m_{f'_{\wh{e}^{(C)}} \ra v'_C }\l( w \r) \propto \sum_{\mbf{t}_i, \mbf{s}_{j\ra i}, \mbf{t}_j, \mbf{s}_{i\ra j} } \psi_{i\xlrarrow{} j} \l( \mbf{t}_i, \mbf{s}_{j\ra i}, \mbf{t}_j, \mbf{s}_{i\ra j} \r) 
	m_{f'_i\ra v'_{j\ra i} } \l( w, \mbf{t}_i, \mbf{s}_{j\ra i} \r) m_{f'_j\ra v'_{i\ra j} } \l( w, \mbf{t}_j, \mbf{s}_{i\ra j} \r).\numberthis\label{eq:alg2:3}
\end{align*}
The per-iteration-time complexity of computing all such messages is $O \l( |\mcl{C}|T_{max}^5 \r) $.

\subsection{Update Rule of $\mbf{m}_{f'_W \ra w^{(C)} }$ }
Referring to \figurename \ref{fig:messages2}, for any realization $\mbf{z}_{r(v'_C)} = w$, we have
\begin{align*}
	&m_{ {f}'_{ W } \ra v'_{C} } \l( w \r) 
	\propto 
	{f}'_W(w) \prod_{ C' \in \mcl{C} \setminus {C}} m_{ {f}'_{\wh{e}^{(C')}} \ra v'_{C'} } (w).\numberthis\label{eq:alg2:5}
\end{align*}
The per-iteration-time complexity of computing all such messages is $O\l( |\mcl{C}| T_{max} \r)$.

\textbf{Per-Iteration-Time Complexity}: In \eqref{eq:alg2:1:case1} through \eqref{eq:alg2:4}, the per-iteration-time is dominated by messages of the form $\mbf{m}_{f'_{ij} \ra v'_{j\ra i}} $ or $\mbf{m}_{f'_{i} \ra v'_i }$. Therefore, per-iteration-time complexity of Algorithm \ref{alg:sumproductbp2} is $O \l( |E|T_{max}^3 \l( T_{max}^2 \vee \max_{i\in V} |\partial{i}|  \r)\r)$.



\end{document}